%% file: main.tex
\def\eprintversion{1}
\def\anonymous{0}
\def\visiblecomments{0}
\ifnum\eprintversion=0
	\documentclass[runningheads,orivec,envcountsame]{llncs}
\else
	\documentclass[orivec,envcountsame,envcountsect]{llncs}[11pt]
	\usepackage{fullpage}
    \usepackage{breakcites}
\fi

  \newcommand{\PRF}{\mathsf{PRF}}

  \newcommand{\calR}{\mathcal{R}}

\newcommand{\KG}{\mathsf{Kg}}

\newcommand{\win}{\mathsf{win}}

\newcommand{\snote}[1]{\ifnum\visiblecomments=1 {\color{purple} [{\bf Stefano:} #1]} \fi}
\newcommand{\cnote}[1]{\ifnum\visiblecomments=1 {\color{blue} [{\bf Chenzhi:} #1]} \fi}
\newcommand{\chnote}[1]{\ifnum\visiblecomments=1 {\color{red} [{\bf Champ:} #1]} \fi}
\newcommand{\dnote}[1]{\ifnum\visiblecomments=1 {\color{orange} [{\bf Dennis:} #1]} \fi}

\usepackage[dvips]{graphicx}
\usepackage{enumitem}
\usepackage{amsmath}
\usepackage{enumitem}
\usepackage{amsfonts}
\usepackage{amssymb}
\usepackage{bbm}
\usepackage{color}
\usepackage{multicol}
\usepackage{mathabx}
\usepackage{calc}
\usepackage[colorlinks=true,
			linkcolor=blue,
			citecolor=blue,
			pdftitle={},
    		pdfauthor={}]{hyperref}
\usepackage{colonequals}  
\usepackage[noabbrev]{cleveref}
\usepackage{xcolor}
\usepackage{wrapfig}
\usepackage[lambda,n,operators]{cryptocode}
\input{macros.tex}

\input{macros-champ.tex}

\usepackage{thm-restate}
\usepackage{stmaryrd}
\usepackage{dashbox}
\usepackage{float}
\usepackage{tikz}
\usepackage{orcidlink}
\usetikzlibrary{shapes.geometric, chains, calc, fit, matrix}
\tikzstyle{system}=[rectangle,draw,fill=lightgray,minimum height=0.8cm,minimum
            width=0.8cm,thick]
\tikzstyle{BC}=[system]
\tikzstyle{resource}=[system]
\tikzstyle{RO}=[resource, minimum width=1cm]
\tikzstyle{protocol}=[circle, inner sep=0.7mm, draw]
\tikzstyle{simulator}=[circle, inner sep=0.7mm, draw]
\tikzstyle{memory}=[resource]
\tikzstyle{distinguisher}=[resource,fill=white,minimum width=3.5cm,
minimum height=1.2cm]
\tikzstyle{link}=[]

\allowdisplaybreaks[1]

\begin{document}
\ifnum\eprintversion=0
\titlerunning{}
\else

\fi
\ifnum\eprintversion=1
	\title{Tight Security for BBS Signatures~\thanks{A preliminary version of this paper appears in the proceedings of EUROCRYPT 2026. This is the full version. © IACR 2026.}}
\else
	\title{Tight Security for BBS Signatures~\thanks{The full version of this work can be found at~\cite{EPRINT:ChaHofTes25}. }}
\fi

\ifnum\anonymous=0
	 \author{Rutchathon~Chairattana-Apirom~\inst{1}\thanks{A portion of this work is done while visiting ETH Zurich.}\orcidlink{0009-0006-1990-1329} 
	 	\and Dennis~Hofheinz~\inst{2} 
	 	\and Stefano~Tessaro~\inst{1}\orcidlink{0000-0002-3751-8546}}
	\institute{Paul G.\ Allen School of Computer Science \& Engineering\\University of
	Washington, Seattle, US\\
	\email{\{rchairat,tessaro\}@cs.washington.edu} \and
	Department of Computer Science\\
	ETH Zurich, Switzerland\\
	\email{hofheinz@inf.ethz.ch}
	}
\else
	\author{}
	\institute{}
\fi

\maketitle

\input{tex_files/abstract}
\input{tex_files/intro}

\input{tex_files/overview}

\input{tex_files/prelims}
\input{tex_files/tight-security-statements}
\input{tex_files/polynomial-properties}
\input{tex_files/full-approach}
\input{tex_files/lower-bound}

\ifnum\anonymous=0
\section*{Acknowledgements}
The authors wish to thank Kenneth Paterson for discussions at an earlier stage of the work and discussions on the running time of our reduction. Chairattana-Apirom and Tessaro's research was partially supported by NSF grants CNS-2026774, CNS-2154174, CNS-2426905, a gift from Microsoft, and a Stellar Development Foundation Academic Research Award. 
\fi

\ifnum\eprintversion=0
\bibliographystyle{splncs04}
\bibliography{references.bib}

\else
\bibliographystyle{alpha}
\bibliography{arxiv.bib}
\fi

\newpage
\appendix
\ifnum\eprintversion=0
\else
\input{tex_files/proof-any-length}
\fi

\end{document}

%% file: macros.tex
\definecolor{rvone}{RGB}{139,34,82}
\definecolor{rvtwo}{rgb}{0.00,0.62,0.22}

\renewcommand{\epsilon}{\varepsilon}

\providecommand{\poly}{{\normalfont \small \textsf{poly}}}

\renewcommand{\Pr}{\mathsf{Pr}}

\providecommand{\Adv}{\mathsf{Adv}}

\ifnum\eprintversion=1
\providecommand{\eAnd}{\; \cap\; }
\providecommand{\eOr}{\; \cup\; }
\else
\providecommand{\eAnd}{ \cap }
\providecommand{\eOr}{ \cup }
\fi

\providecommand{\secpar}{n}

\providecommand{\rsample}{{\:{\leftarrow{\hspace*{-3pt}\raisebox{.75pt}{$\scriptscriptstyle\$$}}}\:}}

\providecommand{\advA}{\bA}

\makeatletter
\providecommand\parag{%
   \@startsection{paragraph}{4}{\z@}%
       {-3\p@ \@plus -1\p@ \@minus -1\p@}%
       {-0.2em \@plus -0.12em \@minus -0.05em}%
       {\normalfont\normalsize\scshape}}
\providecommand{\heading}[1]{\parag{\underline{#1}}}
\makeatother

\providecommand{\cancel}[1]{}

\providecommand{\bits}{\{0,1\}}

\providecommand{\sk}{\mathsf{sk}}

\providecommand{\getsr}{\rsample}

\renewcommand{\advA}{\mathcal{A}}

\providecommand{\true}{\texttt{true}}
\providecommand{\false}{\texttt{false}}

\providecommand{\advB}{\mathcal{B}}

\providecommand{\SD}{\mathsf{SD}}

\providecommand{\advC}{\mathcal{C}}

\providecommand{\msgSp}{\mathsf{M}}

\providecommand{\Coll}{\mathsf{Coll}}

\providecommand{\Group}{\mathbb{G}}
\providecommand{\bbZ}{\mathbb{Z}}

\definecolor{highlight-gray}{gray}{0.90}
\definecolor{highlight-yellow}{cmyk}{0,0,0.90,0}
\definecolor{highlight-cyan}{cmyk}{0.40,0,0,0}
\definecolor{highlight-magenta}{cmyk}{0,0.90,0,0}

\providecommand{\gamechange}[2][highlight-gray]{{\setlength{\fboxsep}{0pt}\colorbox{#1}{\ifmmode$\displaystyle#2$\else#2\fi}}}
\providecommand{\gamechangey}[2][highlight-yellow]{{\setlength{\fboxsep}{0pt}\colorbox{#1}{\ifmmode$\displaystyle#2$\else#2\fi}}}
\providecommand{\gamechangec}[2][highlight-cyan]{{\setlength{\fboxsep}{0pt}\colorbox{#1}{\ifmmode$\displaystyle#2$\else#2\fi}}}
\providecommand{\gamechangem}[2][highlight-magenta]{{\setlength{\fboxsep}{0pt}\colorbox{#1}{\ifmmode$\displaystyle#2$\else#2\fi}}}

\providecommand{\gamechanger}[2][highlight-green]{{\setlength{\fboxsep}{0pt}\colorbox{#1}{\ifmmode$\displaystyle#2$\else#2\fi}}}
\colorlet{highlight-green}{green!40!yellow}

\providecommand{\calI}{\mathcal{I}}
\providecommand{\calU}{\mathcal{U}}

\providecommandx{\Imid}{\calI_{\mathrm{mid}}}

\providecommand{\closed}{\mathtt{completed}}
\providecommand{\str}{\mathrm{str}}

\providecommand{\calF}{\mathcal{F}}
\providecommand{\calT}{\mathcal{T}}

\providecommand{\param}{\mathsf{par}}

\providecommand{\Setup}[1]{#1.\mathsf{Setup}}
\renewcommand{\KG}[1]{#1.\mathsf{KG}}

\providecommand{\Ver}[1]{#1.\mathsf{Ver}}

\providecommand{\SignS}{\mathsf{SS}}

\providecommand{\DL}{\mathrm{DL}}

\providecommand{\dl}{\mathrm{dl}}

\providecommand{\oraF}{\textsc{F}}

\providecommand{\oraS}{\textsc{S}}

\providecommand{\oraCh}{\textsc{Chal}}

\providecommand\defaultboxsep{2pt}
\providecommand{\dgamebox}[1]{\dashbox{\shortstack[l]{#1}}}

\providecommand{\fhgamebox}[1]{\setlength\fboxsep{0pt}\colorbox{gray!20}{\setlength\fboxsep{\defaultboxsep}\framebox{\shortstack[l]{#1}}}\setlength\fboxsep{\defaultboxsep}}

%% file: macros-champ.tex
\let\underbrace\LaTeXunderbrace

\pcsetargs{codesize=\scriptsize,subcodesize=\scriptsize}

\spnewtheorem{observation}{Observation}{\bfseries}{\rmfamily}
\Crefname{observation}{Observation}{Observations}
\Crefname{conjecture}{Conjecture}{Conjectures}

\providecommand{\secpar}{\lambda}

\providecommand{\getsr}{\leftarrow_{\$}}

\providecommand{\poly}{\ensuremath{\mathsf{poly}}}

\providecommand{\Good}{\ensuremath{\mathsf{Good}}}
\providecommand{\GoodE}{\ensuremath{\mathsf{GoodE}}}

\providecommand{\Expect}{\ensuremath{\mathbb{E}}}

\providecommand{\Zp}{\mathbb{Z}_p}

\providecommand{\Group}{\mathbb{G}}

\providecommand{\GOne}{\mathbb{G}_1}
\providecommand{\GTwo}{\mathbb{G}_2}
\providecommand{\GTarg}{\mathbb{G}_T}
\providecommand{\pairingF}{\mathsf{e}}
\providecommand{\pairingOp}[2]{\mathsf{e}\left(#1, #2\right)}
\providecommand{\pGroupDescript}{(p, \GOne, \GTwo, \GTarg, \pairingF)}

\providecommand{\Setup}{\ensuremath{\mathsf{Setup}}}
\providecommand{\GroupGen}{\mathsf{GGen}}

\providecommand{\sk}{\ensuremath{\mathsf{sk}}}

\providecommand{\vk}{\ensuremath{\mathsf{vk}}}

\providecommand{\param}{\ensuremath{{\mathsf{par}}}}

\providecommand{\hardprobfont}[1]{\texorpdfstring{\mathrm{#1}}{#1}}

\providecommand{\DL}{\hardprobfont{DL}}

\providecommand{\dlog}{\ensuremath{\mathrm{dlog}}}

\providecommand{\inp}{\ensuremath{\mathsf{inp}}}

\providecommand{\Sign}[1]{#1.\mathsf{Sign}}

\providecommand{\BBS}{\mathsf{BBS}}

\providecommand{\SUF}{\mathsf{SUF}}
\providecommand{\suf}{\mathrm{suf}}
\providecommand{\qSDH}{\qnum\text{-}\hardprobfont{SDH}}
\providecommand{\qDL}{\qnum\text{-}\hardprobfont{DL}}
\providecommand{\qsdh}{\qnum\text{-}\hardprobfont{sdh}}
\providecommand{\sdh}{\hardprobfont{sdh}}
\providecommand{\SDH}{\hardprobfont{SDH}}

\providecommand{\DBBS}{\mathsf{DBBS}}

\providecommand{\qnum}{\mathsf{q}}
\providecommand{\varX}{\mathsf{X}}
\providecommand{\varY}{\mathsf{Y}}

\providecommand{\Forge}{\mathsf{Forge}}

\providecommand{\Root}{\mathsf{root}}

\providecommand{\BadRep}{\mathsf{BadRep}}
\providecommand{\BadIdx}{\mathsf{BadIdx}}
\providecommand{\BadQ}{\mathsf{BadQ}}

\providecommand{\interprob}{\mathsf{p}}
\providecommand{\metaR}{\mathcal{M}}
\providecommand{\Rewind}{\mathrm{Rewind}}

\providecommand{\Error}{\mathsf{Er}}

\providecommand{\actual}[1]{\underline{#1}}
\providecommand{\actualMsg}{\actual{m}}
\providecommand{\actualMu}{\actual{\mu}}
\providecommand{\actualE}{\actual{e}}
\providecommand{\actualS}{\actual{S}}
\providecommand{\actualX}{\actual{x}}
\providecommand{\actualA}{\actual{a}}
\providecommand{\actualB}{\actual{b}}
\providecommand{\actualC}{\actual{c}}

\providecommand{\runscompleted}{\mathsf{cmpt}\text{-}\mathsf{ord}}
\providecommand{\started}{\mathsf{started}}

\ifnum\eprintversion=0

\else

\fi

%% file: tex_files/abstract.tex
\begin{abstract}
    This paper studies the concrete security of BBS signatures (Boneh, Boyen, Shacham, CRYPTO '04; Camenisch and Lysyanskaya, CRYPTO '04), a popular algebraic construction of digital signatures which underlies practical privacy-preserving authentication systems and is undergoing standardization by the W3C and IRTF.

    Sch\"age (Journal of Cryptology '15) gave a tight standard-model security proof under the $q$-SDH assumption for a less efficient variant of the scheme, called BBS+---here, $q$ is the number of issued signatures. In contrast, the security proof for BBS (Tessaro and Zhu, EUROCRYPT '23), also under the $q$-SDH assumption, is \emph{not} tight. Nonetheless, this recent proof shifted both standardization and industry adoption towards the more efficient BBS, instead of BBS+, and for this reason, it is important to understand whether this tightness gap is inherent. Recent cryptanalysis by Chairattana-Apirom and Tessaro (ASIACRYPT '25) also shows that a tight reduction to $q$-SDH is the best we can hope for.

    This paper closes this gap in two different ways. On the positive end, we show a novel tight reduction for BBS in the case where each message is signed at most once--this case covers in particular the common practical use case which derandomizes signing. On the negative end, we use a meta-reduction argument to prove that if we allow generating multiple signatures for the same message, then {\em no} algebraic reduction to $q$-SDH (and its variants) can be tight.  
\end{abstract}

%% file: tex_files/intro.tex
\section{Introduction}

Many privacy-preserving authentication protocols benefit from digital signature schemes that are {\em fully algebraic}, i.e., all operations performed during signing are within the group (for example, signing does not involve the evaluation of a cryptographic primitive such as a hash function). Such signatures typically enable efficient zero-knowledge proofs of knowledge of valid message-signature pairs, possibly additionally ensuring the message satisfies a certain predicate. Typical applications include anonymous credentials~\cite{C:CamLys04,CCS:TAKS07}, direct anonymous attestation~\cite{chen2010daa,brickell2010pairing,camenisch2016anonymous}, and $k$-time anonymous authentication~\cite{SJ:AuSMC13}.

There are multiple examples of such signatures, such as structure preserving~\cite{C:AFGHO10} or Pointcheval-Sanders~\cite{RSA:PoiSan16} signatures. In this work, we focus on BBS signatures~\cite{C:BonBoySha04,SJ:AuSMC13,camenisch2016anonymous,EC:TesZhu23b}, which have been favored in real-world deployment, especially in the space of anonymous credentials. In particular, BBS signatures are implemented in industry~\cite{mattr-bbs,microsoft-bbs} and covered by standardization efforts by the W3C Verifiable Credentials working group~\cite{w3c-vc-draft} and by the Decentralized Identity Foundation, the latter leading to an RFC draft by the IRTF~\cite{irtf-cfrg-bbs-signatures-10}.

Due to widespread interest in adoption, it is also crucial to gain a complete understanding of their concrete security, but recent works on security proofs~\cite{EC:TesZhu23b} and cryptanalysis~\cite{AC:ChaTes25} highlight important gaps between upper and lower bounds on the concrete security of BBS. As the main contribution of this work, we complete the picture of our understanding of the concrete security of BBS signatures. We give a new tight proof for the most common use case of BBS signatures (which covers in particular its standardized version), along with an impossibility result showing that a more general tightness result is not possible. 

\heading{BBS and its security.} BBS signatures were first put forward by Camenisch and Lysyanskaya~\cite{C:CamLys04} as a standalone signature scheme, taking inspiration from the concurrent work by Boneh, Boyen, and Shacham~\cite{C:BonBoySha04}, which used a similar construction to add exculpability 
to the well-known BBS {\em group} signature scheme. (Hence, the rather confusing naming.) 
BBS can sign any message vectors $\vec m \in \Zp^\ell$, and produce a signature $\sigma = (A, e)$ consisting of a group element $A \in \GOne$ and a scalar $e \in \Zp$ such that \[
A = \frac{1}{x - e}\left(G_1 + \sum_{i = 1}^\ell \vec{m}[i] \vec{H}[i]\right)\;,
\] where $x$ is the secret key and $G_1, \vec{H}[1], \ldots, \vec{H}[\ell]$ are generators of $\GOne$. It is not hard to see that the signature can be verified, given the verification key $X_2 = x G_2 \in \GTwo$, using a bilinear pairing.

Neither work~\cite{C:CamLys04,C:BonBoySha04}, however, gave a proof of security. Au, Susilo, and Mu~\cite{SJ:AuSMC13} later proved the security of a less efficient variant of BBS, which is typically referred to as BBS+. One can think of BBS+ as a variant of BBS where the first message coordinate $\vec{m}[1]$ is set to a random scalar $s$, which is then included as a part of the signature $\sigma = (A, e, s)$.
Their security proof is based on the $\qSDH$ assumption in type-2 pairings with $\qnum$ denoting the number of signing queries. Later on, Camenisch, Drijvers, and Lehmann~\cite{camenisch2016anonymous} adapted the proof to the more efficient type-3 pairings. These proofs are not tight, due to an extra multiplicative advantage loss of $\qnum$. Sch\"age~\cite{JC:Schage15} gave an alternative tight reduction to $\qSDH$ for a scheme essentially equivalent to BBS+. 

Hofheinz and Kiltz~\cite{JC:HofKil12}, with a non-tight proof, showed existential unforgeability security for a signature scheme equivalent to BBS but with bit-vector messages $\vec m \in \bits^\ell$. Tessaro and Zhu~\cite{EC:TesZhu23b} more recently showed that the original BBS signature scheme (i.e., without using the random scalar $s$) can also be proved secure based on the $\qSDH$ assumption. This led to the RFC draft~\cite{irtf-cfrg-bbs-signatures-10} transitioning to BBS signatures instead of BBS+. However, unlike the case of BBS+, their proof is non-tight, whereas the tight proof~\cite{JC:Schage15} inherently relies on the extra randomization via the scalar $s$, and does not apply to BBS. They were only able to give a tight security proof in the Algebraic Group Model (AGM)~\cite{C:FucKilLos18}. This leads to the following natural question: 
\begin{quote}
{\em Can we give a tight security reduction for BBS signatures without resorting to the Algebraic Group Model?}
\end{quote}
We note that these works are complemented by recent cryptanalysis work by Chairattana-Apirom and Tessaro~\cite{AC:ChaTes25} which shows that {\em any} $\qSDH$ attack can be turned into an attack against all variants of BBS/BBS+ with equal advantage (and roughly equal running time). This means that a tight reduction to $\qSDH$ is the best we can hope for any of these schemes, and the question remains open on whether such a reduction exists for BBS.

\heading{Why does it matter?} Ascertaining whether a multiplicative $\qnum$ loss is inherent has important consequences on choosing parameters that guarantee a certain security level. To this end, we would typically plug into our bound the generic advantage upper bound $O(\qnum t^2/p)$ for breaking $\qSDH$ by a $t$-time adversary, where $p$ is the order of the group. This $\qSDH$ bound is tight for many choices of $\qnum$ due to the attacks by Cheon~\cite{EC:Cheon06} and Brown-Gallant~\cite{EPRINT:BroGal04}. A non-tight reduction would then give a concrete security advantage upper bound of $O(\qnum^2 t^2/p)$. For a typical parameter choice such as $p \approx 2^{256}$ and $\qnum = 2^{64}$, such a proof would only guarantee $64$-bit security ($t = 2^{64}$). In contrast, a tight reduction would ensure $96$-bit security, essentially matching the cryptanalysis of~\cite{AC:ChaTes25}.


\heading{Result I: Tightness for distinct messages.} For our first main contribution, we give an affirmative answer to the above question for the case of signing {\em distinct} messages. This special case is the {\em most relevant to practice.} While BBS indeed allows signing a message multiple times, picking distinct values of $e$, in practice one would derandomize this choice (using a pseudorandom function applied to the message), and this is the construction covered by the IRTF draft~\cite{irtf-cfrg-bbs-signatures-10}.
Accordingly, our result effects the parameter choices for the standardization by reinforcing that the version considered by the draft does not admit any better generic attack than that of~\cite{AC:ChaTes25}. 

Our reduction only shares very basic elements from prior work, but is otherwise entirely novel. As in prior work, the main barrier is handling the case where the adversary produces a forgery $\sigma^* = (A^*, e^*)$ for a value $e^*$ which has been used by one of the signatures that the adversary has previously obtained. The reduction is given a $\qSDH$ instance consisting of $xG_1, x^2 G_1, \ldots, x^{\qnum}G_1$ to simulate the $\ell + 1$ generators used by the scheme, hoping that a forgery of the above type will allow us to extract $x$ using a well-known trick by Boneh-Boyen~\cite{JC:BonBoy08}. In particular, one can think without loss of generality of such a strategy as assigning specific polynomials in $x$ to the simulated generator, i.e., for $\ell = 1$, we may set $\overline{G}_1 = f(x) G_1$ and $\overline{H} = g(x) G_1$, for known polynomials $f$ and $g$, and unknown $x$. Our strategy differs substantially from prior works in how we pick these polynomials, and how we generate the values $e$ used by each signature.  

In particular, verifying that our strategy produces the right distribution is non-trivial, and requires a fairly involved probabilistic analysis using the H-coefficient technique from symmetric cryptography. (This technique was also used in~\cite{EC:TesZhu23b} as well as more recently in~\cite{EC:BDLR25}.) One caveat is that we require the $d$-$\SDH$ assumption with $d = \Theta(\qnum)$ instead of $\qSDH$, but still with a small constant (i.e., $d \approx 10\qnum$). Note that this still implies the same asymptotic generic upper bound as $\qSDH$. 
We give further details below in our technical overview.

\heading{Result II: Non-tightness for general messages.} 
This still leaves open \ifnum\eprintversion=0\linebreak\fi whether such tight reductions exist when messages may repeat. On that note, we show via a meta-reduction argument that a tight reduction in the standard model is in fact {\em not} possible. Concretely, no algebraic and straight-line reduction to $d$-$\SDH$ (for {\em any} integer $d$) can achieve less than $\Theta(\qnum)$ factor loss, matching the reduction in~\cite{EC:TesZhu23b}. (Here, algebraic reductions can explain their group element outputs with algebraic representations based on previously obtained elements.) This indeed covers most reductions for group-based signatures and all existing reductions for BBS. We stress that the reductions in our impossibility results are {\em in the standard model}, while our meta-reductions work in the AGM~\cite{C:FucKilLos18}. We also extend the impossibility to rewinding reductions: concretely, any reduction that runs the adversary $r$ times will incur at least a $\Theta(\qnum/r^2)$ factor loss.

\heading{Broader picture: Tight signatures.} 
Generally, the quest for tight security reductions has a long history, starting with \emph{concrete security}~\cite{FOCS:BDJR97}, and tight security reductions for public-key encryption~\cite{EC:BelBolMic00,C:HofJag12,EC:BJLS16}, digital signatures~\cite{RSA:CheJoy07,C:HofJag12,C:BlaKilPan14}, and other cryptographic primitives (e.g., key exchange~\cite{TCC:BHJKL15,EC:BJLS16,C:HesHofKoh18} or identity-based encryption~\cite{C:CheWee13,C:BlaKilPan14}). The motivation for these works is both understanding the exact security guarantees given by a particular construction, and optimizing parameter choices (as explained above).

In the case of digital signatures, we do have tightly secure constructions from standard complexity assumptions, even for schemes with additional properties (like structure-preserving~\cite{C:AHNOP17} or multi-signatures~\cite{AC:BacWag24}). However, those constructions are less efficient than “ordinary” (i.e., not known to be tightly secure) signature schemes. Little is known about the tight security of practical and popular signature schemes (such as BBS signatures), although of course it is particularly important to be able to give concrete security guarantees for such schemes.

We remark that our work {\em does not} consider the multi-user settings, where the adversary can (adaptively) corrupt multiple signers. Indeed, the lower bound by~\cite{EC:BJLS16} rules out tight reductions for multi-user security with adaptive corruptions of any standard model signature scheme. Thus, we cannot hope to prove BBS tight in this setting.

%% file: tex_files/overview.tex
\section{Technical Overview} \label{sec:tech-overview}

\subsection{Prior Non-Tight Reductions}\label{sec:abstract-prior-approaches}
We start by reviewing prior security reductions for BBS/BBS+ signatures for message length $\ell = 1$ (extending to any $\ell > 1$ is quite simple). Recall that a BBS signature on message $m \in \Zp$ for a verification key $X_{2, 1} \in \GTwo$ and secret key $x \in \Zp$ is of the form $(A, e)$ where $A = \frac{1}{x - e} (G_1 + m H)$ and $e \getsr \Zp$ ($G_1, H$ are the public parameters). For simplicity, we will call $e$ the {\em tag}. The BBS security proof (for $\ell = 1$) considers two cases for the forgery $(A^*, e^*)$ with respect to the previously obtained signatures $(\sigma_i = (A_i, e_i))_{i \in [\qnum]}$: 
(a) $e^* \notin \{e_i\}_{i \in [\qnum]}$ (which already has a tight reduction), and
(b) $e^* \in \{e_i\}_{i \in [\qnum]}$ (which is our focus throughout this paper). 
%
In both cases, the reduction taking as input a $\qSDH$ instance $(G_1, (X_{1,i} = x^iG_1)_{i \in [\qnum]}, G_2, X_{2,1} = xG_2)$ follows a common template: \begin{description}[topsep=0pt]
	\item [Set up the public parameters:] The reduction sets the verification key as $X_{2,1}$ and the public parameters as $\overline{G}_1 = f(x) G_1, \overline{H} = g(x) G_1$ depending on some polynomials $f, g \in \Zp^{\leq \qnum}[\varX]$. These can be efficiently computed from the $\qSDH$ instance. (At this stage, some tags $e_i$ may not be determined.)
	
	\item [Simulate signing queries:] On query $m_i \in \Zp$, select $e_i$ such that $(\varX - e_i) \divides f(\varX) + m_i g(\varX)$. Indeed, this divisibility condition is necessary for the reduction to compute $A_i =  \frac{f(x) + m_i g(x)}{x - e_i}G_1$ from its SDH instance. 
	
	\item [Extract SDH solution:] On the forgery $(m^*, \sigma^* = (A^*, e^*))$, if $(\varX - e^*) \ndivides f(\varX) + m^* g(\varX)$, the reduction can extract an SDH solution $(e^*, Z = \frac{1}{x - e^*} G_1)$ via the by-now well-known technique by Boneh-Boyen~\cite{JC:BonBoy08} (see \Cref{remark:compute-sdh}).
\end{description} 
Importantly, the chosen $e_i$'s should still look uniformly random. 
For Case (b), if the messages $m_1, \ldots, m_\qnum$ are known to the reduction beforehand, the reduction can sample $e_1, \ldots, e_\qnum \getsr \Zp$ and find $f$ and $g$ such that for each $i \in [\qnum]$, $(\varX - e_i) \divides f(\varX) + m_i g(\varX)$ (allowing simulation) but $(\varX - e_i) \ndivides f(\varX) + m' g(\varX)$ (allowing extraction if $e^* = e_i$) for $m' \neq m_i$ via polynomial interpolation. This is the key insight behind the tight proof for BBS+ signatures (not BBS) given in \cite{JC:Schage15}. However, this only gives a tight reduction for non-adaptively chosen queries. 

With adaptive queries, the proof given by Tessaro and Zhu~\cite{EC:TesZhu23b} first guesses an index $i^* \in [\qnum]$ such that $e^* = e_{i^*}$, and set $f(\varX) = \alpha (\varX - e'_{i^*}) \prod_{i \neq i^*} (\varX - e_i)$ and $g(\varX) = \beta \prod_{i \neq i^*} (\varX - e_i)$. This allows them to simulate the $i^*$-th query for any $m_{i^*}$ using $e_{i^*} = \alpha e'_{i^*} + m_{i^*}$ and still extract a solution from a forgery $(m^* \neq m_{i^*}, A^*, e^* = e_{i^*})$. However, the guessing results in a $\qnum$ factor loss.  

\subsection{Our Tight Reduction}\label{sec:overview-tight-reduction}
We now give a high-level overview of our tight reduction against adversaries who make {\em distinct signing queries} $m_1, \ldots, m_\qnum \in \Zp$. This implies that derandomized BBS signatures (which is the version specified in the IRTF draft~\cite{irtf-cfrg-bbs-signatures-10}) is actually tightly secure from the $d$-$\SDH$ assumption for $d = \Theta(\qnum)$. 

\heading{Reduction idea.}
One property of Tessaro and Zhu's reduction is that for the $i^*$-th signing query, $e_{i^*}$ (not $e'_{i^*}$) can only be used to simulate a specific message $m_{i^*}$ and a signature $(A^*, e_{i^*})$ for any $m^* \neq m_{i^*}$ can be used to extract a SDH solution. Our hope is to achieve this property for a large number of tags used in the signing queries (a constant fraction suffices) while still being able to simulate without knowing the messages in advance. At a high level, we will set up $\overline{G}_1, \overline H$ so that (1) on query $m_i$, attempt to find such good $e_i$ to simulate the signing, and (2) if we cannot, we use the tags from a ``stash", denoted $\widetilde{e}_1, \ldots \widetilde{e}_\qnum$, prepared in advance; note that the tags from the stash will not allow us to extract an SDH solution. More precisely, our reduction strategy is as follows: 
\begin{description}[topsep=0pt]
	\item [Setting up $\overline{G}_1, \overline H$:] Sample a random monic polynomial $f$ of degree $d$, $\beta \getsr \Zp^*$, and the stash $\widetilde{e}_1, \ldots \widetilde{e}_\qnum \allowbreak \getsr \Zp$. Then, we set \begin{align*}
		\textstyle \overline{G}_1 = f(x) \prod_{i \in [\qnum]} (x - \widetilde{e}_i) G_1\;, \pcind[2] \overline H = \beta \prod_{i \in [\qnum]} (x - \widetilde{e}_i)  G_1\;.
	\end{align*}
	\item [Simulation:] On the $i$-th signing query $m_i \in \Zp$, check if $f(\varX) + m_i \beta$ has a ``unique" zero in $\Zp$, denoted $e$. If so, set $e_i \gets e$; otherwise, set $e_i \gets \widetilde{e}_i$. Then, compute $A_i = \frac{(f(x) + m_i \beta) \prod_{i \in [\qnum]} (x - \widetilde{e}_i)}{x - e_i} G_1$ as a linear combination of the SDH instance. This is possible because $(\varX - e_i) \divides (f(\varX) + m_i \beta) \prod_{i \in [\qnum]} (\varX - \widetilde{e}_i)$. 
	
	\item [Extraction:] On the forgery $(m^*, (A^*, e^*))$ such that $e^* = e_{i^*}$ for some $i^* \in [\qnum]$, if $i^*$ is the index where we use the stash, abort. Otherwise, since $m^* \neq m_{i^*}$ (by freshness of the forgery), $f(\varX) + m^* \beta$ does not have $e_{i^*}$ as a zero. Thus, $(\varX - e_{i^*}) \ndivides (f(\varX) + m^* \beta)$, allowing us to extract a SDH solution. 
\end{description}
The success of our reduction hinges on whether the index $i^*$ (defined by the forgery) is such that $e_{i^*}$ is not from the stash. More precisely, this is the event $i^* \notin S \subseteq [\qnum]$ where we define $S := \{i \in [\qnum]: |f^{-1}(- m_i \beta)| = 1\}$, as the set of indices $i$ where $e_i$ is a zero of $f(\varX) + m_i \beta$. Note that the set of preimages of $a$, $f^{-1}(a)$, is equivalent to the set of zeros of $f(\varX) - a$.
To show why our reduction is tight, we analyze the distribution of the transcript between the reduction and the adversary augmented with the set $S$, i.e., $(\overline{G}_1, \overline H, G_2, X_2, (m_i, e_i)_{i \in [\qnum]}, S)$. The formal analysis of the distribution relies on the H-coefficient technique, but essentially, the distribution of the transcript satisfies the following properties:
\begin{enumerate}[noitemsep,topsep=0pt,label={\bf (\arabic*)}]
	\item \label{prop:uniform-tags} The tags $e_i$'s are indistinguishable from uniformly random.
	\item \label{prop:S-indep} The distribution of $S$ is (almost) independent of the rest of the transcript.
	\item \label{prop:small-loss} For any $i^* \in [\qnum]$, the probability that a sampled $S$ contains $i^*$ is at least $e^{-1}$.~\footnote{We overload $e$ with the tags $e_i \in \Zp$ and the real number $e \in \mathbb{R}$ when clear from the context.}
\end{enumerate}
Assuming that the above are true, we will now see why our reduction is tight.
Due to \ref{prop:uniform-tags}, the view of the adversary within the reduction will be indistinguishable from its view in the unforgeability game. Thus, it still wins within the reduction with probability close to in the game. Accordingly, we know that the reduction wins given that $i^*$ is in $S$. Hence, by \ref{prop:S-indep} and \ref{prop:small-loss}, the reduction's loss is only $e^{-1}$. 
We refer to \ifnum\eprintversion=1 \Cref{sec:main-thm-statements} \else \Cref{sec:tight-thm-statements} \fi for the concrete security statement, with the advantage and running time of our reduction, and \Cref{sec:full-proof} for the formal proof.
%


\heading{Case study: one-query adversary.} We give an intuition why the properties \ref{prop:uniform-tags}--\ref{prop:small-loss} are satisfied by examining the simpler case of $\qnum = 1$. Specifically in this overview, we will consider a simplified transcript that omits the generators and the verification key\footnote{Our formal proof analyzes the full transcript, which makes things slightly more complex. Here, we only present the core difficulties in the analysis of the distribution.} and is of the form \[
	(m \in \Zp, e \in \Zp, S \subseteq \{1\})\;.
\] This is generated by (a) sampling $f \getsr \Zp^d[\varX]$, $\beta \getsr \Zp^*$ and $\widetilde{e} \getsr \Zp$, (b) running the adversary who picks $m$, (c) selecting the tag $e$ by finding the number of zeros of $f(\varX) + m \beta$ as in the reduction, and (d) setting $S$ as $\emptyset$ if $e$ is set to $\widetilde{e}$ and $\{1\}$ otherwise. We will now show \ref{prop:uniform-tags} and \ref{prop:S-indep} by considering the probability of each possible transcript being generated for two main cases: $S = \{1\}$ and $S = \emptyset$.

First, consider the transcript of the form $(\actualMsg, \actualE, \actualS = \{1\})$\footnote{The notation $\underline{(\cdot)}$ denotes an actual value, and not a random variable.} for any $\actualMsg, \actualE \in \Zp$. The corresponding event over the randomness $f, \beta, \widetilde{e}$ is ``$f(\varX) + \actualMsg \beta$ has only $\actualE$ as its zero" (so that $S = \{1\}$), or equivalently $f^{-1}(-\actualMsg \beta) = \{\actualE\}$.
Then, we calculate the probability of this event over the randomness as: (following a similar analysis from \cite{leont2006roots} counting monic degree-$d$ polynomials with no zeros in $\Zp$)
\begin{align*}
	&\textstyle \Pr_{f, \beta,\widetilde{e}}[f^{-1}(-\actualMsg \beta) = \{\actualE\}] = \Pr\left[(f(e) = -\actualMsg \beta) \eAnd \neg \bigcup_{e' \in \Zp \setminus \{e\}} (f(e') = -\actualMsg \beta)\right]\\
	&\textstyle = \Pr[f(e) = -\actualMsg \beta] - \Pr\left[(f(e) = -\actualMsg \beta) \eAnd \bigcup_{e' \in \Zp \setminus \{e\}} (f(e') = -\actualMsg \beta)\right] & \\
	&\textstyle = p^{-1} - \Pr\left[ \bigcup_{e' \in \Zp \setminus \{e\}} (f(e') = f(e) = -\actualMsg \beta)\right] \\ 
	&= p^{-1} - \sum_{\emptyset \neq T \subseteq \Zp \setminus \{e\}} (-1)^{|T| - 1} \Pr\left[ \forall e' \in T \cup \{e\} : f(e')= -\actualMsg \beta\right] \; 
\end{align*}
The first equality follows from $(f^{-1}(- \actualMsg \beta) = \{\actualE\})$ being equivalent to $(f(e) = - \actualMsg \beta) \eAnd \neg \bigcup_{e' \in \Zp \setminus \{e\}} (f(e') = -\actualMsg \beta)$ (ruling out other preimages $e' \neq e$). The second to last equality follows because $f$ is uniformly random. The last equality follows from the inclusion-exclusion principle.

Here, the event $\forall e' \in T \cup \{e\} : f(e')= -\actualMsg \beta$ is equivalent to writing $f(\varX) + \actualMsg \beta = \prod_{e' \in T \cup \{e\}} (\varX - e') \cdot g(\varX)$ where $g \in \Zp[\varX]$ is monic and $\deg g = d - |T| - 1$. Therefore, we have that 
\ifnum\eprintversion=1
\[
\Pr\left[ \forall e' \in T \cup \{e\} : f(e')= -\actualMsg \beta\right] = \begin{cases}
	p^{-|T| - 1}& |T| \leq d - 1 \\
	0 & \text{otherwise}
\end{cases}\;.
\]
\else
$\Pr[ \forall e' \in T \cup \{e\} : f(e')= -\actualMsg \beta] = p^{d - |T| - 1}/p^d = p^{-|T| - 1}$ if $|T| \leq d - 1$ and 0 otherwise.
\fi
Hence, we continue the calculation as 
\begin{align*}
	& p^{-1} - \sum_{\emptyset \neq T \subseteq \Zp \setminus \{e\}} (-1)^{|T| - 1} \Pr\left[ \forall e' \in T \cup \{e\} : f(e')= -\actualMsg \beta\right]\\
	& = p^{-1} - \sum_{i = 1}^{d - 1} (-1)^{i - 1} \binom{p - 1}{i} p^{- i - 1} \pcind[2] \; \; ; \text{Counting subsets $T \subseteq \Zp \setminus \{e\}$ of size $i$}\\
	& = \underbrace{p^{-1}}_{\Pr[e = \actualE]} \cdot \underbrace{\left(\sum_{i = 0}^{d - 1} (-1)^{i} \binom{p - 1}{i} p^{- i - 1} \right)}_{\Pr[S = \{1\}]}\;.
\end{align*}
This clearly shows properties \ref{prop:uniform-tags} and \ref{prop:S-indep} with the two terms attributing to the probabilities of getting $\actualE$ and $\actualS = \{1\}$. 

Next, for transcripts of the form $(\actualMsg, \actualE, \actualS = \emptyset)$, the corresponding event over the randomness is ``$|f^{-1}(-\actualMsg \beta)| \neq 1$ and $\widetilde{e} = \actualE$." Since $f$ and $\widetilde{e}$ are independently sampled, we can write \[ 
	\Pr_{f, \beta, \widetilde{e}}[|f^{-1}(-\actualMsg \beta)| \neq 1 \eAnd \widetilde{e} = \actualE] = \underbrace{\Pr[\widetilde{e} = \actualE]}_{\Pr[e = \actualE] = p^{-1}} \cdot \underbrace{\Pr [|f^{-1}(-\actualMsg \beta)| \neq 1]}_{\Pr[S = \emptyset]}\;.
\] 
From the above analysis of both cases, properties \ref{prop:uniform-tags} and \ref{prop:S-indep} are satisfied. 

Finally, we show property \ref{prop:small-loss}, i.e., $S = \{1\}$ with probability roughly $e^{-1}$. In particular, the term $\Pr[S = \{1\}]$ derived above can be estimated as \[
\sum_{i = 0}^{d - 1} (-1)^{i} \binom{p - 1}{i} p^{- i - 1}  \approx \sum_{i = 0}^{p - 1} (-1)^{i} \binom{p - 1}{i} p^{- i - 1} = \left(1 - \frac{1}{p}\right)^{p - 1} \approx e^{-1}.
\] The first approximation follows when $d$ is large enough, making the truncated terms relatively small. The equality follows from  binomial expansion. The last approximation follows if $p$ is large. This completes the 1-query case analysis.

\heading{Generalizing to $\qnum$ queries.} 
To extend our analysis to $\qnum$-query adversaries, we make a crucial observation from the analysis above. In particular, we observe that the approximated value $(1 - \frac{1}{p})^{p - 1}$ is {\em exactly} the probability of sampling a random function $f: \Zp \to \Zp$ and for a fixed $z \in \Zp$, $|f^{-1}(z)| = 1$. This follows by counting the number of such $f$ as follows: (1) pick the unique preimage $e$ of $z$ ($p$ possible ways), (2) for $e' \neq e$, pick the image $f(e') \neq z$ ($(p - 1)^{p - 1}$ ways). 

For $\qnum$-query adversaries, we can consider a simplified transcript of the form $((m_i, e_i)_{i \in [\qnum]}, S \subseteq [\qnum])$. For the probability of getting a particular transcript $((\actualMsg_i, \actualE_i)_{i \in [\qnum]}, \actualS)$, the corresponding event over the randomness $f, \beta, (\widetilde{e}_i)_{i \in [\qnum]}$ is: \begin{equation*}
	\left(\forall i \in \actualS : f^{-1}(-\actualMsg_i \beta) = \{\actualE_i\}\right) \eAnd \left(\forall i \notin \actualS : |f^{-1}(-\actualMsg_i \beta)| \neq 1 \eAnd \widetilde{e}_i = \actualE_i \right) \;.
\end{equation*} Our analysis then relies on the observed similarities between random polynomials and random functions $\Zp \to \Zp$ on events regarding the number of preimages for {\em any distinct} elements $z_1, \ldots, z_\qnum$ (setting $z_i = -\actualMsg_i \beta$ in the analysis). Essentially, we show that for $d = \Theta(\qnum)$, any {\em distinct} $z_1, \ldots, z_\qnum$ and {\em distinct} $\actualE_1,\ldots, \actualE_\qnum$ 
\ifnum\eprintversion=1
\[
	\underset{f \getsr \Zp^{d}[\varX]}{\Pr}\left[ \begin{array}{l}
		\forall i \in  \actualS: |f^{-1}(z_i)| = 1\\
		\forall i \notin \actualS: |f^{-1}(z_i)| \neq 1\\
	\end{array}\middle| \begin{array}{l}
	\forall i \in \actualS :
	 f(\actualE_i) = z_i 
	\end{array}\right] \approx \underset{f : \Zp \to \Zp}{\Pr}\left[\begin{array}{l}
	\forall i \in  \actualS: |f^{-1}(z_i)| = 1\\
	\forall i \notin \actualS: |f^{-1}(z_i)| \neq 1\\
	\end{array}\right].
\]\else
\begin{align*}
	&\underset{f \getsr \Zp^{d}[\varX]}{\Pr}\left[ \begin{array}{l}
		(\forall i \in  \actualS: |f^{-1}(z_i)| = 1) \eAnd (\forall i \notin \actualS: |f^{-1}(z_i)| \neq 1)\\
	\end{array}\middle| \begin{array}{l}
		\forall i \in \actualS : f(\actualE_i) = z_i 
	\end{array}\right] \\
	&\approx \underset{f : \Zp \to \Zp}{\Pr}\left[\begin{array}{l}
		(\forall i \in  \actualS: |f^{-1}(z_i)| = 1) \eAnd 
		(\forall i \notin \actualS: |f^{-1}(z_i)| \neq 1)
	\end{array}\right].
\end{align*}
\fi
This says that the extra conditioning on which preimage $\actualE_i$ is used does not influence the distribution much.
We abstract out the calculation of these probabilities and give our reduction with formal analysis in \Cref{sec:random-functions-poly,sec:full-proof}, respectively.

Finally, we remark that our strategy applies only when the messages $m_1, \ldots, \allowbreak m_\qnum$ are distinct; otherwise, two queries $m_i = m_j$ can be answered with the same $e_i = e_j$ with very high probability (and distinguishable from the original game). This limitation turns out to be inherent, and no tight reduction exists for the general case. We explain in the next section why such a loss cannot be avoided.

\subsection{Impossibility of Tight Reductions for General Messages}\label{sec:lower-bound-overview}

For our impossibility results, let us consider the supposed ``worst-case scenario" 
where {\em all the signing queries are identical,} and the adversary outputs a forgery with a tag $e^* = e_i$ for some $i \in [\qnum]$. From the sketched proof structure in \Cref{sec:abstract-prior-approaches}, the reduction can be viewed as committing to polynomials $f$ and $g$ through the group elements $\overline{G}_1$ and $\overline{H}$. 
Moreover, the tags $e_1, \ldots, e_\qnum$ should satisfy that $(\varX - e_i)$ divides $f + m_i g$. We can then categorize the tag $e_i$ as ``useful" for a message $m_i$ when $(\varX - e_i)$ only divides $f + m_i g$ but not $f + m' g$ for any $m' \neq m_i$; such tags allow extracting the SDH solution. 
If an adversary can force only a small portion (say $O(1)$) of the tags to be useful and later forge using a uniformly random tag from $e_1, \ldots, e_\qnum$, then the reduction can reliably extract the SDH solution only with $O(1/\qnum)$ probability, leading to a $O(\qnum)$ loss. 

Our key observation is that for any fixed polynomials $f$ and $g$, no tag $e$ can be useful for two messages $m \neq m'$ (i.e., if $(\varX - e)$ divides $f + m g$ but not $f + m' g$ for $m' \neq m$, $e$ cannot be useful for $m'$). Hence, for a uniformly random message $m \getsr \Zp$, there exists {\em at most one} useful tag on average. Therefore, if an adversary makes $\qnum$ identical signing queries on a random $m$ and forges using a tag $e_i$ for a random $i \getsr [\qnum]$, the reduction can only find a SDH solution with $1/\qnum$ probability. Otherwise, the reduction has to output a tag $e_i$ such that $(\varX - e_i) \ndivides f + m g$, in which case it inherently breaks the SDH assumption. 

In \Cref{sec:lower-bound}, we formalize the sketched intuition as a meta-reduction against any {\em algebraic} reductions (i.e., it outputs group elements along with their algebraic representations). We also extend the impossibility to: (a) reductions to variants of the SDH assumption, (b) general verification key $\vk$ (this overview assumes $\vk = X_{2,1}$ taken from the SDH instance), (c) rewinding reductions \ifnum\eprintversion=1(a reduction running the adversary $r$ times incurs at least $\Theta(\qnum/r^2)$ factor loss)\fi, and (d) fine-grained adversaries making queries on at least $k < \qnum$ distinct messages. 

%% file: tex_files/prelims.tex
\section{Preliminaries}\label{sec:prelims}

\heading{Notations.} Let $[n .. m] = \{n, n+1, \ldots, m\}$ for any $n, m \in \bbZ$ where $n \leq m$ and $[n] = [1 .. n]$ for any $n \geq 1$. 
We use $\secpar$ as the security parameter, denote vectors with the arrow symbols (e.g., $\vec{v}, \vec{H}$), write $\vec{v}_{S}$ for the subvector $(v_i)_{i \in S}$ of $\vec{v}$. We define an error factor $\Error(a, l; p) := \prod_{k = 0}^{l-1} \left(1 - \frac{a + k}{p}\right) = p^{-l} \cdot \frac{(p - a)!}{(p - a - l)!}$ (used repeatedly in our proofs) for integers $p \geq a, l \geq 0$. For boolean-or and boolean-and, we use the operators $\cup$ and $\cap$, respectively. 

\heading{Polynomials.} We denote formal variables with sans-serif letters (e.g., $\varX, \varY$). For any prime modulus $p$, let $\Zp[\varX]$ (and $\Zp^d[\varX]$) denote the ring of ({\em monic} {degree-$d$}) polynomials $g(\varX) = \sum_{i = 0}^d a_i \varX^i$ with coefficients in $\Zp$ and $\deg g = d$ as the degree of $g(\varX)$. We often refer to $g(\varX)$ using the shorthand $g$. 
Denote $\Root(f ) := \{\alpha \in \Zp: f(\alpha) = 0\}$ as the set {\em (not a multiset)} of zeros of $f \in \Zp[\varX]$. 
\ifnum\eprintversion=1 We refer to the notation as a set and {\em not a multiset}; accordingly, if $(\varX - a)^2 | f(\varX)$, we still count $\alpha$ as one zero. \fi


\heading{Bilinear pairing groups.}  For any group $\Group$ with prime-order $p$, we write $0_\Group$ as the identity element, $\Group^*$ as the set of generators of $\Group$, and group elements and scalars with upper- and lower- case letters, respectively. We adopt additive notations. For $G \in \Group^*$ and $H\in \Group$, we use $\mathrm{DL}_G(H) \in \Zp$ as the discrete logarithm of $H$ base $G$. 
A {\em bilinear group parameter generator} is a probabilistic algorithm $\GroupGen$ with input $1^\secpar$ and output $\pGroupDescript$. Here, $\GOne, \GTwo$ and $\GTarg$ are groups of $\secpar$-bit prime order $p$, and $\pairingF : \GOne \times \GTwo \to \GTarg$ is a bilinear map satisfying (1) {\em bilinearity}, i.e., for any $A \in \GOne, B \in \GTwo$ and $x, y \in \Zp$, $\pairingOp{xA}{yB} = \allowbreak (xy) \cdot \pairingOp{A}{B}$, and (2) {\em non-triviality}, i.e., for any $G_1 \in \GOne^*, G_2 \in \GTwo^*$, $\pairingOp{G_1}{G_2} \in \GTarg^*$. 

\heading{Signature schemes.} A digital signature scheme $\SignS$ is a tuple of algorithms $(\Setup{\SignS}, \allowbreak  \KG{\SignS}, \Sign{\SignS}, \allowbreak \Ver{\SignS})$ 
\ifnum\eprintversion=1
with the following syntax: 
\begin{itemize}[topsep=0pt]
    \item The setup algorithm $\Setup{\SignS}(1^\lambda)$ generates public parameters $\param$ which are implicit inputs to all other algorithms and also define the message space $\SignS.\msgSp = \SignS.\msgSp(\param)$. The key generation algorithm $\KG{\SignS}()$ outputs the signing and verification keys $(\sk, \vk)$. 
    \item The (possibly randomized) signing algorithm $\Sign{\SignS}(\sk, m)$ takes as inputs, the signing key $\sk$ and a message $m \in \SignS.\msgSp$, and outputs a signature $\sigma$. 
    \item The deterministic verification algorithm outputs a bit $\Ver{\SignS}(\vk, m, \sigma)$.
\end{itemize}
Correctness says that for any public parameters $\param$ and key pair $(\sk, \vk)$ generated from the setup and key generation algorithms and any message $m \in \msgSp$, the signature $\sigma \getsr \Sign{\SignS}(\sk, m)$ always satisfies $\Ver{\SignS}(\vk, m, \sigma) = 1$.  We consider {\em strong unforgeability} security, defined by the game $\SUF^\advA_{\SignS}(\secpar)$ (given in \Cref{fig:suf-game}), and for any adversary $\advA$, its advantage is \ifnum\eprintversion=1\[
\Adv^\suf_{\SignS}(\advA, \secpar) := \Pr[\SUF^\advA_{\SignS}(\secpar) = 1]\;.
\]\else
$\Adv^\suf_{\SignS}(\advA, \secpar) := \Pr[\SUF^\advA_{\SignS}(\secpar) = 1]$.
\fi
\begin{figure}[t]
	\centering
	\begin{pchstack}[boxed]
		\procedure{\scriptsize Game $\SUF^\advA_{\SignS}(\secpar)$:}{
			\mathsf{Sigs} \gets \emptyset\; ; \; \param \getsr \Setup{\SignS}(1^\secpar)\\
			(\sk, \vk) \getsr \KG{\SignS}()\\
			(m^*, \sigma^*) \getsr \advA^\oraS(\param, \vk)\\
			\pcreturn (m^*, \sigma^*) \notin \mathsf{Sigs} \eAnd \Ver{\SignS}(\vk, m^*, \sigma^*) = 1 
		}
		\procedure{\scriptsize Oracle $\oraS(m)$:}{
			\sigma \getsr \Sign{\SignS}(\sk, m)\\
			\pcif \sigma \neq \bot \pcthen \\
			\pcind \mathsf{Sigs} \gets \mathsf{Sigs} \cup \{(m, \sigma)\}\\
			\pcreturn \sigma
		}
	\end{pchstack}
	\caption{Strong unforgeability game}
	\label{fig:suf-game}
\end{figure}
\else 
such that: the setup $\Setup{\SignS}(1^\lambda)$ generates public parameters $\param$, which are implicit inputs to all other algorithms and define the message space $\SignS.\msgSp = \SignS.\msgSp(\param)$; the key generation $\KG{\SignS}()$ outputs the signing and verification keys $(\sk, \vk)$; the randomized signing $\Sign{\SignS}(\sk, m \in \SignS.\msgSp)$ outputs a signature $\sigma$ of $m$; and the verification $\Ver{\SignS}(\vk, m, \sigma)$ outputs a bit. We consider standard {\em strong unforgeability}, defined in the full version.
\fi

\heading{Security assumptions.} We use the {\em $(d_1, d_2)$-Strong Diffie-Hellman} ($(d_1, d_2)$-$\SDH$) assumption in a format supporting type-3 pairings. This generalizes the $\qSDH$ assumption defined by Boneh and Boyen~\cite{JC:BonBoy08}, obtained by setting $d_1 = \qnum$ and $d_2 = 1$. We also consider the {\em $(d_1, d_2)$--Discrete Logarithm} ($(d_1, d_2)$-$\DL$)~\cite{TCC:Lipmaa12} assumption, which is implied by the $(d_1, d_2)$-$\SDH$ assumption. We simply refer to $(1, 1)$-$\DL$ as the {\em Discrete Logarithm} ($\DL$) assumption. The corresponding games \ifnum\eprintversion=1 $(d_1, d_2)\text{-}(\SDH/\DL)$ \fi are formalized in \Cref{fig:assumptions}, and the corresponding advantages for any adversary $\advA$ are \ifnum\eprintversion=1\[
\Adv^{(d_1, d_2)\text{-}(\sdh/\dl)}_{\GroupGen}(\advA, \secpar) = \Pr[(d_1, d_2)\text{-}(\SDH/\DL)^{\advA}_{\GroupGen}(\secpar) = 1]\;.
\] \else
$\Adv^{(d_1, d_2)\text{-}(\sdh/\dl)}_{\GroupGen}(\advA, \secpar) = {\Pr[(d_1, d_2)\text{-}(\SDH/\DL)^{\advA}_{\GroupGen}(\secpar) = 1]}$.
\fi

\begin{figure}[t]
	\centering
	\begin{pchstack}[boxed]
		\procedure{\scriptsize Game $(d_1, d_2)$-$\SDH^\advA_{\GroupGen}(\secpar)$:}{
			\param = \pGroupDescript \getsr \GroupGen(1^\secpar)\\
			G_1 \getsr \GOne^*; G_2 \getsr \GTwo^*; x \getsr \Zp\\
			\ifnum\eprintversion=0
			\mathsf{input} \gets (\param, G_1, (x^i G_1)_{i \in [d_1]}, G_2, (x^i G_2)_{i \in [d_2]})\\
			(e, Z) \getsr \advA(\mathsf{input})\\
			\else
			(e, Z) \getsr \advA(\param, G_1, (x^i G_1)_{i \in [d_1]}, G_2, (x^i G_2)_{i \in [d_2]})\\
			\fi
			\pcreturn \textstyle \left(Z = \frac{1}{x + e} G_1\right)}
		\procedure{\scriptsize Game $(d_1, d_2)$-$\DL^\advA_{\GroupGen}(\secpar)$:}{
			\param = \pGroupDescript \getsr \GroupGen(1^\secpar)\\
			G_1 \getsr \GOne^*; G_2 \getsr \GTwo^*; x \getsr \Zp\\
			\ifnum\eprintversion=0
			\mathsf{input} \gets (\param, G_1, (x^i G_1)_{i \in [d_1]}, G_2, (x^i G_2)_{i \in [d_2]})\\
			x' \getsr \advA(\mathsf{input})\\
			\else
			x' \getsr \advA(\param, G_1, (x^i G_1)_{i \in [d_1]}, G_2, (x^i G_2)_{i \in [d_2]})\\
			\fi
			\pcreturn (x = x')}
	\end{pchstack}
	\caption{Games $(d_1, d_2)$-$\SDH$ and $(d_1, d_2)\text{-}\DL$}
	\label{fig:assumptions}
\end{figure}

\begin{remark}\label{remark:compute-sdh}
    Our security proofs will rely on the observation (due to Boneh and Boyen~\cite{JC:BonBoy08}) that, given group elements $(x^i G_1)_{i \in [0, d]}$ and $A = \frac{f(x)}{x - e} G_1$ for any $f \in \Zp[\varX]$ of degree at most $d$ and a scalar $e \neq x$, if $f(e) \neq 0$, one can efficiently compute a $d$-$\SDH$ solution 
    $(-e, \frac{1}{x-e}G_1)$. Due to the polynomial remainder theorem, we can write $f(\varX) = g(\varX) (\varX - e) + r$ where $r = f(e) \neq 0$ and $\deg g \leq d - 1$. Thus, $A = (g(x) + \frac{r}{x - e})G_1$. Then, the solution 
    \ifnum\eprintversion=1
    \[
         \frac{1}{x - e} \; G_1 = \frac{A - g(x) G_1}{r} \;,
    \]is efficiently computable, as $g(x) G_1$ can be computed from the given inputs.\else
    $\frac{1}{x - e} \; G_1$ is efficiently computable.
    \fi  
\end{remark}

\heading{Algebraic group model (AGM).} Our impossibility results on tightness require the AGM. We assume that the reductions are algebraic~\cite{C:FucKilLos18}, i.e., when outputting a group element, it also provides the representation as a linear combination of previously seen group elements to the meta-reduction. 

\ifnum\eprintversion=1
\heading{Useful inequalities.} 
We recall several useful inequalities. The first one, which is used throughout the paper, is due to Bonferroni's giving upper and lower bounds for probability of union of sets/events through the inclusion-exclusion principle. The second is a lower bound on factorials due to Sterling's approximation. The third lemma follows from calculus and implies the fourth lemma.
\begin{lemma}\label{lemma:inclusion-exclusion-bound}
	For subsets $A_1,\ldots,A_n \subseteq \calU$ of some universe $\calU$, let $1 \leq K \leq n$ be an even integer and $\Pr$ denotes $\Pr_{u \getsr \calU}$. Then,
	\begin{align*}
		\sum_{\substack{B \subseteq [n] : \\ 1\leq |B| \leq K}} (-1)^{|B| - 1}\Pr\left[ \bigcap_{i \in B} A_i \right] &\leq \Pr\left[ \bigcup_{i \in [n]} A_i \right] = \sum_{\emptyset \neq B \subseteq [n]} (-1)^{|B| - 1} \Pr\left[ \bigcap_{i \in B} A_i \right] \\
		&\leq \sum_{\substack{B \subseteq [n] : \\ 1\leq |B| \leq K + 1}} (-1)^{|B| - 1} \Pr\left[ \bigcap_{i \in B} A_i \right]\;.
	\end{align*}
\end{lemma}

\begin{lemma}\label{lem:sterling-lower-bound}
	For positive integer $k$,  $k! \geq \left(\frac{k}{e}\right)^k$
\end{lemma}

\begin{lemma}\label{lem:log-expo-bounds}
	For any real $x$, $1 + x \leq e^x$. For any $x > -1$, $\ln(1 + x) \geq \frac{x}{1 + x}$. 
\end{lemma}
\begin{proof}
	The second inequality follows by substituting $x = - \ln(1 + x)$ (which is well-defined for $x > -1$) into the first inequality, giving \[
	1 - \ln(1 + x) \leq \frac{1}{1 + x} \iff \ln(1 + x) \geq \frac{x}{1 + x}\;. \pcind[2] \qed
	\]
\end{proof}

\begin{lemma}\label{lem:lower-bound-expo}
	For real numbers $0 < l, k < p$, $\left(1- \frac{l}{p}\right)^{p -k} \geq e^{-l \cdot \frac{p - k}{p - l}}$. 
\end{lemma}
\begin{proof}
	By setting $x = - \frac{l}{p}$ in \Cref{lem:log-expo-bounds}, we have \begin{align*}
		\left(1- \frac{l}{p}\right)^{p -k} &= e^{(p - k) \cdot \ln \left(1- \frac{l}{p}\right)} \geq e^{(p - k) \cdot \ln \frac{- l / p}{1 - l/p}} = e^{-l \cdot \frac{p - k}{p - l}}\;. \pcind[2] \qed
	\end{align*}
\end{proof}
\else
\heading{Useful inequalities.} 
We recall a lower bound with the proof  in the full version.
\begin{lemma}\label{lem:lower-bound-expo}
	For real numbers $0 < l, k < p$, $\left(1- \frac{l}{p}\right)^{p -k} \geq e^{-l \cdot \frac{p - k}{p - l}}$. 
\end{lemma}
\fi

\heading{BBS signatures.} The signature scheme $\BBS = \BBS[\GroupGen, \ell]$ in \Cref{fig:bbs-signature} is parameterized by \ifnum\eprintversion=1 a group parameter generator\fi  $\GroupGen$ and the message length $\ell \geq 1$. We make a syntactical change replacing $(x + e)$ with $(x - e)$.
\ifnum\eprintversion=1 This is for readability since we can write $e$ as zero of some polynomial $f$ if $(\varX - e) | f$ instead of $-e$. \fi
When $x - e = 0$, we denote $1/0 = 0$, following the notations in~\cite{EC:TesZhu23b}.
We also consider derandomized BBS, denoted $\DBBS = \DBBS[\GroupGen, \ell, \PRF]$ where $e$ is derived by evaluating the pseudorandom function $\PRF$ on $\vec{m}$ and the PRF key that is sampled at key generation. \smallskip

\begin{figure}[t]
	\begin{pchstack}[center,boxed,space=2pt]
		\begin{pcvstack}
			\procedure{\scriptsize Algorithm $\Setup{\BBS}(1^\secpar):$}{
				(p, \GOne, \GTwo, \GTarg, \pairingF) \getsr \GroupGen(1^\secpar)\\
				G_1 \getsr \GOne^*\; ; \; G_2 \getsr \GTwo^*\; ; \; \vec{H} \getsr \GOne^\ell\\
				\pcreturn \param \gets (p, G_1, \vec{H}, G_2,\GOne, \GTwo, \GTarg, \pairingF)}
			\procedure{\scriptsize Algorithm $\KG{\BBS}():$}{
				x \getsr \Zp \ifnum\eprintversion=1 \\ \else \; ; \; \fi
				\pcreturn (\sk \gets x, \vk \gets x G_2)}
		\end{pcvstack}
		\begin{pcvstack}
			\procedure{\scriptsize Algorithm $\Sign{\BBS}(\sk = x, \vec{m} \in \Zp^\ell):$}{ 
				\textstyle e \getsr \Zp \; ; \; A \gets \frac{1}{x - e} \left(G_1 + \sum_{i = 1}^\ell \vec{m}[i] \vec{H}[i]\right)\\
				\pcreturn (A, e)}
			\procedure{\scriptsize Algorithm $\Ver{\BBS}(\vk, \vec{m}, \sigma = (A, e)):$}{ 
				\textstyle C \gets G_1 + \sum_{i = 1}^\ell \vec{m}[i] \vec{H}[i]\\
				\pcreturn \pairingOp{A}{\vk - e G_2} = \pairingOp{C}{G_2}}
		\end{pcvstack}
	\end{pchstack}
	
	\caption{Signature scheme $\BBS = \BBS[\GroupGen, \ell]$. }
	\label{fig:bbs-signature}
\end{figure}

%% file: tex_files/tight-security-statements.tex
\ifnum\eprintversion=1 \section{Tight Security Proofs of BBS Signatures}\label{sec:main-thm-statements}

In this section, we state our tight security guarantees for BBS signatures against adversaries with distinct signing queries in \Cref{sec:tight-thm-statements} and discuss our reduction running time in \Cref{sec:discuss-runtime}.

\subsection{Tight Concrete Security of BBS Signatures}
\else
\section{Tight Concrete Security of BBS Signatures}
\fi \label{sec:tight-thm-statements}
The following theorem establishes the tight SUF security of $\BBS$ signatures under the $d$-$\SDH$ assumption (with $d = \Theta(\qnum + \secpar)$) against adversaries making $\qnum$ \underline{distinct} signing queries. It also implies tight security for $\DBBS = \DBBS[\GroupGen, \ell, \PRF]$ by additionally assuming the security of PRF.

\begin{theorem}\label{thm:tight-security-statement}
	Let $\GroupGen$ be a group parameter generator outputting bilinear \ifnum\eprintversion=0\linebreak\fi groups of prime-order $p = p(\secpar)$, $\ell = \ell(\secpar)$, $\BBS = \BBS[\GroupGen, \ell]$, and  $T_\Group = T_\Group(\secpar) , T_p = T_p(\secpar)$ be the running time for group exponentiation and field operation in $\Zp$, respectively. For any $\SUF$ adversary $\advA$ making at most $\qnum = \qnum(\secpar)$ \underline{distinct} signing queries with running time $t_\advA = t_\advA(\secpar)$, there exist adversaries $ \advB_1, \advB_2, \advB_3,\advB_{\BadQ}$, and $\advB_{\Coll}$ such that for $d' =  \lceil(e^2 + 2)\qnum + 2\log_2 p\rceil + 1$, 
	\begin{align*}
		&\Adv^{\suf}_{\BBS}(\advA, \secpar) \leq e \cdot \Adv^{d'\text{-}\sdh}_{\GroupGen}(\advB_1, \secpar)  + \Adv^{\dl}_{\GroupGen}(\advB_{\Coll}, \secpar)  + \Adv^{\dl}_{\GroupGen}(\advB_{\BadQ}, \secpar) \\
		&\pcind[3] +  \Adv^{\dl}_{\GroupGen}(\advB_2, \secpar) + \Adv^{\qsdh}_{\GroupGen}(\advB_3, \secpar)  + \frac{(4e + 1)\qnum^2 + (8e + 2)\qnum + 2e + 6}{2p} \;.
	\end{align*}
	Also, $\advB_1$ and $\advB_3$ run in time roughly $t_\advA + O(\qnum^2(\log^2 \qnum \log_2 p + T_\Group + T_p))$ and  $t_\advA + O(\qnum^2(T_\Group + T_p))$, respectively; $\advB_2$, $\advB_{\BadQ}$ and $\advB_{\Coll}$ run in time roughly $t_\advA$. 
\end{theorem}

We concretely specify the additive term in the reduction's running time in contrast to the prior works giving security proofs for BBS signatures~\cite{SJ:AuSMC13,JC:Schage15,camenisch2016anonymous,EC:TesZhu23b} and note that the additive running time of $ O(\qnum^2 (\log^2\qnum \log p + T_\Group + T_p))$  might not lead to a totally tight reduction. For example, an adversary $\advA$ may only ask queries, so that $t_\advA \approx \qnum \cdot \poly(\secpar) \ll \qnum^2  \cdot \poly(\secpar)$. However realistically, the concrete number of signing queries, which are online interactions, are relatively small compared to offline computations done by the adversary. As a result, we can mildly assume that $t_\advA \geq \widetilde{\Omega}(\qnum^2)$. More importantly, similar additive terms \ifnum\eprintversion=1 depending on $\qnum^2$ \fi also appear in prior works~\cite{EC:Gentry06,SJ:AuSMC13,camenisch2016anonymous,JC:Schage15,EC:TesZhu23b} on tight security reductions to $\qnum$-type assumptions\ifnum\eprintversion=1, and especially for BBS/BBS+ signatures\fi. We further discuss why such terms appear in \ifnum\eprintversion=1\Cref{sec:discuss-runtime}\else the full version\fi .

\begin{figure}[!t]
	\centering
	\begin{pchstack}[boxed]
		\procedure{\scriptsize Game ${\SUF'}^\advA_{\BBS}(\secpar)$:}{
			\mathsf{Sigs}, \mathsf{Tags} \gets \emptyset \; ; \; \param \getsr \Setup{\BBS}(1^\secpar)\\
			(\sk, \vk) \getsr \KG{\BBS}()\\
			(\vec m^*, \sigma^* = (A^* , e^*)) \getsr \advA^\oraS(\param, \vk)\\
			\pcreturn (\vec m^*, \sigma^*) \notin \mathsf{Sigs} \eAnd \protect\fhgamebox{$e^* \in \mathsf{Tags}$} 
			\eAnd \Ver{\BBS}(\vk, \vec m^*, \sigma^*) = 1 
		}
		\procedure{\scriptsize Oracle $\oraS(\vec m \in \Zp^\ell)$:}{
			\sigma = (A, e) \getsr \Sign{\BBS}(\sk, \vec m)\\
			\mathsf{Sigs} \gets \mathsf{Sigs} \cup \{(\vec m, \sigma)\}\\
			\protect\fhgamebox{$\mathsf{Tags} \gets \mathsf{Tags} \cup \{e\}$}\\
			\pcreturn \sigma
		}
	\end{pchstack}
	\caption{Game $\SUF'$ for $\BBS = \BBS[\GroupGen, \ell]$.} 
\label{fig:suf-bbs}
\end{figure}

The following lemma shows that the $\SUF$ security of $\BBS$ is implied tightly by the $\SUF'$ security of $\BBS$ formalized in \Cref{fig:suf-bbs}. The $\SUF'$ security is a special case for the security analysis of $\BBS$ in \cite{EC:TesZhu23b} where the forgery contains $e^*$ which is already output by the signing oracle. We do not give a proof of \Cref{lemma:adapt-TZ23-BBS} as it immediately follows from the proof of \cite[Theorem 1]{EC:TesZhu23b} by abstracting out the specific case where the previous reduction incurs a $\qnum$ factor in security loss.
\begin{lemma}\label{lemma:adapt-TZ23-BBS}
Let $\GroupGen, p, T_\Group, T_p, \BBS$ be as in \Cref{thm:tight-security-statement}. For any \ifnum\eprintversion=1 SUF \fi adversary $\advA$ making $\qnum = \qnum(\secpar)$ signing queries with running time $t_\advA = t_\advA(\secpar)$, there exist adversaries $\advB_1, \advB_2, \advB_3$ such that \[
\Adv^{\suf}_\BBS(\advA, \secpar) \leq \Adv^{\suf'}_{\BBS}(\advB_1, \secpar) + \Adv^{\dl}_{\GroupGen}(\advB_2, \secpar) + \Adv^{\qSDH}_{\GroupGen}(
\advB_3) + \frac{\qnum^2 + 2\qnum + 4}{2p}\;.
\]
Also, $\advB_1$ runs in time roughly $t_\advA$ and makes at most $\qnum$ queries to its signing oracle; $\advB_2$ and $\advB_3$ runs in time roughly $t_\advA$ and $t_\advA + O(\qnum^2(T_\Group + T_p))$, respectively.
\end{lemma}

We now establish the tight $\SUF'$ security of $\BBS$ against adversaries making \underline{distinct} signing queries via the two following lemmas, which combined with the above imply \Cref{thm:tight-security-statement}.
\Cref{lemma:tight-proof-extend} tightly reduces $\SUF'$ security for arbitrary length $\ell \geq 1$ messages to $\SUF'$ security for length-$1$.  \Cref{lem:tight-proof-main} is our core technical lemma establishing the tight security of $\SUF'$ for length-$1$ messages. 
\begin{lemma}\label{lemma:tight-proof-extend}
	Let $\GroupGen, p, T_\Group, T_p, \BBS$ be as in \Cref{thm:tight-security-statement}, and $\BBS_1 = \BBS \allowbreak [\GroupGen, 1]$. For any adversary $\advA$ making at most $\qnum = \qnum(\secpar)$ \underline{distinct} signing queries with running time $t_\advA = t_\advA(\secpar)$, there exist adversaries $\advB'$ and $\advB_{\Coll}$ running in time roughly $t_\advA$ such that
	\ifnum\eprintversion=1
	\begin{align*}
		\Adv^{\suf'}_{\BBS}(\advA, \secpar) &\leq  \Adv^{\suf'}_{\BBS_1}(\advB', \secpar) + \Adv^{\dl}_{\GroupGen}(\advB_{\Coll}, \secpar) \;.
	\end{align*}
	\else %
	$
	\Adv^{\suf'}_{\BBS}(\advA, \secpar) \leq  \Adv^{\suf'}_{\BBS_1}(\advB', \secpar) + \Adv^{\dl}_{\GroupGen}(\advB_{\Coll}, \secpar) + \frac{1}{p}$.
	\fi %
	Moreover, $\advB'$ makes at most $\qnum$ \underline{distinct} signing queries.
\end{lemma}
The proof\ifnum\eprintversion=1 \ in \Cref{sec:tight-proof-extend} \else \ in the full version \fi simply sets up the public parameters $\vec{H}[j] \gets \alpha_j H'$  for $\alpha_j \getsr \Zp$  (the reduction receives $H'$ and outputs $\vec{H}[j]$) and convert a length $\ell$ message $\vec{m}$ to $m = \sum_{j = 1}^\ell \alpha_j\vec{m}[j]$. Additionally, we rule out the event where two messages $\vec{m} \neq \vec{m}'$ maps to the same message via the map $\vec{m} \mapsto \sum_{j = 1}^\ell \alpha_j \vec{m}[j]$, which implies a collision $\sum_{j = 1}^\ell \vec{m}[j] \vec{H}[j] = \sum_{j = 1}^\ell \vec{m}'[j] \vec{H}[j]$ breaking DL.

\begin{lemma}\label{lem:tight-proof-main}
	Let $\GroupGen, p, T_\Group, T_p, \BBS_1$ be as in \Cref{lemma:tight-proof-extend}. For any adversary $\advA$ making at most $\qnum = \qnum(\secpar)$ \underline{distinct} signing queries with running time $t_\advA = t_\advA(\secpar)$, there exist adversaries  $\advB_1$ and $\advB_{\BadQ}$ such that, 
	\begin{align*}
		 \Adv^{\suf'}_{\BBS_1}(\advA, \secpar) \leq e \cdot \Adv^{d'\text{-}\sdh}_{\GroupGen}(\advB_1, \secpar) + \Adv^{\dl}_{\GroupGen}(\advB_{\BadQ}, \secpar) + \frac{e\cdot (2\qnum^2 + 4 \qnum + 1)}{p} \;,
	\end{align*}
	where ${d' =  \lceil(e^2 + 2)\qnum + 2\log_2 p\rceil + 1}$. Moreover, $\advB_1$ runs in time roughly $t_\advA + O(\qnum^2(\log^2 \qnum \log_2 p + T_\Group + T_p))$, and $\advB_{\BadQ}$ runs in time roughly $t_\advA$. 
\end{lemma}
The proof in \Cref{sec:full-proof} will rely heavily on the toolkit on statistical properties of random functions and random polynomials established in \Cref{sec:random-functions-poly}.

\ifnum\eprintversion=1
\input{tex_files/discussion-runtime}

\fi

%% file: tex_files/discussion-runtime.tex
\subsection{Discussion on the Reduction's Running Time}\label{sec:discuss-runtime}


As mentioned earlier, the $\widetilde{O}(\qnum^2)$ additive loss in our reduction's runtime also appears in prior works. 
The security reduction to $\qnum$-ABDHE assumption of Gentry's IBE~\cite{EC:Gentry06} requires additional $O(\qnum^2 \cdot T_\Group)$ where $T_\Group = T_\Group(\secpar)$ denotes running time for group exponentiation to compute group elements for the challenge queries. 
Similarly, the prior tight reductions for BBS/BBS+~\cite{SJ:AuSMC13,camenisch2016anonymous,JC:Schage15,EC:TesZhu23b} also has $O(\qnum^2 \cdot (T_\Group + T_p))$ additive term arising from computing for each signing query, the group element $A_i$ from the SDH instance and the underlying polynomial that represents $A_i$. 

\heading{Possible optimization.} {\em Only for certain bilinear groups} (e.g., ones where $p - 1$ is divisible by a large enough power-of-two $d$), there exists an algorithm, due to~\cite{TCC:GarHajOst20,EPRINT:FeiKho23} to precompute openings for KZG polynomial commitments~\cite{AC:KatZavGol10} in $O(d \log^2 d \cdot T_\Group)$ at $\leq d$ points where $d$ is the degree of the committed polynomial and $d  | (p - 1)$. This particular algorithm can be adapted to the reductions for BBS to precompute the group elements $A_i$ in the signatures faster than $O(\qnum^2)$ time. In particular, we can view the two group elements $\overline{G}_1$ and $\overline H$ that the reduction outputs as KZG polynomial commitments to two polynomials $f$ and $g$. (Recall the reduction structure in \Cref{sec:tech-overview}.) Then, the reduction can do the following: \begin{itemize}[topsep=0pt, itemsep=0pt]
	\item Setup $f, g$ and $\overline{G}_1 \gets f(x) G_1 , \overline H \gets g(x) G_1$. 
	
	\item Assuming that the tags $e_1, \ldots, e_\qnum$ {\em are known beforehand}, the reduction uses the mentioned algorithm to precompute the KZG openings $B_i$ for the evaluation $f(e_i)$  and $B_i'$ for the evaluation $g(e_i)$. For all $i \in [\qnum]$, the following equations are then satisfied by the precomputed values  \begin{align*}
		\pairingOp{\overline{G}_1 - f(e_i) G_1}{G_2} &= \pairingOp{B_i}{X_{2,1} - e_i G_2} \;, \\
		\pairingOp{\overline H - g(e_i) G_1}{G_2} &= \pairingOp{B_i'}{X_{2,1} - e_i G_2}
	\end{align*}
	
	\item Since it is ensured that $f(e_i) + m_i g(e_i) = 0$ by how the reduction is structured, we have that $A_i = B_i + m_i B_i'$ satisfies \begin{align*}
		\pairingOp{A_i}{X_{2,1} - e_i G_2} &= \pairingOp{B_i + m_i B_i'}{X_{2,1} - e_i G_2} \\
		&= \pairingOp{\overline{G}_1 + m_i \overline H -( f(e_i) + m_i g(e_i))G_1}{G_2}\\
		&= \pairingOp{\overline{G}_1 + m_i \overline H }{G_2}\;.
	\end{align*}
\end{itemize}
It is easy to see that the above sketched algorithm's runtime is dominated by the second step.
This optimization directly applies to all prior reductions due to the tags being known beforehand. 
In the case of the standard model reduction in \cite{EC:TesZhu23b}, $\qnum - 1$ of the tags $e_i$ are already known beforehand (meaning the precomputation can be done for those tags). For the only query where $e_i$ is not known beforehand, naively computing $A_i$ incurs only additive $O(\qnum (T_{\Group} + T_p))$ running time and does not change the asymptotic runtime.

This optimization, unfortunately, does not apply to our tight reduction as it does not know most of the $e_i$'s in advance. In particular, we can precompute the $A_i$'s for when $e_i = \widetilde{e}_i$, but not for $e_i$'s derived from the unique zero. More importantly, because our analysis shows that the number of indices $i$ such that $e_i \neq \widetilde{e}_i$ is on average $\approx \qnum / e$, this optimization does not quite apply. Still, we stress again that this optimization only works for certain classes of bilinear groups.

%% file: tex_files/polynomial-properties.tex
\section{Probability Toolkit for the Tight Reductions}\label{sec:random-functions-poly}
In this section, we introduce a toolkit for statistical properties that are heavily used by the analysis of our tight reduction in \Cref{sec:full-proof}. In particular, we consider certain properties regarding number of preimages for arbitrarily fixed elements of random functions (in \Cref{sec:set-S-properties-RF}) and relate them to similar properties of random high-degree polynomials (in \Cref{sec:set-S-rand-poly}). We suggest that the readers first go through the proof in \Cref{sec:full-proof} to see how these derived properties are relevant to the analysis and refer back to this section when they are invoked.

\subsection{Probability Toolkit from Random Functions}\label{sec:set-S-properties-RF}
In this section, we consider the distribution of a certain set $S$ sampled with respect to certain properties of a random function $f: \Zp \to \Zp$. 
Looking ahead in \Cref{sec:full-proof}, we will augment the $\SUF'$ game's transcript with the set $S$ independently sampled from this distribution. Then, we show that the distribution of this augmented transcript is {\em indistinguishable} from the transcript between our reduction and the adversary. There, $S$ is defined according to the reduction strategy using a random polynomial, of which the relevant probabilities are considered in \Cref{sec:set-S-rand-poly}. Importantly, our security proof {\em will not explicitly} consider random functions but only refer to the properties of $S$ derived in this section.

For any distinct values $z_0, z_1, \ldots, z_\qnum \in \Zp$ and $\actualX \in \Zp$, we will consider the distribution of the set $S \subseteq [\qnum]$ defined by the following experiment: \begin{enumerate}[noitemsep,topsep=0pt]
	\item Sample a random function $f : \Zp \to \Zp$ conditioned on $f(\actualX) = z_0$. 
	\item Set $S := \{i \in [\qnum] : |f^{-1}(z_i)| = 1\}$.
\end{enumerate} 
Here, the extra conditioning comes from the fact that our reduction will embed $f(\actualX)$ in one of the generators with $\actualX$ being the discrete logarithm of the SDH instance. This further constrains the distribution of $S$ defined in the reduction. The following lemma\ifnum\eprintversion=1 , proved in \Cref{sec:dist-S-random-f},\fi \ describes the probability density function of $S$, denoted $\eta$.

\begin{lemma}\label{lem:dist-S-random-f}
	Let $\calF$ be the family of all functions $\Zp \to \Zp$, $z_0, z_1, \ldots, z_\qnum \in \Zp$ be any distinct values and $\actualX \in \Zp$. For any $\actualS \subseteq [\qnum]$, define  \begin{equation}
	\eta(\actualS) := \Pr_{f \getsr \calF}[(\forall i \in \actualS : |f^{-1}(z_i)| = 1) \eAnd (\forall i \notin \actualS : |f^{-1}(z_i)| \neq 1)~|~f(\actualX) = z_0]\;. \label{eq:def-eta-S}
	\end{equation} Accordingly, $\sum_{\actualS \subseteq [\qnum]} \eta(\actualS) = 1$ and $\eta(\actualS) \geq 0$ for all $\actualS \subseteq [\qnum]$. Moreover, \[
	\eta(\actualS) = \sum_{i = 0}^{\qnum - |\actualS|} (-1)^i \binom{\qnum - |\actualS|}{i} \Error(1, |\actualS| + i; p)\left(1 - \frac{|\actualS| + i}{p}\right)^{p - |\actualS| - i - 1}\;.
	\]
\end{lemma}
We write the probability as $\eta(\actualS)$ since its value is independence of the choices of $z_0, z_1, \ldots, z_\qnum$ and $\actualX$. This means that for any $z_0, z_1, \ldots, z_\qnum$ distinct and $\actualX$, if one sample the set $S$ according to the above experiment, the distribution of $S$ is independent of the choices of $z_i$'s and $\actualX$. 
The proof given in \ifnum\eprintversion=1 \Cref{sec:dist-S-random-f}\else the full version\fi \ is roughly as follows: (1) observe that $\forall i \notin \actualS : |f^{-1}(z_i)| \neq 1$ is the negation of $\exists i \notin \actualS : |f^{-1}(z_i)| = 1$, (2) break the probability calculation down via the inclusion-exclusion principle, and (3) this finally boils down to computing the probability for some $T \subseteq [\qnum]$ that $\forall i \in T : |f^{-1}(z_i)| = 1$.

Next, the following corollaries establish lower bounds on specific probabilities regarding a sampled set $S$. The first corollary will be used to show that our reduction has a constant factor loss of $e$ in security. The second corollary will be used in \Cref{sec:size-f} to relate the distribution of $S$ defined with respect to random polynomials as in our reduction to the distribution of $S$ sampled according to $\eta$. 
\ifnum\eprintversion=0
The proofs are given in the full version.
\fi
\begin{corollary}\label{corol:prob-contain-index}
	For any $i^* \in [\qnum]$, $\Pr_{S \getsr \eta}[i^* \in S] \geq e^{-1}$.
\end{corollary}
\ifnum\eprintversion=1
\begin{proof}
	Let $z_0, z_1, \ldots, z_\qnum \in \Zp$ be distinct and $\actualX \in \Zp$ as in \Cref{lem:dist-S-random-f}.
	Consider 
	\begin{align*}
		 \Pr_{S \getsr \eta}[i^* \in S] &=  \ifnum\eprintversion=0 \textstyle \fi \Pr[|f^{-1}(z_{i^*})| = 1~|~f(\actualX) = z_0]= \left(1 - \frac{1}{p}\right)^{p - 1} \geq e^{-1}\;.
	\end{align*} 
	The first equality follows from definition of $S$ as the set of indices $i \in [\qnum]$ where $z_{i}$ has only one preimage. The second follows from a counting argument: (1) pick a preimage of $z_{i^*}$ that is not $\actualX$ and (2) pick images of other elements in the domain. The last inequality follows from \Cref{lem:lower-bound-expo} setting $k = l = 1$. \qed
\end{proof}
\fi
\begin{corollary}\label{corol:lower-bound-eta}
	For any $\actualS \subseteq [\qnum]$, $\eta(\actualS) \geq \left(1- \frac{\qnum}{p}\right)^{p - |\actualS| - 1} \cdot \Error(1 , |\actualS|; p)$.
\end{corollary}
\ifnum\eprintversion=1
\begin{proof}
	First, we recall the definition of $\eta$ 
	\ifnum\eprintversion=1
	\begin{align*}
		\eta(\actualS) &:= \Pr_{f \getsr \calF}[(\forall i \in \actualS : |f^{-1}(z_i)| = 1) \eAnd (\forall i \notin \actualS : |f^{-1}(z_i)| \neq 1)~|~f(\actualX) = z_0]\\
		& \geq \Pr[(\forall i \in \actualS : |f^{-1}(z_i)| = 1) \eAnd {\color{blue}(\forall i \notin \actualS : |f^{-1}(z_i)| = 0)}~|~f(\actualX) = z_0]\\
		& = \left(1- \frac{\qnum}{p}\right)^{p - |\actualS| - 1} \cdot \Error(1 , |\actualS|; p)\;.
	\end{align*}
	\else
	\begin{align*}
		\eta(\actualS) &:= \Pr_{f \getsr \calF}[(\forall i \in \actualS : |f^{-1}(z_i)| = 1) \eAnd (\forall i \notin \actualS : |f^{-1}(z_i)| \neq 1)~|~f(\actualX) = z_0]\\
		& \geq \Pr[(\forall i \in \actualS : |f^{-1}(z_i)| = 1) \eAnd {\color{blue}(\forall i \notin \actualS : |f^{-1}(z_i)| = 0)}~|~f(\actualX) = z_0]\\
		& \textstyle = \left(1- \frac{\qnum}{p}\right)^{p - |\actualS| - 1} \cdot \Error(1 , |\actualS|; p)\;.
	\end{align*}
	\fi
	The first inequality follows from $|f^{-1}(z_i)| = 0$ implying $|f^{-1}(z_i)| \neq 1$. The last equality follows from the following counting argument and dividing by $p^{p - 1}$ (i.e., the number of $f \in \calF$ such that $f(\actualX) = z_0$): 
	\begin{itemize}[topsep=0pt]
		\item {\em Count number of ways to assign unique preimages of $z_i$:} For each $i \in \actualS$, we can assign distinct $e_i \neq \actualX$ such that $f(e_i) = z_i$. In particular, this corresponds to selecting a vector $(e_i)_{i \in \actualS}$ with distinct elements, i.e., there are $(p - 1)! / (p - |\actualS| - 1)!$ such vectors.
		
		\item {\em Count number of ways to assign images to the non-selected elements in $\Zp$:} For each $e \in \Zp \setminus (\{\actualX\} \cup \{e_i\}_{i \in \actualS})$, we select $f(e)$ to be any value in $\Zp$ except for $\{z_i\}_{i \in [\qnum]}$, since we already fixed all the preimages of these values. Hence, there are exactly $(p - \qnum)^{p - |\actualS| - 1}$ ways to assign these images.  \qed
	\end{itemize}
\end{proof}
\fi

\subsection{Probability Toolkit from Random Polynomials}\label{sec:set-S-rand-poly}
In this section, we consider the distribution of the set $S \subseteq [\qnum]$ sampled according to our reduction strategy in \Cref{sec:full-proof} based on a random polynomial of large enough degree $d$. The following lemma, which is proved in \Cref{sec:size-f} and is the key to our reduction's analysis, establishes the probability that a random monic polynomial $f$ of degree $d = \Theta(\qnum)$ satisfies the following constraints for any distinct $z_0, z_1, \ldots, z_\qnum$, a set $\actualS \subseteq [\qnum]$, and distinct $\actualX, (\actualE_i)_{i \in \actualS}$: (a) for $i \notin \actualS$, the number of preimages 
of $z_i$ with respect to the function $f$ is not $1$, and (b) for $i \in \actualS$, $z_i$ has $\actualE_i$ as its only preimage. Note that $f^{-1}(z)= \Root(f(\varX) - z)$, of which the notation is used throughout the paper. 
As a connection to \Cref{sec:set-S-properties-RF}, the lemma shows the probability lower bound of sampling an $f$ that leads to a particular set $\actualS$ {\em with extra constraints on specific preimages of $z_i$ for $i \in \actualS$} (hence, the $p^{-|\actualS|}$ factor). This relates the probability of sampling a specific transcript based on the reduction's simulation and the transcript in the $\SUF'$ game.

\begin{lemma}\label{lemma:size-f}
	Let $d = \lceil (e^2 + 1) \qnum + 2 \log_2 p \rceil + 1$, $p \geq 2\qnum^2$, $\actualS \subseteq [\qnum]$, $z_0, z_1, \ldots, z_\qnum \in \Zp$ be distinct values,  and $\actualX \in \Zp, (\actualE_i)_{i \in \actualS} \in \Zp^{|\actualS|}$ also be distinct values. Then, with $\varepsilon = e^{-(e^2 - 3) \qnum - \log_2 p + 1} < p^{-1}$, \[
	\Pr_{f \getsr \Zp^d[\varX]}\left[\begin{array}{l}f(\actualX) = z_0 \eAnd \forall i \notin \actualS: |f^{-1}(z_i)| \neq 1\\
		\eAnd \forall i \in \actualS: f^{-1}(z_i) =  \{\actualE_i\}
	\end{array} \right] \geq (1 - \varepsilon) \cdot p^{- 1 - |\actualS|} \eta(\actualS) \;.
	\]
\end{lemma}

We stress that the lemma  {\em does not immediately follow} from the fact that a random degree-$d$ polynomial $f$ is a $d$-wise independent function (i.e., for $d$ distinct points, the evaluations of $f$ is uniformly random). Indeed, the particular event we consider is {\em global}, in that reasoning about the (size of) set of preimages introduce constraints to {\em all values} in the domain. For example, fixing $f$ to have only $e \in \Zp$ mapping to $z \in \Zp$ means that {\em for all} $e' \neq e$, we require $f(e') \neq z$. 

If $f$ is a uniformly random function $\Zp \to \Zp$ instead of degree-$d$ polynomials, the corresponding probability as defined in the lemma is exactly $p^{- |\actualS| - 1} \cdot  \Error(1, |\actualS|; p)^{-1} \eta(\actualS)$.\footnote{The error term stems from the additional constraints $f(\actualE_i) = z_i$ for all $i \in \actualS$.} For random degree-$d$ polynomial $f$, we can also show that the probability is within $(1 \pm \varepsilon)$ factor of this mentioned value, but we only state the relevant lower bound required by our security proof later on.

\ifnum\eprintversion=1
\input{tex_files/random-function-proofs}
\fi
\input{tex_files/key-lemma-proofs}

%% file: tex_files/random-function-proofs.tex
\subsection{Proof of \Cref{lem:dist-S-random-f}}\label{sec:dist-S-random-f}
As an overview, we will simplify the calculation of the probability by defining placeholder events and apply the inclusion-exclusion principle so that the events can be easily analyzed.
For readability, we first define the events $B_i := |f^{-1}(z_i)| = 1$ for $ i \in [\qnum]$, and $A := \bigcap_{i \in \actualS} B_i$. Then, we can write the the event defining $\eta(\actualS)$ as: 
\ifnum\eprintversion=0
\begin{align*}
	(\forall i \in \actualS : |f^{-1}(z_i)| = 1) \eAnd (\forall i \notin \actualS : |f^{-1}(z_i)| \neq 1) 
	&= A \eAnd \left(\bigcap_{i \in [\qnum] \setminus \actualS} \neg B_i \right)\\
	&= A \eAnd \neg \left(\bigcup_{i \in [\qnum] \setminus \actualS} B_i \right)\;.
\end{align*}
\else
\begin{align*}
	(\forall i \in \actualS : |f^{-1}(z_i)| = 1) \eAnd (\forall i \notin \actualS : |f^{-1}(z_i)| \neq 1) 
	&= A \eAnd \left(\bigcap_{i \in [\qnum] \setminus \actualS} \neg B_i \right)= A \eAnd \neg \left(\bigcup_{i \in [\qnum] \setminus \actualS} B_i \right)\;.
\end{align*}
\fi
The first equality follows from $\neg B_i = |f^{-1}(z_i)| \neq 1$. 

Hence, we have that \begin{align*}
	\eta(\actualS) &= \Pr_{f \getsr \calF} \left[A \eAnd \neg \left(\bigcup_{i \in [\qnum] \setminus \actualS} B_i \right)~\middle|~f(\actualX) = z_0 \right] \\ 
	&= \Pr[A~|~f(\actualX) = z_0] - \Pr\left[A \eAnd \left(\bigcup_{i \in [\qnum] \setminus \actualS} B_i \right)~\middle|~f(\actualX) = z_0 \right]\\
	&= \Pr[A~|~f(\actualX) = z_0] - \Pr\left[\bigcup_{i \in [\qnum] \setminus \actualS} (A \eAnd B_i) ~\middle|~f(\actualX) = z_0 \right]\;.
\end{align*}
Next, by inclusion-exclusion principle, we can compute the probability of the event $\bigcup_{i \in [\qnum] \setminus \actualS} (A \eAnd B_i)$, as the following alternating sum 
\ifnum\eprintversion=0
\begin{align*}
	&\Pr\left[\bigcup_{i \in [\qnum] \setminus \actualS} (A \eAnd B_i)~\middle|~f(\actualX) = z_0 \right] \\
	&= \sum_{\emptyset \neq T \subseteq [\qnum] \setminus \actualS} (-1)^{|T| - 1} \Pr\left[\bigcap_{i \in T} (A \eAnd B_i)~\middle|~f(\actualX) = z_0\right] \\
	&= \sum_{\emptyset \neq T \subseteq [\qnum] \setminus \actualS} (-1)^{|T| - 1} \Pr\left[A \eAnd \bigcap_{i \in T}  B_i~\middle|~f(\actualX) = z_0\right]\;.
\end{align*}
\else
\begin{align*}
	\Pr\left[\bigcup_{i \in [\qnum] \setminus \actualS} (A \eAnd B_i)~\middle|~f(\actualX) = z_0 \right] 
	&= \sum_{\emptyset \neq T \subseteq [\qnum] \setminus \actualS} (-1)^{|T| - 1} \Pr\left[\bigcap_{i \in T} (A \eAnd B_i)~\middle|~f(\actualX) = z_0\right] \\
	&= \sum_{\emptyset \neq T \subseteq [\qnum] \setminus \actualS} (-1)^{|T| - 1} \Pr\left[A \eAnd \bigcap_{i \in T}  B_i~\middle|~f(\actualX) = z_0\right]\;.
\end{align*}
\fi
Note that $A = A \eAnd \bigcap_{i \in \emptyset}  B_i$, meaning we can also write $\eta(\actualS)$ as 
\ifnum\eprintversion=0
\begin{align*}
	\eta(\actualS) &= \Pr\left[ A \eAnd \bigcap_{i \in \emptyset}  B_i \middle | ~f(\actualX) = z_0\right] \\
	&\pcind[2] - \sum_{\emptyset \neq T \subseteq [\qnum] \setminus \actualS} (-1)^{|T| - 1} \Pr\left[A \eAnd \bigcap_{i \in T}  B_i~\middle|~f(\actualX) = z_0\right] \\
	&= \sum_{T \subseteq [\qnum] \setminus \actualS} (-1)^{|T|} \Pr\left[A \eAnd \bigcap_{i \in T}  B_i~\middle|~f(\actualX) = z_0\right]\;.
\end{align*}
\else
\begin{align*}
	\eta(\actualS) &= \Pr\left[ A \eAnd \bigcap_{i \in \emptyset}  B_i \middle | ~f(\actualX) = z_0\right] - \sum_{\emptyset \neq T \subseteq [\qnum] \setminus \actualS} (-1)^{|T| - 1} \Pr\left[A \eAnd \bigcap_{i \in T}  B_i~\middle|~f(\actualX) = z_0\right] \\
	&= \sum_{T \subseteq [\qnum] \setminus \actualS} (-1)^{|T|} \Pr\left[A \eAnd \bigcap_{i \in T}  B_i~\middle|~f(\actualX) = z_0\right]\;.
\end{align*}
\fi
The second equality follows from $-(-1)^{|T| - 1} = (-1)^{|T|}$. 
Now, we state the following lemma which determines the exact probability of $A \eAnd \bigcap_{i \in T}  B_i$. The proof is in \Cref{sec:all-one-primage-subset}.

\begin{lemma}\label{lem:all-one-primage-subset}
	For any distinct $z_0, z_1, \ldots, z_\qnum \in \Zp$, $\actualX \in \Zp$, and $S' \subseteq [\qnum]$, let $l = |S'| \leq \qnum$, then 
	\[
	\Pr_{f \getsr \calF}\left[\forall i \in S': |f^{-1}(z_i)| = 1~\middle|~f(\actualX) = z_0\right] = \left(1 - \frac{l}{p}\right)^{p - l - 1} \cdot \Error(1, l; p)\;.
	\]
\end{lemma}

By substituting in the terms and setting $S' = \actualS \cup T$ to the lemma statement, we have that \begin{align}
	\Pr\left[A \eAnd \bigcap_{i \in T}  B_i~\middle|~f(\actualX) = z_0\right] &= \Pr[\forall i \in \actualS \cup T: |f^{-1}(z_i)| = 1~|~f(\actualX) = z_0] \nonumber \\
	&=  \left(1 - \frac{|\actualS| + |T|}{p}\right)^{p - |\actualS| - |T| - 1} \cdot \Error(1,|\actualS| + |T| ; p) \label{eq:Prob-A-and-Bi}\;.
\end{align}
Then, notice that $\Pr[A \eAnd \bigcap_{i \in T}  B_i~|~f(\actualX) = z_0]$ only depends the size of $\actualS$ and $T$, meaning we can write \begin{align*}
	\eta(\actualS)&= \sum_{T \subseteq [\qnum] \setminus \actualS} (-1)^{|T|} \Pr\left[A \eAnd \bigcap_{i \in T}  B_i~\middle|~f(\actualX) = z_0\right] \\
	&= \sum_{T \subseteq [\qnum] \setminus \actualS} (-1)^{|T|} \left(1 - \frac{|\actualS| + |T|}{p}\right)^{p - |\actualS| - |T| - 1} \cdot \Error(1,|\actualS| + |T| ; p) \\
	&= \sum_{i = 0}^{\qnum - |\actualS|} (-1)^{i} \binom{\qnum - |\actualS|}{i} \left(1 - \frac{|\actualS| + i}{p}\right)^{p - |\actualS|- i - 1} \Error(1 ,|\actualS| + i; p)\;.
\end{align*}
The second equality follows from \Cref{eq:Prob-A-and-Bi}.
The third equality follows by replacing each size $|T|$ with $i$ and counting all subsets $T \subseteq [\qnum] \setminus \actualS$ of size $i \geq 1$.   \qed

	\subsection{Proof of \Cref{lem:all-one-primage-subset}} \label{sec:all-one-primage-subset}
	We note first that there are exactly $p^{p - 1}$ functions $f \in \calF$ such that $f(\actualX) = z_0$. Then, we will now count the number of functions $f \in \calF$ such that for all $i \in S'$, $|f^{-1}(z_i)| = 1$, which is done in the following steps: \begin{enumerate}
		\item {\em Count number of ways to assign unique preimages of $z_i$:} For each $i \in S'$, we can assign distinct $e_i \neq \actualX$ such that $f(e_i) = z_i$. In particular, this correspond to selecting a vector $(e_i)_{i \in S'}$ of length $l = |S'|$ with distinct elements, i.e., there are $(p - 1)! / (p - l - 1)!$ such vectors.
		
		\item {\em Count number of ways to assign images to the non-selected elements in $\Zp$:} For each $e \in \Zp \setminus (\{\actualX\} \cup \{e_i\}_{i \in S'})$, we can select $f(e)$ to be any value in $\Zp$ except for $\{z_i\}_{i \in S'}$, since we constrained the number of preimages for these values to $1$. Hence, there are exactly $(p - l)^{p - l - 1}$ ways to assign these images. 
	\end{enumerate}
	Combining the two steps and dividing by $p^{p - 1}$, we have that \begin{align*}
		\Pr_{f \getsr \calF}\left[\forall i \in S': |f^{-1}(z_i)| = 1~\middle|~f(\actualX) = z_0\right] &= p^{- p + 1} \cdot \frac{(p - 1)!}{(p - l - 1)!} \cdot (p - l)^{p - l - 1} \\
		&= \left(1 - \frac{l}{p}\right)^{p - l - 1} \cdot \Error(1, l; p)\;.
	\end{align*} The equality follows from the definition of $\Error(1 , l ; p) = p^{- l} \cdot \frac{(p - 1)!}{(p - l - 1)!}$. \qed

%% file: tex_files/key-lemma-proofs.tex
\subsection{Proof of \Cref{lemma:size-f}}\label{sec:size-f}
Let $(z_0, \vec{z}, \actualS, \actualX, (\actualE_i)_{i \in \actualS})$ be as in the lemma statement with $S \subseteq [\qnum]$, $z_0, \ldots, z_\qnum$ distinct, and $\actualX, (\actualE_i)_{i \in \actualS}$ distinct.
To prove the lemma and for readability, we first define for each $i \in [\qnum]$, the events $D$ and $E_i$ as 
\ifnum\eprintversion=1
\begin{align*}
	D &:=  f(\actualX) = z_0 \eAnd
	(\forall i \in \actualS: f^{-1}(z_i) = \{\actualE_i\})
	;, 
	&E_i := |f^{-1}(z_i)| = 1\;.
\end{align*}
\else
\begin{align*}
	D &:=  f(\actualX) = z_0 \eAnd
	\forall i \in \actualS: f^{-1}(z_i) = \{\actualE_i\}\;, \;
	&E_i := |f^{-1}(z_i)| = 1\;.
\end{align*}
\fi
Then, we rewrite our desired probability (which we will denote as $F$) as  
\ifnum\eprintversion=0
\begin{align*}
F&:=\Pr_{f \getsr \Zp^d[\varX]}\left[\begin{array}{l}f(\actualX) = z_0 \eAnd \forall i \notin \actualS: |f^{-1}(z_i)| \neq 1\\
	\eAnd \forall i \in \actualS: f^{-1}(z_i) =  \{\actualE_i\}
\end{array} \right] \textstyle = \Pr\left[D \mathbin{\bigg\backslash} \bigcup_{j \in [\qnum] \setminus \actualS} E_i\right] \\
&= \textstyle \Pr\left[D  \mathbin{\bigg\backslash} \bigcup_{j \in [\qnum] \setminus \actualS} (D \eAnd E_j) \right] = \Pr[D] - \Pr\left[\bigcup_{j \in [\qnum] \setminus \actualS} (D \eAnd E_i) \right] \;.
\end{align*}
\else
\begin{align*}
	F&:=\Pr_{f \getsr \Zp^d[\varX]}\left[\begin{array}{l}f(\actualX) = z_0 \eAnd \forall i \notin \actualS: |f^{-1}(z_i)| \neq 1\\
		\eAnd \forall i \in \actualS: f^{-1}(z_i) =  \{\actualE_i\}
	\end{array} \right] \\
	& = \Pr\left[D \mathbin{\bigg\backslash} \bigcup_{j \in [\qnum] \setminus \actualS} E_i\right] \\
	&= \Pr\left[D  \mathbin{\bigg\backslash} \bigcup_{j \in [\qnum] \setminus \actualS} (D \eAnd E_j) \right] \\
	&= \Pr[D] - \Pr\left[\bigcup_{j \in [\qnum] \setminus \actualS} (D \eAnd E_i) \right] \;.
\end{align*}
\fi
Next, we apply the principle of inclusion-exclusion \ifnum\eprintversion=1(\Cref{lemma:inclusion-exclusion-bound})\fi to compute the probability of $\bigcup_{j \in [\qnum] \setminus \actualS} (D \eAnd E_i) $ as \[
\textstyle \Pr\left[\bigcup_{j \in [\qnum] \setminus \actualS} (D \eAnd E_i) \right] = \sum_{\emptyset \neq T \subseteq  [\qnum] \setminus \actualS} (-1)^{|T| + 1} \Pr\left[\bigcap_{j \in T} (D \eAnd E_j) \right]\;.
\]
Also, notice that $D = \bigcap_{j \in \emptyset} (D \eAnd E_j)$, meaning 
\ifnum\eprintversion=1
\begin{equation}
F =  \sum_{T \subseteq  [\qnum] \setminus \actualS} (-1)^{|T|} \Pr\left[\bigcap_{j \in T} (D \eAnd E_j) \right]\;.\label{eq:prob-F}
\end{equation}
\else
\begin{equation}
	\textstyle F =  \sum_{T \subseteq  [\qnum] \setminus \actualS} (-1)^{|T|} \Pr\left[\bigcap_{j \in T} (D \eAnd E_j) \right]\;.\label{eq:prob-F}
\end{equation}
\fi
Note that by definition, we can write the event $\bigcap_{j \in T} (D \eAnd E_j)$ as \ifnum\eprintversion=0
\begin{align*}
	\textstyle \bigcap_{j \in T} (D \eAnd E_j) = D \eAnd \bigcap_{j \in T}  E_j = \left( 
	\begin{array}{l}
		f(\actualX) = z_0 \eAnd 
		\forall j \in T:  |f^{-1}(z_j)| = 1 \eAnd\\
		\forall j \in \actualS: f^{-1}(z_j) = \{\actualE_j\}
	\end{array} \right)\;.
\end{align*}
\else
\begin{align*}
	\bigcap_{j \in T} (D \eAnd E_j) = D \eAnd \bigcap_{j \in T}  E_j = \left(  
	f(\actualX) = z_0 \eAnd 
	\forall j \in T:  |f^{-1}(z_j)| = 1 \eAnd
	\forall j \in \actualS: f^{-1}(z_j) = \{\actualE_j\} 
	\right)\;.
\end{align*}
\fi
We now compute $\Pr[D \eAnd \bigcap_{j \in T} E_j]$. We observe that this event can be partitioned into $p^{|T|}$ disjoint events, $G_{T,  (e_i')_{i \in T}}$ for $(e_i')_{i \in T} \in \Zp^{|T|}$ defined as 
\ifnum\eprintversion=0
\[
	 G_{T,  (e_i')_{i \in T}} := \left( \begin{array}{l}
		f(\actualX) = z_0 \eAnd  
		\forall j \in T:  f^{-1}(z_j) = \{e'_j\} \eAnd  
		\forall j \in \actualS: f^{-1}(z_j) = \{\actualE_j\}
	\end{array} \right)\;.
\]
\else
\[
G_{T,  (e_i')_{i \in T}} := \left( \begin{array}{l}
	f(\actualX) = z_0 \eAnd 
	\forall j \in T:  f^{-1}(z_j) = \{e'_j\} \eAnd 
	\forall j \in \actualS: f^{-1}(z_j) = \{\actualE_j\}
\end{array} \right)\;.
\]
\fi
Note that the event $G_{T,  (e_i')_{i \in T}}$ partitions $D \eAnd \bigcap_{j \in T} E_j$ because (1) by definition, all $G_{T,  (e_i')_{i \in T}}$ is a subset of (i.e., implies) $D  \eAnd \bigcap_{j \in T} E_j$, (2) the union of all $G_{T,  (e_i')_{i \in T}}$ contains $D  \eAnd \bigcap_{j \in T} E_j$ since each of the sets explains what the only zero of $f(\varX) - z_j = 0$ could be, and (3) $G_{T,  (e_i')_{i \in T}}$ and $G_{T,  (e_i'')_{i \in T}}$ for $(e_i')_{i \in T} \neq (e_i'')_{i \in T}$ are disjoint as no function can have two ``unique preimages"--i.e., $f^{-1}(z_j) = \{e'_j\} = \{e''_j\}$ for $e'_j \neq e''_j$ for some $j \in T$, which is clearly a contradiction.

\ifnum\eprintversion=1
Also, for $(e_i' \in \Zp)_{i \in T}$, $\Pr[G_{T, (e_i')_{i \in T}}] = 0$ if one of the following is true: \begin{enumerate}[noitemsep,topsep=0pt,label=(\alph*)]
	\item For some $i \in T, j \in \actualS$, $e_i' = \actualE_j$ {\em or} 
	\item For some $i \in T, \actualX = e_i'$, {\em or}
	\item For some $i\neq j \in T$,  $e_i' = e'_j$.
\end{enumerate} 
\else
Also, we claim that $\Pr[G_{T, (e_i')_{i \in T}}] = 0$ for $(e_i')_{i \in T}$ that satisfies one of the following is true: (a) For some $i \in T, j \in \actualS$, $e_i' = \actualE_j$ {\em or}, (b) For some $i \in T, \actualX = e_i'$, {\em or}, (c) For some $i\neq j \in T$,  $e_i' = e'_j$.

\fi
The reasoning is as follows. If (a) occurs, $z_i = f(e_i') = f(\actualE_j) = z_j$, a contradiction as $z_i \neq z_j$. If (b) occurs, $z_0 = f(\actualX) = f(e_i') = z_i$, also a contradiction since $z_i \neq z_0$. If (c) occurs, $z_i  = z_j$, which is again a contradiction.

\ifnum\eprintversion=1
Hence, we can write \[
\Pr\left[\bigcap_{j \in T} (D \eAnd E_j)\right] = \sum_{(e_i')_{i \in T} \text{ distinct and not overlap with } \{\actualE_i\}_{i \in \actualS} \cup \{\actualX\}} \Pr[G_{T, (e_i')_{i \in T}}]\;.
\]
\else
Hence, we write $
\Pr\left[\bigcap_{j \in T} (D \eAnd E_j)\right] = \sum_{\substack{(e_i')_{i \in T} \text{ distinct and not} \\ \text{ overlap with } \{\actualE_i\}_{i \in \actualS} \cup \{\actualX\}}} \Pr[G_{T, (e_i')_{i \in T}}]$.
\fi
We will now bound the probability of the event $G_{T, (e_i')_{i \in T}}$ via the following lemma, proved in \ifnum\eprintversion=1 \Cref{sec:key-lemma-size}\else the full version\fi. In particular, it is a special case of \Cref{lemma:size-f} where we do not consider the constraints that for $i \notin S$, the number of preimages of $z_i$ is not $1$. 
We note that the term $p^{-|S| - 1}\left(1 - \frac{|S|}{p}\right)^{p - |S| - 1}$ is exactly the probability over a random function $f$ such that $f(\actualX) = z_0$ and for each $j \in S$, $z_j$ has only one preimage being $\actualE_j$\ifnum\eprintversion=1 ---drawing comparison to \Cref{lem:all-one-primage-subset} (with extra constraints on which elements are the preimages)\fi. 
The proof idea is that in the binomial expansion of $\left(1 - \frac{|S|}{p}\right)^{p - |S| - 1}$, the terms after a large enough $d$ will be {\em very small relative to the overall value}; hence, the multiplicative error $\varepsilon$. The formal proof follows similar (but more complex) counting argument as in~\Cref{sec:tech-overview} relying on the inclusion-exclusion principle.

\begin{lemma}\label{lemma:key-lemma-size}
	Let $d = \lceil (e^2 + 1) \qnum + 2 \log_2 p \rceil + 1$, $S \subseteq [\qnum]$, 
	$\actualX, (\actualE_i)_{i \in S}$ be distinct values, and $z_0, z_1,\ldots,z_\qnum$ be distinct values. 
	Then, with $\varepsilon = e^{-(e^2 - 1) \qnum - \log_2 p}$, 
	 \ifnum\eprintversion=0 
	 \begin{align*}
		\Pr_{f \getsr \Zp^{d}[\varX]}\left[ \begin{array}{l}
			f(\actualX) = z_0 \; \eAnd\\ 
			\forall j \in S: f^{-1}(z_j) = \{\actualE_j\}
		\end{array} \right] \in [1 \pm \varepsilon] \cdot p^{- |S| -1 }\left(1 - \frac{|S|}{p}\right)^{p - |S| - 1} \;.
	\end{align*} 
	\else 
	\begin{align*}
		\Pr_{f \getsr \Zp^{d}[\varX]}\left[ \begin{array}{l}
			f(\actualX) = z_0 \; \eAnd
			\forall j \in S: f^{-1}(z_j) = \{\actualE_j\}
		\end{array} \right] \in [1 \pm \varepsilon] \cdot p^{- |S| -1 }\left(1 - \frac{|S|}{p}\right)^{p - |S| - 1} \;.
	\end{align*} 
	\fi
\end{lemma}
By substituting  $S = \actualS \cup T$, $\actualE_i = \actualE_i$ for $i \in \actualS$ and $\actualE_i = e'_i$ for $i \in T$, we have that 
\ifnum\eprintversion=1
\[
	\Pr[G_{T, (e_i')_{i \in T}}] \in [1 \pm \varepsilon] \cdot p^{- |\actualS| - |T| -1 }\left(1 - \frac{|\actualS| + |T|}{p}\right)^{p - |\actualS| - |T| - 1} \;.
	\] \else
	$	\Pr[G_{T, (e_i')_{i \in T}}] \in [1 \pm \varepsilon] \cdot p^{- |\actualS| - |T| -1 }\left(1 - \frac{|\actualS| + |T|}{p}\right)^{p - |\actualS| - |T| - 1}$.
	\fi
	Hence, we bound 
	\ifnum\eprintversion=1
	\begin{align*}
	\Pr\left[\bigcap_{j \in T} (D \eAnd E_j)\right] &\in  [1 \pm \varepsilon] \cdot \Error(|\actualS| + 1, |T| ; p) p^{|T|} \cdot p^{- |\actualS| - |T| -1 }\left(1 - \frac{|\actualS| + |T|}{p}\right)^{p - |\actualS| - |T| - 1} \\
	&=  [1 \pm \varepsilon] \cdot \Error(|\actualS| + 1, |T| ; p) \cdot p^{- |\actualS| -1 }\left(1 - \frac{|\actualS| + |T|}{p}\right)^{p - |\actualS| - |T| - 1} \;, 
	\end{align*} 
	\else
	\begin{align*}
		&\textstyle\Pr[\bigcap_{j \in T} (D \eAnd E_j)] \\
		&\in  [1 \pm \varepsilon] \cdot \Error(|\actualS| + 1, |T| ; p) p^{|T|} \cdot p^{- |\actualS| - |T| -1 }\left(1 - \frac{|\actualS| + |T|}{p}\right)^{p - |\actualS| - |T| - 1} \\
		&=  [1 \pm \varepsilon] \cdot \Error(|\actualS| + 1, |T| ; p) \cdot p^{- |\actualS| -1 }\left(1 - \frac{|\actualS| + |T|}{p}\right)^{p - |\actualS| - |T| - 1} \;, 
	\end{align*} 
	\fi
	where the multiplicative term $\Error(|\actualS| + 1, |T| ; p) \cdot p^{|T|} = \frac{(p - |\actualS| - 1)!}{(p - |\actualS| - |T| - 1)!}$ is the number of distinct $(e'_i)_{i \in T} \in (\Zp \setminus (\{\actualX\} \cup \{\actualE_i\}_{i \in T}))^{|T|}$. We emphasize that this probability only depends on the size of $\actualS$ and $T$ (ignoring for the error term). 
	
	Therefore, we can (upper and lower) bound $F$ by substituting into \Cref{eq:prob-F} the derived bound as \begin{align*}
	F &\in \sum_{T \subseteq [\qnum]\setminus\actualS} (-1)^{|T|} [1 \pm \varepsilon] \cdot \Error(|\actualS| + 1, |T| ; p)  p^{- |\actualS| -1 } \left(1 - \frac{|\actualS| + |T|}{p}\right)^{p - |\actualS| - |T| - 1} \\
	&= p^{- |\actualS| - 1} \sum_{i = 0}^{\qnum - |\actualS|} (-1)^i [1 \pm \varepsilon] \binom{\qnum - |\actualS|}{i } \Error(|\actualS| + 1, i ; p) \left(1 - \frac{|\actualS| + i}{p}\right)^{p - |\actualS| - i - 1}.
\end{align*}
The second equality follows from setting $i = |T|$ and factoring out $p^{- |\actualS| - 1}$. 
Now, we lower bound $p^{|\actualS| + 1}  F $ by taking the worst case of $\pm \varepsilon$ in the alternating sum: 
\begin{align}
	p^{ |\actualS| + 1} F& \geq  \sum_{i = 0}^{\qnum - |\actualS|} (-1)^i \binom{\qnum - |\actualS|}{i } \Error(|\actualS| + 1, i ; p) \cdot \left(1 - \frac{|\actualS| + i}{p}\right)^{p - |\actualS| - i - 1} \label{eq:first-term-F} \\
	& \pcind[2] - \varepsilon \sum_{i = 0}^{\qnum - |\actualS|} \binom{\qnum - |\actualS|}{i} \Error(|\actualS| + 1, i ; p) \cdot \left(1 - \frac{|\actualS| + i}{p}\right)^{p - |\actualS| - i - 1}\label{eq:second-term-F}\;.
\end{align}
We first see that the right-hand side of  (\ref{eq:first-term-F}) is  \begin{align}
	&\sum_{i = 0}^{\qnum - |\actualS|} (-1)^i \binom{\qnum - |\actualS|}{i } \Error(|\actualS| + 1, i ; p) \cdot \left(1 - \frac{|\actualS| + i}{p}\right)^{p - |\actualS| - i - 1} \nonumber \\
	&= \Error(1, |\actualS|; p)^{-1} \cdot \sum_{i = 0}^{\qnum - |\actualS|} (-1)^i \binom{\qnum - |\actualS|}{i } \Error(1, |\actualS| + i ; p) \cdot \left(1 - \frac{|\actualS| + i}{p}\right)^{p - |\actualS| - i - 1} \nonumber \\
	&=  \Error(1, |\actualS|; p)^{-1} \cdot \eta(\actualS) \label{eq:derived-first-term}
\end{align}
The first equality follows from $\Error(|\actualS| + 1, i ; p) = \Error(1 , |\actualS|; p)^{-1}\Error(1 , |\actualS| +  i ; p)$. The second equality follows from \Cref{lem:dist-S-random-f}. 

Finally, we upper bound the term in (\ref{eq:second-term-F}), so that ultimately, we can lower bound $p^{|\actualS| + 1} \cdot F$ with $(1 - \varepsilon') \cdot \eta(\actualS) $ for $\epsilon' = e^{-(e^2 - 3) \qnum - \log_2 p + 1}$. 
\ifnum\eprintversion=1
\begin{align}
	&\varepsilon \cdot \sum_{i = 0}^{\qnum - |\actualS|} \binom{\qnum - |\actualS|}{i} \Error(|\actualS| + 1, i ; p) \cdot \left(1 - \frac{|\actualS| + i}{p}\right)^{p - |\actualS| - i - 1} \nonumber\\
	& \leq \varepsilon  \sum_{i = 0}^{\qnum - |\actualS|} \binom{\qnum - |\actualS|}{i } \nonumber\\
	&= \varepsilon 2^{\qnum - |\actualS|} \nonumber\\
	& \leq e^{-(e^2 - 1) \qnum - \log_2 p} \cdot e^{\qnum}  = e^{-(e^2 - 3) \qnum - \log_2 p + 1} e^{-\qnum - 1}\nonumber \\
	& \leq e^{-(e^2 - 3) \qnum - \log_2 p + 1} \left(1 - \frac{\qnum}{p}\right)^{p - |\actualS| - 1} \nonumber\\
	& \leq \epsilon' \cdot \Error(1, |\actualS|; p)^{-1} \cdot \eta(\actualS)\;. \label{eq:derived-second-term}
\end{align}
\else
\begin{align}
	&\varepsilon \cdot \sum_{i = 0}^{\qnum - |\actualS|} \binom{\qnum - |\actualS|}{i } \Error(|\actualS| + 1, i ; p) \cdot \left(1 - \frac{|\actualS| + i}{p}\right)^{p - |\actualS| - i - 1} \nonumber\\
	& \leq \varepsilon  \sum_{i = 0}^{\qnum - |\actualS|} \binom{\qnum - |\actualS|}{i } = \varepsilon 2^{\qnum - |\actualS|}  \leq e^{-(e^2 - 1) \qnum - \log_2 p} \cdot e^{\qnum}  = e^{-(e^2 - 3) \qnum - \log_2 p + 1} e^{-\qnum - 1}\nonumber \\
	& \leq e^{-(e^2 - 3) \qnum - \log_2 p + 1} \left(1 - \frac{\qnum}{p}\right)^{p - |\actualS| - 1} \leq \epsilon' \cdot \Error(1, |\actualS|; p)^{-1} \cdot \eta(\actualS)\;. \label{eq:derived-second-term}
\end{align}
\fi
The first inequality follows simply from $(1 - x) \leq 1$. The second equality follows from the binomial sum. 
The third inequality is subtituting $\varepsilon$ and that $2^{\qnum - |\actualS|} \leq e^{\qnum}$. The second to last inequality follows from \Cref{lem:lower-bound-expo} with $l = \qnum$ and $k = |\actualS| + 1 \leq \qnum + 1$, giving $\left(1 - \frac{\qnum}{p}\right)^{p - |\actualS| - 1} \geq e^{-\qnum\frac{p - |\actualS| - 1}{p - \qnum}} \geq e^{-\qnum - \qnum(\qnum - |\actualS| - 1) / (p - \qnum)} \geq e^{-\qnum - 1}$ (last step follows from $p \geq 2 \qnum^2$). The last inequality follows from \Cref{corol:lower-bound-eta}. 

Therefore, we have that from \Cref{eq:derived-first-term,eq:derived-second-term} \ifnum\eprintversion=1 \[
p^{|\actualS| + 1} \cdot F \geq (1 - \varepsilon') \Error(1, |\actualS|; p)^{-1} \cdot \eta(\actualS) \geq  (1 - \varepsilon') \eta(\actualS)\;.
\] \else
$p^{|\actualS| + 1} \cdot F \geq (1 - \varepsilon') \cdot \allowbreak \Error(1, |\actualS|; p)^{-1} \cdot \eta(\actualS) \geq  (1 - \varepsilon') \eta(\actualS)$.
\fi The final inequality follows from $\Error(1, |\actualS|; p) \allowbreak \leq 1$. This concludes the proof. \qed

\ifnum\eprintversion=1
\input{tex_files/proof-poly-bound}
\fi

%% file: tex_files/proof-poly-bound.tex
\subsection{Proof of \Cref{lemma:key-lemma-size}}\label{sec:key-lemma-size}
Fix distinct $z_0, z_1, \ldots, z_\qnum$ and $\actualX,(\actualE_j)_{j \in S}$ as in the lemma statement. We will first simplify the probability calculation by defining placeholder events and apply the inclusion-exclusion principle (more precisely using the bounds in \Cref{lemma:inclusion-exclusion-bound}) so that the events can be easily analyzed. This proof will then crucially relies on the fact that $d$ is relatively large.
 Denote the event in the lemma statement as \[
W := (
f(\actualX) = z_0 \eAnd \forall j \in S: f^{-1}(z_j) = \{\actualE_j\})
\] 
First, we will define the following sets/events containing polynomials with specific evaluation constraints: \begin{align*}
	V &:= \left( \begin{array}{l}
		f(\actualX) = z_0 \eAnd \forall j \in S: f(\actualE_j) = z_j
	\end{array} \right)\\
	U_{i, \alpha} &:=  f(\alpha) = z_i \; ; \; \text{ for } i \in S, \alpha \in \Zp\;.
\end{align*}
Note that $W \subseteq V$ by definition, and $U_{i, \alpha}$ are sets of polynomials which maps $\alpha$ to $z_i$. 
We then claim the following: 
\begin{claim}
	$W = V \setminus \bigcup_{i \in S, \alpha \in \Zp \setminus (\{\actualE_j\}_{j \in S} \cup \{\actualX\})} U_{i, \alpha}$
\end{claim}
\begin{proof}[of Claim]
	We first prove that $W \subseteq  V \setminus \bigcup_{i \in S, \alpha \in \Zp \setminus (\{\actualE_j\}_{j \in S} \cup \{\actualX\})} U_{i, \alpha}$. Consider $f \in W$. Since $W \subseteq V$, $f \in V$. Hence, we only need to show that $f \notin \bigcup_{i \in S, \alpha \in \Zp \setminus (\{\actualE_j\}_{j \in S} \cup \{\actualX\})} U_{i, \alpha}$, i.e., for all $i \in S$ and $\alpha \notin \{\actualE_j\}_{j \in S} \cup \{\actualX\}$, $f(\alpha) \neq z_i$. This is implied by the fact that $f^{-1}(z_i) = \{\actualE_i\}$ for all $i \in S$. 
	
	Now, we consider $f \in V \setminus \bigcup_{i \in S, \alpha \in \Zp \setminus (\{\actualE_j\}_{j \in S} \cup \{\actualX\})} U_{i, \alpha}$ and will show that $f \in W$. To do so, we consider for each $i \in S$, the set $f^{-1}(z_i)$. Note that by definition of $V$, $f(\actualE_i) = z_i$ and that for any $\alpha \notin \{\actualX\} \cup \{\actualE_j\}_{j \in S}, \alpha \notin f^{-1}(z_i)$. Hence, we only need to consider $\alpha \in \{\actualX\} \cup \{\actualE_j\}_{j \in S}$ which is not $\actualE_i$. However, since we know that $z_0, z_1, \ldots, z_\qnum$ are distinct, it cannot be the case that $f(\alpha) = z_i$ and $f(\alpha) = z_j$ for $j \neq i$. Therefore, we have that $f^{-1}(z_i) = \{\actualE_i\}$, meaning $f \in W$, proving the claim. \qed
\end{proof}

With the claim, we can write $\Pr[W]$ as 
\ifnum\eprintversion=0
\begin{align*}
	\Pr[W] &= \Pr[V] - \Pr\left[V \eAnd \bigcup_{\substack{i \in S, \\ \alpha \in \Zp \setminus (\{\actualE_j\}_{j \in S} \cup \{\actualX\})}} U_{i, \alpha}\right]  \\
	&= \Pr[V] - \Pr\left[\bigcup_{\substack{i \in S, \\ \alpha \in \Zp \setminus (\{\actualE_j\}_{j \in S} \cup \{\actualX\})}} (V \eAnd U_{i, \alpha})\right]\;.
\end{align*}
\else
\begin{align*}
	\Pr[W] &= \Pr[V] - \Pr\left[V \eAnd \bigcup_{\substack{i \in S, \\ \alpha \in \Zp \setminus (\{\actualE_j\}_{j \in S} \cup \{\actualX\})}} U_{i, \alpha}\right] 
	= \Pr[V] - \Pr\left[\bigcup_{\substack{i \in S, \\ \alpha \in \Zp \setminus (\{\actualE_j\}_{j \in S} \cup \{\actualX\})}} (V \eAnd U_{i, \alpha})\right]\;.
\end{align*}
\fi

We also make the following observation. 
\begin{lemma}\label{lemma:observation}
	Let $V$ and $U_{i, \alpha}$ be as defined before.
	\begin{enumerate}
		\item For $i \neq j$ and any $\alpha \in \Zp$, $\Pr[U_{i, \alpha} \cap U_{j, \alpha}] = 0$. 
		
		\item For any subset $B \subseteq S \times \Zp \setminus (\{\actualE_j\}_{j \in S} \cup \{\actualX\})$ such that $|B| \leq d - |S| - 1$ and each $(i, \alpha) \in B$ has distinct $\alpha$ (crucially, the $i$'s can be the same), we have that $ \Pr[\bigcap_{(i ,\alpha) \in B} (V \cap U_{i, \alpha})] = p^{- |S| - 1 - |B|}$.
	\end{enumerate}
\end{lemma}
\begin{proof}
	The first point is true simply because $z_i \neq z_j$, so it cannot be that $f(\alpha) = z_i \neq z_j = f(\alpha)$. The second point follows from the fact that the event $\bigcap_{(i ,\alpha) \in B} (V \cap U_{i, \alpha})$ can be written as follows \[
	\bigcap_{(i ,\alpha) \in B} (V \cap U_{i, \alpha}) =
		f(\actualX) = z_0 \eAnd (\forall j \in S: f(\actualE_j) = z_j) \eAnd 
		(\forall (i, \alpha) \in B: f(\alpha) = z_i)
	.
	\] In particular, the set contains polynomials satisfying $|S| + 1 + |B|$ evaluation constraints on $f$, since $(i, \alpha) \in B$ contains all distinct $\alpha$ and $\alpha \notin \{\actualE_j\}_{j \in S} \cup \{\actualX\}$. This boils down to $|S| + 1 + |B|$ linearly independent constraints on the coefficients of the possible monic polynomials $f$ of degree $d$. Therefore, with $f$ uniformly random with degree $d$ and $|S| + 1 + |B| \leq d$, we have that the probability is $\Pr[\bigcap_{(i ,\alpha) \in B} (V \cap U_{i, \alpha})] = p^{- |S| - 1 - |B|}$ \qed
\end{proof}
The above lemma gives $\Pr[V] = p^{- |S| - 1}$ from setting $B = \emptyset$. We also consider the probability of the union of $(V \cap U_{i, \alpha})$. In particular, by the upper and lower bounds in \Cref{lemma:inclusion-exclusion-bound}, we have that for an even integer $K$, which we pick to be such that $K + 1 \leq d - |S| - 1$ we can bound (we mark the slight differences of the upper and lower bounds with {\color{blue} blue}) \begin{align*}
	&\Pr\left[\bigcup_{\substack{i \in S, \\ \alpha \in \Zp \setminus (\{\actualE_j\}_{j \in S} \cup \{\actualX\})}} (V \eAnd U_{i, \alpha})\right]\\
	&{\color{blue} \geq}  \sum_{\color{blue} \substack{B \subseteq S \times (\Zp \setminus \{\actualE_i\}_{i \in S} \cup \{\actualX\}) :\\ 1\leq |B| \leq K }} (-1)^{|B| - 1} \Pr\left[\bigcap_{(i, \alpha) \in B} (V \cap U_{i, \alpha}) \right]\\
	&= \sum_{\color{blue} \substack{B \subseteq S \times (\Zp \setminus \{\actualE_i\}_{i \in S} \cup \{\actualX\}) : \\1\leq |B| \leq K}} (-1)^{|B| - 1}  p^{- |S| - |B| - 1}  \mathbbm{1} [ B \text{ contains $(j, \alpha)$ with distinct } \alpha ] \\
	&= \sum_{i = 1}^{\color{blue} K} (-1)^{i  - 1} \binom{p - |S| - 1}{i} |S|^i p^{- |S| - i - 1}\;.
\end{align*}
The second equality for the lower bound follows from the probability of the event $\bigcap_{(i, \alpha) \in B}(V \cap U_{i, \alpha})$ for $|B| \leq K + 1\leq d - |S| - 1$ as established in \Cref{lemma:observation} (the indicator refers to the conditions from the lemma). 
The third equality for the lower bound follows by counting the number of subsets $B$ of size $i$ satisfying the indicator (i.e., containing $(j, \alpha)$ with distinct $\alpha$). The total count for each $i$ is $ \binom{p - |S| - 1}{i} |S|^i$, where $\binom{p - |S| - 1}{i}$ is the number of ways to choose $i$ distinct $\alpha$ values from $\Zp  \setminus \{\actualE_i\}_{i \in S} \cup \{\actualX\}$, and $|S|^i$ is the number of ways to assign an index $j \in S$ to each $\alpha$. 

Similarly, we also have the upper bound derived from similar steps.
\ifnum\eprintversion=0
\begin{align*}
	&\Pr\left[\bigcup_{\substack{i \in S, \\ \alpha \in \Zp \setminus (\{\actualE_j\}_{j \in S} \cup \{\actualX\})}} (V \eAnd U_{i, \alpha})\right]\\
	&{\color{blue} \leq} \sum_{\color{blue} \substack{B \subseteq S \times (\Zp \setminus \{\actualE_i\}_{i \in S} \cup \{\actualX\}) :\\ 1\leq |B| \leq K + 1}} (-1)^{|B| - 1} \Pr\left[\bigcap_{(i, \alpha) \in B}(V \cap U_{i, \alpha}) \right]\\
	&=  \sum_{i = 1}^{\color{blue} K + 1} (-1)^{i - 1} \binom{p - |S| - 1}{i} |S|^i p^{d - |S| - i - 1}\;,
\end{align*}
\else
\begin{align*}
	\Pr\left[\bigcup_{\substack{i \in S, \\ \alpha \in \Zp \setminus (\{\actualE_j\}_{j \in S} \cup \{\actualX\})}} (V \eAnd U_{i, \alpha})\right]
	&{\color{blue} \leq} \sum_{\color{blue} \substack{B \subseteq S \times (\Zp \setminus \{\actualE_i\}_{i \in S} \cup \{\actualX\}) :\\ 1\leq |B| \leq K + 1}} (-1)^{|B| - 1} \Pr\left[\bigcap_{(i, \alpha) \in B}(V \cap U_{i, \alpha}) \right]\\
	&=  \sum_{i = 1}^{\color{blue} K + 1} (-1)^{i - 1} \binom{p - |S| - 1}{i} |S|^i p^{d - |S| - i - 1}\;,
\end{align*}
\fi
From the above bounds, we can lower and upper bound $\Pr[W]$ as 
\ifnum\eprintversion=0
\begin{align*}
	\Pr[W] &{\color{red}\; \leq \;} p^{- |S| - 1} - \sum_{i = 1}^{\color{blue} K} (-1)^{i  - 1} \binom{p - |S| - 1}{i} |S|^i p^{- |S| - i - 1}\\
	&= p^{- |S| - 1} \sum_{i = 0}^{\color{blue} K} (-1)^{i} \binom{p - |S| - 1}{i}\left(\frac{|S|}{p}\right)^i\\
	\Pr[W] &{\color{red} \; \geq \;}  p^{- |S| - 1} - \sum_{i = 1}^{\color{blue} K + 1} (-1)^{i  - 1} \binom{p - |S| - 1}{i} |S|^i p^{- |S| - i - 1} \\
	&= p^{- |S| - 1} \sum_{i = 0}^{\color{blue} K + 1} (-1)^{i} \binom{p - |S| - 1}{i}\left(\frac{|S|}{p}\right)^i\;.
\end{align*} 
\else
\begin{align*}
	\Pr[W] &{\color{red}\; \leq \;} p^{- |S| - 1} - \sum_{i = 1}^{\color{blue} K} (-1)^{i  - 1} \binom{p - |S| - 1}{i} |S|^i p^{- |S| - i - 1}
	= p^{- |S| - 1} \sum_{i = 0}^{\color{blue} K} (-1)^{i} \binom{p - |S| - 1}{i}\left(\frac{|S|}{p}\right)^i\\
	\Pr[W] &{\color{red} \; \geq \;}  p^{- |S| - 1} - \sum_{i = 1}^{\color{blue} K + 1} (-1)^{i  - 1} \binom{p - |S| - 1}{i} |S|^i p^{- |S| - i - 1} 
	= p^{- |S| - 1} \sum_{i = 0}^{\color{blue} K + 1} (-1)^{i} \binom{p - |S| - 1}{i}\left(\frac{|S|}{p}\right)^i\;.
\end{align*} 
\fi
The equalities for both the upper and lower bounds follow from $-(-1)^{i - 1} = (-1)^i$ and that $p^{- |S| - 1} = \binom{p - |S| - 1}{0} |S|^0 p^{ - |S| - 0 - 1}$. 

Now, we observe that the term $ \sum_{i = 0}^K (-1)^{i} \binom{p - |S| - 1}{i}\left(\frac{|S|}{p}\right)^i$ without the factor $p^{- |S| - 1}$ is simply a truncated sum of the binomial expansion of $\left(1 - \frac{|S|}{p}\right)^{p - |S| - 1}$. In particular, we have 
\ifnum\eprintversion=1
\begin{align*}
\sum_{i = 0}^K (-1)^{i} \binom{p - |S| - 1}{i}\left(\frac{|S|}{p}\right)^i = \left(1 - \frac{|S|}{p}\right)^{p - |S| - 1} + \sum_{i = K + 1}^{p - |S| - 1} (-1)^{i - 1}  \binom{p - |S| - 1}{i}\left(\frac{|S|}{p}\right)^i
\end{align*}
\else
\begin{align*}
	&\sum_{i = 0}^K (-1)^{i} \binom{p - |S| - 1}{i}\left(\frac{|S|}{p}\right)^i \\
	&= \left(1 - \frac{|S|}{p}\right)^{p - |S| - 1} + \sum_{i = K + 1}^{p - |S| - 1} (-1)^{i - 1}  \binom{p - |S| - 1}{i}\left(\frac{|S|}{p}\right)^i
\end{align*}
\fi
Choosing $K = e^2 \qnum + 2 \log_2 p$ -- this can be done by choosing $d \geq (e^2 + 1)\qnum + 2 \log_2 p + 1$. Each of the term in the summation on the RHS with $i > K$ can be upper bounded as \begin{align}
	\binom{p - |S| - 1}{i}\left(\frac{|S|}{p}\right)^i &= \frac{(p - |S| - 1)! |S|^i}{i! (p - |S| - 1 - i)! p^i} \leq \frac{|S|^i}{i!} \nonumber \\
	&\leq \left(\frac{e |S|}{i}\right)^i \nonumber\\
	&\leq \left(\frac{e\qnum}{i}\right)^i\nonumber \\
	&\leq e^{-i} \leq e^{-(e^2 \qnum + 2 \log_2 p)}\;. \label{eq:individual-term-bound}
\end{align}
The first equality follows from expanding the binomial coefficient.
The second inequality follows from $\frac{(p - |S| - 1)!}{(p - |S| - 1 - i)! p^i} \leq 1$.
The third inequality follows from \Cref{lem:sterling-lower-bound}. 
The fourth inequality follows from $|S| \leq \qnum$.
The last two inequalities follows from $i > K \geq e^2 \qnum + 2 \log_2 p$. Hence, we have that \begin{align}
	\sum_{i = K + 1}^{p - |S| - 1}(-1)^{i - 1}  \binom{p - |S| - 1}{i}\left(\frac{|S|}{p}\right)^i 
	&\leq p e^{-(e^2 \qnum + 2 \log_2 p)} \nonumber\\
	&\leq e^{-(e^2 \qnum + \log_2 p)} \nonumber\\
	&\leq e^{-(e^2 - 1) \qnum - \log_2 p} \left(1 - \frac{|S|}{p}\right)^{p - |S| - 1}\;. \label{eq:important-tail-bound}
\end{align} 
The first inequality follows from the prior bound and that the number of terms in the sum is at most $p$. 
The second inequality follows from $p \leq e^{\log_2 p}$. The third inequality follows from $\left(1 - \frac{|S|}{p}\right)^{p - |S| - 1} \geq \left(1 - \frac{|S|}{p}\right)^{p - |S|} \geq e^{-|S|} \geq e^{-\qnum}$ (the second step is via \Cref{lem:lower-bound-expo} with $l = k = |S|$).
Thus, we have that for $\varepsilon = e^{-(e^2 - 1) \qnum - \log_2 p}$ \[
\Pr[W] \leq  (1 + \varepsilon) \left(1 - \frac{|S|}{p}\right)^{p - |S| - 1} \cdot  p^{- |S| - 1}\;.
\] 
Next, for the lower bound of $\Pr[W]$, we can similarly see that 
\ifnum\eprintversion=1
\begin{align*}
	\sum_{i = 0}^{K + 1} (-1)^{i} \binom{p - |S| - 1}{i}\left(\frac{|S|}{p}\right)^i &= \left(1 - \frac{|S|}{p}\right)^{p - |S| - 1} - \sum_{i = K + 2}^{p - |S| - 1} (-1)^{i}  \binom{p - |S| - 1}{i}\left(\frac{|S|}{p}\right)^i\\
	&\geq \left(1 - \frac{|S|}{p}\right)^{p - |S| - 1} - p \cdot e^{-(e^2 \qnum + 2 \log_2 p)}\\
	&\geq (1 - \varepsilon) \left(1 - \frac{|S|}{p}\right)^{p - |S| - 1}
\end{align*}
\else
\begin{align*}
	&\sum_{i = 0}^{K + 1} (-1)^{i} \binom{p - |S| - 1}{i}\left(\frac{|S|}{p}\right)^i \\
	&= \left(1 - \frac{|S|}{p}\right)^{p - |S| - 1} - \sum_{i = K + 2}^{p - |S| - 1} (-1)^{i}  \binom{p - |S| - 1}{i}\left(\frac{|S|}{p}\right)^i\\
	&\geq \left(1 - \frac{|S|}{p}\right)^{p - |S| - 1} - p \cdot e^{-(e^2 \qnum + 2 \log_2 p)}\\
	&\geq (1 - \varepsilon) \left(1 - \frac{|S|}{p}\right)^{p - |S| - 1}
\end{align*}
\fi
The first equality follows from expanding $\left(1 - \frac{|S|}{p}\right)^{p - |S| - 1} $. The second and third inequalities follow from \Cref{eq:individual-term-bound,eq:important-tail-bound}, respectively. 
\ifnum\eprintversion=1
Therefore, we can similarly lower bound $\Pr[W]$ as \[
\Pr[W] \geq  (1 - \varepsilon) \left(1 - \frac{|S|}{p}\right)^{p - |S| - 1} \cdot  p^{- |S| - 1}\;.
\] This proves the lemma.
\else
Therefore, we can similarly lower bound $\Pr[W] \geq  (1 - \varepsilon) \left(1 - \frac{|S|}{p}\right)^{p - |S| - 1}  p^{- |S| - 1}$, proving the lemma.
\fi \qed

%% file: tex_files/full-approach.tex
\section{Our Tight Reductions for Distinct Messages}\label{sec:full-proof}

In this section, we give our tight reductions in \Cref{sec:tight-reduction-state} with the analysis in the following subsections.

\subsection{Proof of \Cref{lem:tight-proof-main}}\label{sec:tight-reduction-state}
For simplicity, assume that all algorithms implicitly take as input the group description $\pGroupDescript$.  
We consider an $\SUF'$ adversary $\advA$ making exactly $\qnum$ distinct signing queries (since $\qnum$ is the upper bound we assume this without loss of generality). The adversary receives at the start of the game the tuple of public parameters and verification key $(G_1, H, G_2, X_{2,1})$. For each $i \in [\qnum]$, we denote $m_i \in \Zp$ as the query and $\sigma_i = (A_i, e_i)$ as the resulting signature. At the end of the game, $\advA$ outputs a forgery $(m^*, \sigma^* = (A^*, e^*))$. We consider the event where $\advA$ wins in $\SUF'$ game, denoted $\Forge_1$ adopting the name from \cite{EC:TesZhu23b}.

We also consider an additional bad event $\BadQ$ where the adversary $\advA$ makes a signing query $m_i$ such that $G_1 + m_i H = 0_{\GOne}$. We denote the event $(\Forge_1 \eAnd \neg \BadQ)$ as $\Forge_1'$. Here, we have that \ifnum\eprintversion=1 \[
	\Adv^{\suf'}_{\BBS_1}(\advA, \secpar) = \Pr[\Forge_1] \leq \Pr[\Forge_1'] + \Pr[\BadQ] \;.
\] \else 
$\Pr[\Forge_1] \leq \Pr[\Forge_1'] + \Pr[\BadQ]$.
\fi

\heading{Bounding $\Pr[\BadQ]$.}
We construct a reduction $\advB_{\BadQ}$ playing the DLOG game such that 
\ifnum\eprintversion=1
\[
	\Pr[\BadQ] \leq \Adv^{\dl}_{\GroupGen}(\advB_{\BadQ}, \secpar)\;.
\] 
\else
$\Pr[\BadQ] \leq \Adv^{\dl}_{\GroupGen}(\advB_{\BadQ}, \secpar)$.
\fi
In particular, $\advB_{\BadQ}$ takes as input $G_1 \in \GOne^*, Z_{1,1} \in \GOne, G_2 \in \GTwo^*$. The reduction samples $x \getsr \Zp$ and set $X_{2,1} = xG_2$ and run the adversary with the input $(G_1, H = Z_{1,1}, G_2, X_{2,1})$. For each signing query $m_i$, the reduction returns $(A_i = \frac{G + m_i H}{x - e_i}, e_i \getsr \Zp)$. If $\BadQ$ occurs, the reduction returns $m_i^{-1}$. Note that if $\BadQ$ occurs, $m_i \neq 0$ and is invertible, since $G_1$ is a generator. Thus, the above inequality is justified.

\input{figures/tight-reduction}

\heading{Bounding $\Pr[\Forge_1']$.} 
We additionally assume that $\advA$ does not make signing queries that trigger $\BadQ$ and consider the event $\Forge_1'$. To bound $\Pr[\Forge_1']$, we consider the reduction $\advB_{1}$, given in \Cref{fig:tight-reduction}, following the sketch in \Cref{sec:overview-tight-reduction}, and playing the $d'$-$\SDH$ game where $d' = \lceil(e^2 + 2)\qnum + 2\log_2 p\rceil + 1 = \Theta(\qnum + \secpar)$. 

We note that the running time of $\advB_1$ is dominated by (1) the running time of $\advA$ (roughly $t_\advA$) and (2) the signing simulation which consists of (2.1) checking if $f(\varX) + m_i\beta$ has exactly $1$ root in $\Zp$ and (2.2) computing the group element $A_i$ from the $d$-$\SDH$ instance. Step (2.1) can be done by computing $\gcd(f(\varX) + m_i\beta, \varX^p - \varX)$ and checking if the result is of degree 1 or not. More precisely, we can first compute $\varX^p - \varX \pmod{f(\varX) + m_i\beta}$ via repeated squaring and polynomial division taking $O(\qnum \log \qnum \log^2 p)$ time~(see \cite[Lecture 1]{yap2000fundamental}), and then computing the GCD, via the Half-GCD algorithm~(see \cite[Lecture 2]{yap2000fundamental} and \cite[Theorem 8.19]{aho1974design}), takes $O(\qnum \log^2 \qnum \log p)$ per query with $d = \Theta(\qnum + \secpar)$.\footnote{We do not need to fully factor out the polynomial $f(\varX) + m_i\beta$ nor need to know all available zeros. This helps immensely in reducing the running time of our reduction.} Step (2.2) can be done by computing the polynomial $\frac{ \prod_{i = 1}^\qnum (\varX - \tilde{e}_i) (f(\varX) + m_i \beta)}{\varX - e_i}$ and compute $A_i$ as a linear combination of $G_1,X_{1,1},\ldots,X_{1,d}$, taking $O(\qnum (T_\Group + T_p))$ per query.

The reduction $\advB_1$ wins in the $d'$-$\SDH$ game if the following occurs: (1) $\advA$ returns a ``fresh" forgery $(m^*, (A^*, e^*))$ such that $A_{i^*} \neq A^*$ and $e_{i^*} = e^*$ for some $i^* \in [\qnum]$ and $A^* = \frac{\overline{G}_1 + m^* \overline H}{x - e^*}$, (2) $x$ and $e_1, \ldots, e_\qnum$ are distinct, (3) $i^* \in S$, and (4) $\{e_i\}_{i \in [\qnum]} \cap \{\widetilde{e}_{i}\}_{i \in S} = \emptyset$. This is because if $\{e_i\}_{i \in [\qnum]} \cap \{\widetilde{e}_{i}\}_{i \in S} = \emptyset$ and $i^* \in S$, then $(\varX - e_{i^*}) \nmid  \prod_{i = 1}^\qnum (\varX - \tilde{e}_i)$. Moreover, the forgery being fresh and $i^* \in S$ implies that $m^* \neq m_{i^*}$, so $(\varX - e_{i^*}) \nmid (f + m^*\beta)$ (in order to apply \Cref{remark:compute-sdh}); otherwise, $f(e_{i^*}) = -m^* \beta \neq - m_{i^*} \beta = f(e_{i^*})$, which is a contradiction. 

With that said, we claim the following lemma which concludes the proof. 
\begin{lemma}\label{lemma:analyze-reduction} For a prime $p$, an integer $\qnum$ where  $p \geq 2\qnum^2$, and $d' = \lceil(e^2 + 2)\qnum + 2\log_2 p\rceil + 1 < p$,
	\ifnum\eprintversion=1
	\[
	\Pr[\Forge_1'] \leq e \cdot\left( \Adv^{d'\text{-}\sdh}_{\GroupGen}(\advB_1, \secpar) +  \frac{2\qnum^2 + 4 \qnum + 1}{p} \right)\;.
	\]\else
	$	\Pr[\Forge_1'] \leq e \cdot\left( \Adv^{d'\text{-}\sdh}_{\GroupGen}(\advB_1, \secpar) +  \frac{2\qnum^2 + 4 \qnum + 1}{p} \right)$.
	\fi
\end{lemma}

\subsection{Proof of \Cref{lemma:analyze-reduction}}

We assume $\advA$ to be deterministic without loss of generality.
For our formal analysis, we will use shorthand $\Pr_0[\cdot]$ and $\Pr_1[\cdot]$ to denote probabilities in Experiment 0, which is the $\SUF$ game running the adversary $\advA$, and Experiment 1, which is the simulation by $\advB_1$ to $\advA$, respectively. However, we augment the transcripts in both experiments with the set $S$ further explained below. In particular, we denote the transcript in experiment $b \in \bits$ as $T_b$, which are of the form 
\ifnum\eprintversion=1
\begin{align*}
	T_0 &= (a, b, c, x, (m_1, e_1), \ldots, (m_\qnum, e_\qnum), S)\;,\\
	T_1 &= (\overline{a},\overline{b}, c, x,  (m_1, e_1), \ldots, (m_\qnum, e_\qnum), S)\;.
\end{align*}
\else
\begin{align*}
	T_0 &= (a, b, c, x, (m_i, e_i)_{i \in [\qnum]}, S)\;,
	&T_1 &= (\overline{a},\overline{b}, c, x,  (m_i, e_i)_{i \in [\qnum]}, S)\;.
\end{align*}
\fi
The values $b, \overline{b}$ and $c$ denote the discrete logarithms with respect to some fixed generator of $H$, $\overline{H}$ and $G_2$, respectively. 

In $T_0$, $a$ denotes the discrete logarithm of $G_1$ with respect to some fixed generator, and the set $S \subseteq [\qnum]$ is sampled independently of the rest of the transcript with the probability density function $\eta : 2^{[\qnum]} \to [0, 1]$ as defined in \Cref{lem:dist-S-random-f}. 
\ifnum\eprintversion=1
In particular, let $z_0, \ldots, z_\qnum \in \Zp$ be distinct and any $x_0 \in \Zp$, we can sample $S$ as follows: (1) sample a random function $f : \Zp \to \Zp$ such that $f(x_0) = z_0$ and (2) set $S$ as the set of indices $i \in [\qnum]$ such that $|f^{-1}(z_i)| = 1$, i.e., there is exactly one preimage of $i$. Note that the distribution of $S$ is independent of the choice of images $z_0, \ldots, z_\qnum$ and the fixed $x_0$. 
\fi
In $T_1$, $\overline{a}$ denotes the discrete logarithm of $\overline{G}_1$ with respect to some fixed generator, and the rest of the transcript is sampled as in the reduction. The set $S \subseteq [\qnum]$ is defined as follows: given the messages $m_1, \ldots, m_\qnum$ sent by $\advA$, the sampled $f$, $x$, and $\beta$, $S := \{i \in [\qnum] : |f^{-1}(-m_i \beta)| = 1\}$. 
We excluded the  elements $G_1, H, G_2$, and $A_i$, since they are already determined by the transcript.

For each $i^* \in [\qnum]$, we define the event $\Forge'_{1, i^*}$ as the event that $\Forge_1'$ occurs and the value $e^* = e_{i^*}$. We note that since the adversary is deterministic, this event is well-defined for both transcripts $T_0, T_1$ (as the forgery is determined by the transcript).
Also, denote $\GoodE$ as the event that $x, e_1, \ldots, e_\qnum$ are distinct. Then, according to the winning condition of $\advB_1$,  \ifnum\eprintversion=1 we have that\begin{align*}
	\Adv^{d'\text{-}\sdh}_{\GroupGen}(\advB_1, \secpar) 
	&\geq \Pr_1\left[\GoodE \eAnd \left(\exists i^* \in [\qnum]: \Forge'_{1, i^*} \eAnd i^* \in S\right) \eAnd (\{e_i\}_{i \in [\qnum]} \cap \{\widetilde{e}_{i}\}_{i \in S} = \emptyset)\right]\\
	&\geq  \Pr_1\left[\GoodE \eAnd \left(\exists i^* \in [\qnum]: \Forge'_{1, i^*} \eAnd i^* \in S\right)\right]  - \Pr_1[\{e_i\}_{i \in [\qnum]} \cap \{\widetilde{e}_{i}\}_{i \in S} \neq \emptyset] \\
	&\geq  \Pr_1\left[\GoodE \eAnd \left(\exists i^* \in [\qnum]: \Forge'_{1, i^*} \eAnd i^* \in S\right) \right] - \frac{\qnum^2}{p}\;.
\end{align*}
\else
we bound $\Adv^{d'\text{-}\sdh}_{\GroupGen}(\advB_1, \secpar)$
\begin{align*}
	&\geq \Pr_1\left[\GoodE \eAnd (\{e_i\}_{i \in [\qnum]} \cap \{\widetilde{e}_{i}\}_{i \in S} = \emptyset) \eAnd \exists i^* \in [\qnum]: (\Forge'_{1, i^*} \eAnd i^* \in S) \right]\\
	&\geq  \Pr_1\left[\GoodE \eAnd \exists i^* \in [\qnum]: (\Forge'_{1, i^*} \eAnd i^* \in S)\right] - \Pr_1[\{e_i\}_{i \in [\qnum]} \cap \{\widetilde{e}_{i}\}_{i \in S} \neq \emptyset] \\
	&\geq \textstyle \Pr_1\left[\GoodE \eAnd \exists i^* \in [\qnum]: (\Forge'_{1, i^*} \eAnd i^* \in S) \right] - \frac{\qnum^2}{p}\;.
\end{align*}
\fi
The last inequality follows from $\Pr_1[\{e_i\}_{i \in [\qnum]} \cap \{\widetilde{e}_{i}\}_{i \in S} \neq \emptyset] \leq \frac{\qnum^2}{p}$. 
This is because either (1) there is no element in $S$ (the event is not triggered), or (2) for $i \in S$, $\widetilde{e}_{i}$ is uniformly random and hidden from the view of the adversary (the transcript hides $\widetilde{e}_{i}$ for $i \in S$), so each collides with $\{e_i\}_{i \in [\qnum]}$ with probability at most $\frac{\qnum}{p}$.

Then, the following lemma, proved in \Cref{sec:bound-stat-distance}, upper bounds the total variation distance between the two transcripts given that $d'$ is large enough. 
\begin{lemma}\label{lemma:bound-stat-distance}
	For a prime $p$, an integer $\qnum$ where $p \geq 2\qnum^2$, and $d' = \lceil(e^2 + 2)\qnum + 2 \log_2 p\rceil + 1$,
	\ifnum\eprintversion=1
	\[
	\SD(T_0, T_1) \leq  \frac{(\qnum+ 1)^2 + 4\qnum}{2p}\;.
	\]
	\else
	$	\SD(T_0, T_1) \leq  \frac{(\qnum+ 1)^2 + 4\qnum}{2p}$.
	\fi
\end{lemma}
Using \Cref{lemma:bound-stat-distance}, we have that 
\ifnum\eprintversion=1
\begin{align*}
	&\Pr_1\left[\GoodE \eAnd \exists i^* \in [\qnum]: (\Forge'_{1, i^*} \eAnd i^* \in S) \right] - \frac{\qnum^2}{p} \\
	&\geq \Pr_0\left[\GoodE \eAnd \exists i^* \in [\qnum]: (\Forge'_{1, i^*} \eAnd i^* \in S) \right] -  \frac{(\qnum+ 1)^2 + 4\qnum}{2p}  - \frac{\qnum^2}{p}\\
	& = \sum_{i^* = 1}^\qnum \Pr_0\left[\GoodE \eAnd \Forge'_{1, i^*} \eAnd i^* \in S \right]  -  \frac{3\qnum^2 + 6 \qnum + 1}{2p} \\
	& =  \sum_{i^* = 1}^\qnum \Pr_0\left[\GoodE \eAnd \Forge'_{1, i^*}\right] \Pr_0\left[i^* \in S \right] -  \frac{3\qnum^2 + 6 \qnum + 1}{2p} \\
	& \geq \frac{1}{e} \sum_{i^* = 1}^\qnum \Pr_0\left[\GoodE \eAnd \Forge'_{1, i^*}\right] -  \frac{3\qnum^2 + 6 \qnum + 1}{2p} \\
	&  = \frac{1}{e} \Pr_0[\GoodE \eAnd \Forge_1'] -  \frac{3\qnum^2 + 6 \qnum + 1}{2p} \\
	& \geq \frac{1}{e} \Pr_0[\Forge_1'] - \frac{1}{e}\Pr_0[\neg \GoodE] -  \frac{3\qnum^2 + 6 \qnum + 1}{2p} \\
	&\geq \frac{1}{e} \Pr_0[ \Forge_1'] - \frac{2\qnum^2 + 4 \qnum + 1}{p}\;.
\end{align*} 
\else
\begin{align*}
	&\textstyle \Pr_1\left[\GoodE \eAnd \exists i^* \in [\qnum]: (\Forge'_{1, i^*} \eAnd i^* \in S) \right] - \frac{\qnum^2}{p} \\
	&\textstyle \geq \Pr_0\left[\GoodE \eAnd \exists i^* \in [\qnum]: (\Forge'_{1, i^*} \eAnd i^* \in S) \right] -  \frac{(\qnum+ 1)^2 + 4\qnum}{2p}  - \frac{\qnum^2}{p} \\
	&\textstyle = \sum_{i^* = 1}^\qnum \Pr_0\left[\GoodE \eAnd \Forge'_{1, i^*} \eAnd i^* \in S \right]  -  \frac{3\qnum^2 + 6 \qnum + 1}{2p} \\
	&\textstyle =  \sum_{i^* = 1}^\qnum \Pr_0\left[\GoodE \eAnd \Forge'_{1, i^*}\right] \Pr_0\left[i^* \in S \right] -  \frac{3\qnum^2 + 6 \qnum + 1}{2p} \\
	&\textstyle \geq \frac{1}{e} \sum_{i^* = 1}^\qnum \Pr_0\left[\GoodE \eAnd \Forge'_{1, i^*}\right] -  \frac{3\qnum^2 + 6 \qnum + 1}{2p} \\
	&\textstyle  = \frac{1}{e} \Pr_0[\GoodE \eAnd \Forge_1'] -  \frac{3\qnum^2 + 6 \qnum + 1}{2p} \\
	&\textstyle \geq \frac{1}{e} \Pr_0[\Forge_1'] - \frac{1}{e}\Pr_0[\neg \GoodE] -  \frac{3\qnum^2 + 6 \qnum + 1}{2p}  \geq \frac{1}{e} \Pr_0[ \Forge_1'] - \frac{2\qnum^2 + 4 \qnum + 1}{p}\;.
\end{align*} 
\fi
The first inequality follows from the lemma. 
The second equality follows from the events $\Forge'_{1, i^*} \eAnd i^* \in S$ being disjoint for all $i^* \in [\qnum]$ when $\GoodE$ is true. The third equality follows from $S$ sampled independently from the rest of the transcript in $T_0$ (note that $i^*$ is a fixed value). Then, the fourth inequality follows from \Cref{corol:prob-contain-index}. The fifth equality is obtained by summing over the probability of disjoint events $\Forge'_{1, i^*}$. The last inequality follows from $\Pr[\neg \GoodE] \leq \binom{\qnum + 1}{2} \frac{1}{p}$. Rearranging the terms finally concludes the proof. 
\qed

\subsection{Proof of \Cref{lemma:bound-stat-distance}} \label{sec:bound-stat-distance}

\heading{H-coefficient technique.}
 \ifnum\eprintversion=1 Our goal is to bound the total variation distance  $\SD(T_0, T_1)$ defined as \[
\SD(T_0, T_1):= \frac{1}{2} \sum_{\tau} |\Pr[T_0 = \tau] - \Pr[T_1 = \tau]| = \sum_{\tau} \max\{0, \Pr[T_0 = \tau] - \Pr[T_1 = \tau]\}\;.
\] The sum is over all transcripts $\tau$. We will employ Patarin's H-Coefficient technique~\cite{AFRICACRYPT:Patarin08}, but using the formalism of~\cite{C:HoaTes16}, which we now introduce.

\else
We bound the total variation distance 
$\SD(T_0, T_1) \allowbreak := \frac{1}{2} \sum_{\tau} |\Pr[T_0 = \tau] - \Pr[T_1 = \tau]| = \sum_{\tau} \max\{0, \Pr[T_0 = \tau] - \Pr[T_1 = \tau]\}$, by employing Patarin's H-Coefficient method~\cite{AFRICACRYPT:Patarin08}, but using the formalism of~\cite{C:HoaTes16}, which we now introduce.
\fi  
First, we denote any transcript $\tau$ to be of the form \ifnum\eprintversion=1\[
	\tau = (\actualA, \actualB, \actualC, \actualX, (\actualMsg_1, \actualE_1), \ldots, (\actualMsg_\qnum, \actualE_\qnum), \actualS)\;.
\]\else
$	\tau = (\actualA, \actualB, \actualC, \actualX, (\actualMsg_1, \actualE_1), \ldots, (\actualMsg_\qnum, \actualE_\qnum), \actualS)$.
\fi Note the distinction of a fixed value (denoted with $\underline{(\cdot)}$) and random variables (without). Let $\calT$ denote the set of transcripts $\tau$ such that $\Pr[T_0 = \tau] > 0$. Also, we let $\interprob_0(\tau)$ and $\interprob_1(\tau)$ denote the interpolation probability, i.e., the probability that we pick the random coins in the respective experiment and result in the transcript $\tau$ if the queries $\actualMsg_1, \ldots, \actualMsg_\qnum$ are fixed ahead of time. Note that for any transcript $\tau$, we have $\Pr_b[T_b = \tau] \in \{0, \interprob_b(\tau)\}$, since $\advA$ is deterministic. Then, 
\ifnum\eprintversion=1
\[
	\SD(T_0, T_1) = \sum_{\tau \in \calT} \max\{0, \interprob_0(\tau) - \interprob_1(\tau)\}\;.
\] \else
$\SD(T_0, T_1) = \sum_{\tau \in \calT} \max\{0, \interprob_0(\tau) - \interprob_1(\tau)\}$.
\fi The key idea for the H-coefficient technique is to consider a subset set of good transcripts $\Good \subseteq \calT$ such that for $\tau \in \Good$, 
\ifnum\eprintversion=1
\[
1 - \frac{\interprob_1(\tau)}{\interprob_0(\tau)} \leq \delta
\]
\else
$1 - \frac{\interprob_1(\tau)}{\interprob_0(\tau)} \leq \delta$
\fi for some $\delta \in [0, \infty)$. Then, we can bound $\SD(T_0, T_1)$ as \[
\SD(T_0, T_1) \leq \sum_{\tau \in \Good} \interprob_0(\tau) \max\left\{0, 1 - \frac{\interprob_1(\tau)}{\interprob_0(\tau)}\right\} + \Pr_0[\neg \Good] \leq \delta + \Pr_0[\neg \Good]\;.
\] Here, the event $\neg \Good$ in $T_0$ is independent of the queries of the adversary $\advA$. This lets us circumvent the adaptive nature of the distinguisher when analyzing the statistical distance, as considering $\neg \Good$ in $T_1$ can be quite complicated. 

\heading{Interpolation probability.}
Recall first that we only consider transcripts where $\actualMsg_i \neq  - \actualA/\actualB$ (otherwise, $G_1 + m_i H = 0_{\GOne}$ which is ruled out by the event $\BadQ$) and the queries $m_i$'s are distinct (this is assumed at the beginning). Now, we define $\Good$ to be the set of transcripts $\tau$ such that (1) $\actualX$ and $\actualE_1, \ldots, \actualE_\qnum$ are distinct, and (2) $\actualA$ and $\actualB$ are non-zero. Note here that $\Pr_0[\neg \Good] \leq \frac{(\qnum + 1)^2}{2p} + \frac{1}{p}$ via the union bound over the negation of (1) and (2). Now, fix a transcript $\tau \in \Good$, we will compute the interpolation probabilities $\interprob_0(\tau)$ and $\interprob_1(\tau)$ and lower bound $\frac{\interprob_1(\tau)}{\interprob_0(\tau)}$ to ultimately bound $\SD(T_0, T_1)$ via the above inequality.

For $\interprob_0(\tau)$, we consider the distribution of $T_0$. First, $b, x, e_1, \ldots, e_\qnum$ are uniformly random over $\Zp$ and $a, c$ are uniformly random in $\Zp^*$. Also, $S$ is sampled independently with respect to the PDF $\eta$. Hence we have that the interpolation probability of the sampled transcript being exactly $\tau$ is 
\ifnum\eprintversion=1
\begin{align*}
	\interprob_0(\tau) &= \underbrace{\frac{1}{p-1}}_{a} \cdot \underbrace{\frac{1}{p}}_{b} \cdot \underbrace{\frac{1}{p-1}}_{c} \cdot \underbrace{\frac{1}{p}}_{x} \cdot \underbrace{\frac{1}{p^{\qnum}}}_{(e_i)_{i \in [\qnum]}} \cdot \eta(\actualS) \\
	&= \underbrace{\frac{1}{p}}_{\actualB} \cdot \underbrace{\frac{1}{p-1}}_{\actualC} \cdot \underbrace{\frac{1}{p}}_{\actualX} \cdot \underbrace{\frac{1}{p^{\qnum - |\actualS|}}}_{(\actualE_i)_{i \notin \actualS}} \cdot \underbrace{\frac{1}{(p-1)p^{|\actualS|}} \eta(\actualS)}_{(\actualA, (\actualE_i)_{i \in \actualS}, \actualS)}\;. 
\end{align*}
\else
\begin{align*}
	\interprob_0(\tau) &= \underbrace{(p-1)^{-1}}_{a} \cdot \underbrace{p^{-1}}_{b} \cdot \underbrace{(p-1)^{-1}}_{c} \cdot \underbrace{p^{-1}}_{x} \cdot \underbrace{p^{-\qnum}}_{(e_i)_{i \in [\qnum]}} \cdot \eta(\actualS) \\
	&= \underbrace{p^{-1}}_{\actualB} \cdot \underbrace{(p-1)^{-1}}_{\actualC} \cdot \underbrace{p^{-1}}_{\actualX} \cdot \underbrace{p^{-\qnum + |\actualS|}}_{(\actualE_i)_{i \notin \actualS}} \cdot \underbrace{(p-1)^{-1}p^{-|\actualS|} \eta(\actualS)}_{(\actualA, (\actualE_i)_{i \in \actualS}, \actualS)}\;. 
\end{align*}
\fi
The first equality follows because each value in the transcript is sampled uniformly at random according to their domain. (Note: we write the factors in the order of the transcript $\tau$ and annotate the random variable each factor corresponds to.) The second equality follows by simply rearranging the terms. The corresponding annotated values will assist in comparing the quantity with $\interprob_1$. 

Now, we consider the more complex $\interprob_1(\tau)$. To reiterate, the transcript is defined by the randomness $(x \in \Zp, (\tilde{e}_i \in \Zp)_{i \in [\qnum]}, f \in \Zp^{d}[\varX], \beta \in \Zp^*, c \in \Zp^*)$ sampled by the reduction. This means $\interprob_1(\tau)$ can be written as follows: 
\ifnum\eprintversion=1
\[
\interprob_1(\tau) = \underbrace{\frac{1}{p^d}}_{f} \cdot \underbrace{\frac{1}{p-1}}_{\beta} \cdot \underbrace{\frac{1}{p-1}}_{c} \cdot \underbrace{\frac{1}{p}}_{x} \cdot \underbrace{\frac{1}{p^\qnum}}_{(\widetilde{e}_i)_{i \in [\qnum]}} \sum_{f, x, (\tilde{e}_i)_{i \in [\qnum]}, \beta, c} \mathbbm{1}\left[ \begin{array}{l}
	\actualX = x \eAnd \actualC = c \eAnd\\
	\actualA = f(\actualX) \prod_{ i = 1}^\qnum (\actualX - \tilde{e}_i) \eAnd\\
	\actualB = \beta \prod_{ i = 1}^\qnum (\actualX - \tilde{e}_i) \eAnd\\
	\forall i \notin \actualS: |f^{-1}(-\actualMsg_i \beta)| \neq 1 \eAnd \widetilde{e}_i = \actualE_i \eAnd\\
	\forall i \in \actualS: f^{-1}(-\actualMsg_i \beta) = \{\actualE_i\} 
\end{array}\right]\;.
\] \else 
\begin{align*}
	\interprob_1(\tau) &= \underbrace{p^{-d}}_{f} \cdot \underbrace{(p - 1)^{-1}}_{\beta} \cdot \underbrace{(p - 1)^{-1}}_{c} \cdot \underbrace{p^{-1}}_{x} \cdot \underbrace{p^{-\qnum}}_{(\widetilde{e}_i)_{i \in [\qnum]}} \cdot \\ 
	&\pcind \sum_{f, x, (\tilde{e}_i)_{i \in [\qnum]}, \beta, c} \mathbbm{1}\left[ \begin{array}{l}
		\actualX = x \eAnd \actualC = c \eAnd \\
		\actualA = f(\actualX) \prod_{ i = 1}^\qnum (\actualX - \tilde{e}_i) \eAnd 
		\actualB = \beta \prod_{ i = 1}^\qnum (\actualX - \tilde{e}_i) \eAnd\\
		\forall i \notin \actualS: |f^{-1}(-\actualMsg_i \beta)| \neq 1 \eAnd \widetilde{e}_i = \actualE_i \eAnd\\
		\forall i \in \actualS: f^{-1}(-\actualMsg_i \beta) = \{\actualE_i\} 
	\end{array}\right]\;.
\end{align*}
\fi The sum counts all possible randomness that results in the transcript $\tau$. Note that the indicator is 1 only if $x = \actualX$, $c = \actualC$, $\widetilde{e}_i = \actualE_i$ for all $i \notin \actualS$, and $\widetilde{e}_i \neq  \actualX$ 
for all $i \in [\qnum]$. The last condition arises from considering $\tau \in \Good$ where $\actualA, \actualB \neq 0$ and all $\actualE_i \neq \actualX$. By constraining to these particular values, we  rewrite the sum as 
\ifnum\eprintversion=1 \[
\underbrace{\frac{1}{p-1}}_{\actualB} \cdot \underbrace{\frac{1}{p-1}}_{\actualC} \cdot \underbrace{\frac{1}{p}}_{\actualX}  \cdot \underbrace{\frac{1}{p^{\qnum - |\actualS|}}}_{(\actualE_i)_{i \notin \actualS}} \cdot \underbrace{\frac{1}{p^{|\actualS|} p^{d}} \sum_{(\tilde{e}_i)_{i \in \actualS} \neq \actualX} \sum_{\beta} \sum_{f} \mathbbm{1}\left[ \begin{array}{l}
	f(\actualX) = \actualA \prod_{i \notin \actualS} (\actualX - \actualE_i)^{-1} \prod_{ i \in \actualS} (\actualX - \tilde{e}_i)^{-1} \eAnd\\
	\beta = \actualB \prod_{i \notin \actualS} (\actualX - \actualE_i)^{-1} \prod_{ i \in \actualS} (\actualX - \tilde{e}_i)^{-1} \eAnd\\
	\forall i \notin \actualS: |f^{-1}(-\actualMsg_i \beta)| \neq 1 \eAnd\\
	\forall i \in \actualS: f^{-1}(-\actualMsg_i \beta) = \{\actualE_i\} 
\end{array}\right]}_{(\actualA, (\actualE_i)_{i \in \actualS}, \actualS)}\;.
\] \else 
\begin{align*}
	& \underbrace{(p - 1)^{-1}}_{\actualB} \cdot \underbrace{(p - 1)^{-1}}_{\actualC} \cdot \underbrace{p^{-1}}_{\actualX}  \cdot \underbrace{p^{-\qnum + |\actualS|}}_{(\actualE_i)_{i \notin \actualS}} \cdot \\
	&\pcind \underbrace{p^{-|\actualS|} p^{-d} \sum_{(\tilde{e}_i)_{i \in \actualS} \neq \actualX} \sum_{\beta} \sum_{f} \mathbbm{1}\left[ \begin{array}{l}
			f(\actualX) = \actualA \prod_{i \notin \actualS} (\actualX - \actualE_i)^{-1} \prod_{ i \in \actualS} (\actualX - \tilde{e}_i)^{-1} \eAnd\\
			\beta = \actualB \prod_{i \notin \actualS} (\actualX - \actualE_i)^{-1} \prod_{ i \in \actualS} (\actualX - \tilde{e}_i)^{-1} \eAnd\\
			\forall i \notin \actualS: |f^{-1}(-\actualMsg_i \beta)| \neq 1 \eAnd\\
			\forall i \in \actualS: f^{-1}(-\actualMsg_i \beta) = \{\actualE_i\} 
		\end{array}\right]}_{(\actualA, (\actualE_i)_{i \in \actualS}, \actualS)}\;.
\end{align*}
\fi Now, if we fix each tuple $(\tilde{e}_i)_{i \in [\qnum]}$ such that $\tilde{e}_i \neq \actualX$ for $i \in [\qnum]$ and $(\tilde{e}_i)_{i \in \actualS} = (\actualE_i)_{i \in \actualS}$, $\beta$ is fixed as $\beta = \actualB \prod_{i \notin \actualS} (\actualX - \actualE_i)^{-1} \prod_{ i \in \actualS} (\actualX - \tilde{e}_i)^{-1} \neq 0$, and $f(\actualX)$ is fixed as $\actualA' = \actualA \prod_{i \notin \actualS} (\actualX - \actualE_i)^{-1} \prod_{ i \in \actualS} (\actualX - \tilde{e}_i)^{-1} \neq 0$. Then, for these fixed values, we want to count the number of $f$ such that the following constraint is true  \[
	f(\actualX) = \actualA' \eAnd (\forall i \notin \actualS: |f^{-1}(-\actualMsg_i \beta)| \neq 1) \eAnd
	(\forall i \in \actualS: f^{-1}(-\actualMsg_i \beta) =  \{\actualE_i\})\;.
\]
This coincides with the event considered in \Cref{lemma:size-f}\ifnum\eprintversion=1 in \Cref{sec:set-S-rand-poly}\fi, by setting
$z_0 = \actualA', z_i = - \actualMsg_i \beta,$ and $\actualE_i = \actualE_i$ for $i \in [\qnum]$. 
Note that $z_0, z_1, \ldots, z_\qnum$ are distinct, since $\tau \in \Good$ and $- \actualMsg_i \neq \actualA / \actualB$. Applying \Cref{lemma:size-f}, we have that 
\ifnum\eprintversion=1
\begin{align*}
	\frac{1}{p^d} \cdot \sum_{f} \mathbbm{1}\left[ \begin{array}{l}
		f(\actualX) = \actualA' \eAnd 
		\forall i \notin \actualS: |f^{-1}(-\actualMsg_i \beta)| \neq 1\\
		\eAnd \forall i \in \actualS: f^{-1}(-\actualMsg_i \beta) = \{\actualE_i\}
	\end{array}\right] 
	&= 	\Pr_{f \getsr \Zp^d[\varX]}\left[\begin{array}{l}f(\actualX) = z_0 \eAnd \forall i \notin \actualS: |f^{-1}(-\actualMsg_i \beta)| \neq 1\\
	\eAnd \forall i \in \actualS: f^{-1}(-\actualMsg_i \beta) =  \{\actualE_i\}
	\end{array} \right]\\
	&\geq  (1 - \varepsilon) \cdot p^{- 1 - |\actualS|} \eta(\actualS) \;,
\end{align*}
\else
\begin{align*}
	&p^{-d} \cdot \sum_{f} \mathbbm{1}\left[ \begin{array}{l}
		f(\actualX) = \actualA' \eAnd 
		\forall i \notin \actualS: |f^{-1}(-\actualMsg_i \beta)| \neq 1 \eAnd \forall i \in \actualS: f^{-1}(-\actualMsg_i \beta) = \{\actualE_i\}
	\end{array}\right] \\
	&= 	\Pr_{f \getsr \Zp^d[\varX]}\left[\begin{array}{l}f(\actualX) = z_0 \eAnd \forall i \notin \actualS: |f^{-1}(-\actualMsg_i \beta)| \neq 1\\
		\eAnd \forall i \in \actualS: f^{-1}(-\actualMsg_i \beta) =  \{\actualE_i\}
	\end{array} \right] \geq  (1 - \varepsilon) \cdot p^{- 1 - |\actualS|} \eta(\actualS) \;,
\end{align*}
\fi
where $\varepsilon = e^{-(e^2 - 3) \qnum - \log_2 p + 1}$. 
Finally, we can lower bound $\interprob_1(\tau)$ as \ifnum\eprintversion=1 
\begin{align*}
	\interprob_1(\tau) &= \frac{1}{p-1} \cdot \frac{1}{p-1} \cdot \frac{1}{p}  \cdot \frac{1}{p^{\qnum - |\actualS|}} \cdot \frac{1}{p^{|\actualS|} p^{d}} \sum_{(\tilde{e}_i)_{i \in \actualS} \neq \actualX} \sum_{\beta} \sum_{f} \mathbbm{1}\left[ \begin{array}{l}
			f(\actualX) = \actualA \prod_{i \notin \actualS} (\actualX - \actualE_i)^{-1} \prod_{ i \in \actualS} (\actualX - \tilde{e}_i)^{-1} \eAnd\\
			\beta = \actualB \prod_{i \notin \actualS} (\actualX - \actualE_i)^{-1} \prod_{ i \in \actualS} (\actualX - \tilde{e}_i)^{-1} \eAnd\\
			\forall i \notin \actualS: |f^{-1}(-\actualMsg_i \beta)| \neq 1 \eAnd\\
			\forall i \in \actualS: f^{-1}(-\actualMsg_i \beta) = \{\actualE_i\}
		\end{array}\right]\\ 
		&\geq \frac{1}{p-1} \cdot \frac{1}{p-1} \cdot \frac{1}{p}  \cdot \frac{1}{p^{\qnum - |\actualS|}} \cdot \frac{1}{p^{|\actualS|}} \sum_{(\tilde{e}_i)_{i \in \actualS} \neq \actualX} (1 - \varepsilon) \frac{1}{p^{1 + |\actualS|}} \eta(\actualS)\\
		&= \frac{1}{p-1} \cdot \frac{1}{p-1} \cdot \frac{1}{p}  \cdot \frac{1}{p^{\qnum - |\actualS|}} \cdot \frac{(p - 1)^{|\actualS|}}{p^{|\actualS|}} \cdot (1 - \varepsilon) \frac{1}{p^{1 + |\actualS|}} \eta(\actualS)\\
		&=  (1 - \varepsilon) \left(1 - \frac{1}{p}\right)^{|\actualS|} \cdot \underbrace{(p - 1)^{-1} \cdot (p-1)^{-1} \cdot p^{-1}  \cdot p^{-\qnum + |\actualS|}  \cdot p^{-1 - |\actualS|} \eta(\actualS)}_{\interprob_0(\tau)}\;. 
\end{align*}
\else
\begin{align*}
	&\geq (p - 1)^{-1} \cdot (p-1)^{-1} \cdot p^{-1}  \cdot p^{-\qnum + |\actualS|} \cdot \colorbox{lightgray}{$p^{-|\actualS|}$} p^{-d} \colorbox{lightgray}{$\displaystyle \sum_{(\tilde{e}_i)_{i \in \actualS} \neq \actualX}$} (1 - \varepsilon) p^{-1 - |\actualS|} \eta(\actualS)\\
	&= (p - 1)^{-1} \cdot (p-1)^{-1} \cdot p^{-1}  \cdot p^{-\qnum + |\actualS|} \cdot \colorbox{lightgray}{$p^{-|\actualS|} (p - 1)^{|\actualS|}$} \cdot (1 - \varepsilon) p^{-1 - |\actualS|} \eta(\actualS)\\
	&=  (1 - \varepsilon)\colorbox{lightgray}{$\left(1 - \frac{1}{p}\right)^{|\actualS|}$} \cdot (p - 1)^{-1} \cdot (p-1)^{-1} \cdot p^{-1}  \cdot p^{-\qnum + |\actualS|}  \cdot p^{-1 - |\actualS|} \eta(\actualS) \;. 
\end{align*}
\fi
The \ifnum\eprintversion=1 second \else first \fi inequality follows from the inequality derived above and because there is only one $\beta$ that makes the indicator $1$.
The second to last equality follows from counting the number of $(\widetilde{e}_i)_{i \in \actualS}$ where $\widetilde{e}_i \neq \actualX$. The last equality follows by rearranging the terms. We also mark the modified terms with \ifnum\eprintversion=1 {\color{blue} blue} \else \colorbox{lightgray}{highlighting}\fi.

Therefore, \ifnum\eprintversion=1\[
 \frac{\interprob_1(\tau)}{\interprob_0(\tau)} \geq (1 - \varepsilon) \left(1 - \frac{1}{p}\right)^{|\actualS|}  \geq 1 - \varepsilon - \frac{|\actualS|}{p} \geq 1 - \frac{\qnum + 1}{p}\;.
\]\else
$ \frac{\interprob_1(\tau)}{\interprob_0(\tau)} \geq (1 - \varepsilon) \left(1 - \frac{1}{p}\right)^{|\actualS|}  \geq 1 - \varepsilon - \frac{|\actualS|}{p} \geq 1 - \frac{\qnum + 1}{p}$.
\fi The second inequality follows from $\prod_{i} (1 - x_i) \geq 1 - \sum_i x_i$ for $0 \leq x_i \leq 1$. The last one follows from $\varepsilon \leq 1/p$ and $|\actualS| \leq \qnum$.
\ifnum\eprintversion=1
 Hence, we have that \[
\SD(T_0, T_1) \leq \frac{(\qnum+ 1)^2 + 4\qnum}{2p}\;. \qed
\]
\else
 Hence, $\SD(T_0, T_1) \leq \frac{(\qnum+ 1)^2 + 2}{2p} + \frac{\qnum + 1}{p} \leq \frac{(\qnum+ 1)^2 + 4\qnum}{2p}$. \qed
\fi

%% file: figures/tight-reduction.tex
\begin{figure}[!t]
	\centering
	\begin{pchstack}[boxed]
	\procedure{\scriptsize Algorithm $\advB_1(G_1, (X_{1,i})_{i \in [d']}, G_2, X_{2,1})$:}{
		\mathsf{Sigs}, S \gets \emptyset \; ; \;
		\tilde{e}_1, \ldots, \tilde{e}_\qnum \getsr \Zp \; ; \; \beta \getsr \Zp^* 
		\; ; \; f \getsr \Zp^{d}[\varX] { }^\dagger \\
		\textstyle \overline{G}_1 \gets f(x) \prod_{i = 1}^\qnum (x - \tilde{e}_i) G_1 { }^\ddagger\; ; \; \overline H \gets \beta \prod_{i = 1}^\qnum (x - \tilde{e}_i) G_1   \\
		(m^*, \sigma^* = (A^*, e^*)) \getsr \advA^{\oraS}(\overline{G}_1, \overline H, G_2, X_{2,1})\\
		\textstyle i^* \gets \min \{i : e^* = e_i\}\\
		\pcif i^* = \bot \eOr i^* \notin S \eOr e^* \in \{\widetilde{e}_i\}_{i \in [\qnum]} \eOr (m^*, \sigma^*) \in \mathsf{Sigs} \; \eOr \\
		\pcind \pairingOp{A^*}{X_{2,1} + e^* G_2} \neq \pairingOp{\overline{G}_1 + m^* \overline H}{G_2} \pcthen \pcreturn \bot\\
		\pcreturn (e^*, (x - e^*)^{-1}G_1) \pccomment{Computed via \Cref{remark:compute-sdh}}
	}
	\procedure{\scriptsize Oracle $\oraS(m_i \in \Zp)$:}{
		\pcif |\Root(f(\varX) + m_i \beta) | = 1 \pcthen \\
		\pcind e_i \gets \Root(f(\varX) + m_i \beta) \\
		\pcind S \gets S \cup \{i\}\\
		\pcelse e_i \gets \widetilde{e}_i\\
		\textstyle A_i \gets \frac{(f(x) + m_i \beta) \prod_{i \in [\qnum]} (x - \widetilde{e}_i)}{x - e_i} G_1\\
		\mathsf{Sigs} \gets \mathsf{Sigs} \cup \{(m_i,  (A_i, e_i))\}\\
		\pcreturn (A_i, e_i)
	}
	\end{pchstack}
	\caption{Adversary $\advB_1$ running $\advA$. We assume that all algorithms implicitly take as input the group description, and $\advB_1$ implicitly counts the signing queries. $d' =  d + \qnum  \allowbreak = \lceil(e^2 + 2)\qnum + 2\log_2 p\rceil + 1$. We use $\Root(\cdot)$ as a $\Zp$-scalar instead if its size is 1. {$\dagger$: Our analysis does not require monic $f$, but we use it for consistency with the technical overview and readability as there are exactly $d$ coefficients, rather than $d + 1$. $\ddagger$ : At this point, it is not obvious that $\overline{G}_1$ is uniformly random. We will formally show this in our analysis but intuitively this is due to the degree $d$ of $f$ being relatively large.}} \label{fig:tight-reduction}
\end{figure}

%% file: tex_files/lower-bound.tex
\section{Impossibility of Tight Reductions for General Messages}\label{sec:lower-bound}

In this section, we consider the tightness of security reductions againsts adversaries that makes arbitrary (i.e., not necessarily distinct) signing queries. In particular, \Cref{sec:lower-bound-algebraic-straightline,sec:lower-bound-algebraic-rewinding} give lower bounds on the reduction loss for straight-line and rewinding ``algebraic" reductions, respectively. 

\subsection{Impossibility of Tight Straight-Line Algebraic Reductions}\label{sec:lower-bound-algebraic-straightline}
The following theorem shows that for any straight-line algebraic reduction $\calR$, with black-box access to an adversary $\advA$ that makes at most $\qnum$ identical signing queries, cannot achieve an advantage in breaking the $(d_1, d_2)$-$\SDH$ assumption greater than $O(\qnum^{-1}) \cdot \Adv^{\suf'}_{\BBS}(\advA, \secpar)$. Otherwise, there exists an adversary $\metaR$ breaking $(d_1, d_2)$-$\DL$ with non-negligible advantage. 
Accordingly, any straight-line algebraic reductions must incur at least a $\Theta(\qnum)$ factor loss in the advantage. 

Here, it suffices for us to show impossibility of tight $\SUF'$ security with message length $\ell = 1$, since it is tightly (and trivially) implied by $\SUF$ security for any $\ell > 1$. We also stress that the reduction $\calR$ {\em does not} obtain the algebraic representation of the group elements output by the adversary. (Note that there exists a tight reduction in the AGM~\cite{EC:TesZhu23b}.) Finally, \Cref{corol:generalization} extends this impossibility result to adversaries $\advA$ with fine-grained query patterns: making $\qnum/k$ queries each on $k$ distinct messages. In that case, the loss is at least $\Theta(\qnum/k)$. 

\begin{theorem}\label{thm:straightline-algebraic}
	Let $\GroupGen$ be a group parameter generator outputting bilinear groups of prime-order $p = p(\secpar)$, $d_1 = d_1(\secpar), d_2 = d_2(\secpar)$ be integers, and $\BBS_1 = \BBS[\GroupGen, 1]$. There exists an unbounded adversary $\advA$ making $\qnum = \qnum(\secpar)$ identical signing queries with $\Adv^{\suf'}_{\BBS_1}(\advA, \secpar) \geq 1 - \frac{\qnum^2}{p}$. For any straight-line algebraic reduction $\calR$ in the $(d_1, d_2)$-$\SDH$ game  with running time $t_\calR = t_\calR(\secpar)$ and running the adversary only once, there exists a meta-reduction $\metaR$ (in the AGM) against the $(d_1, d_2)$-$\DL$ assumption running in time roughly $t_\calR$ such that \[
		\Adv^{(d_1, d_2)\text{-}\dl}_{\GroupGen}(\metaR, \secpar) \geq \Adv^{(d_1, d_2)\text{-}\sdh}_{\GroupGen}(\calR^{\advA}, \secpar) - \frac{2}{\qnum} \;. 
	\]
\end{theorem}

\begin{proof} 
	We will discuss the reduction's output and representation in the AGM, describe the unbounded adversary $\advA$ and the meta-reduction $\metaR$ in \Cref{fig:meta-reduction-simple}, and analyze them. For simplicity, assume that all algorithms implicitly takes as input the group description $\pGroupDescript$.

	\heading{Reduction $\calR$.} Note that we are working in Type-3 pairing groups and assuming the AGM~\cite{C:FucKilLos18}. In particular, the reduction takes as input the group elements $(G_1 \in \GOne^*, (X_{1,i} \in \GOne)_{i \in [d_1]}, G_2 \in \GTwo^*, (X_{2,i} \in \GTwo)_{i \in [d_2]})$. It then outputs the following group elements along with their algebraic representation: \begin{itemize}[noitemsep,topsep=0pt]
		\item {\bf Public parameters and key:} The group elements $\overline{G}_1, \overline H, \overline{G}_2, \overline{X}_{2,1} $ are given along with their representation $\vec{f}, \vec{g} \in \Zp^{d_1 + 1}, \vec{a}, \vec{b} \in \Zp^{d_2 + 1}$ such that 
		\ifnum\eprintversion=1
		\begin{align}
		\overline{G}_1 = f_0 G_1 + \sum_{i = 1}^{d_1} f_i X_{1,i}\;&,\; \overline H = g_0 G_1 + \sum_{i = 1}^{d_1} g_i X_{1,i}\;, \label{eq:red-input-1}\\
		\overline{G}_2 = a_0 G_2 + \sum_{i = 1}^{d_2} a_i X_{2,i} \;&, \;\overline{X}_{2,1} = b_0 G_2 + \sum_{i = 1}^{d_2} b_i X_{2,i} \;. \label{eq:red-input-2}
		\end{align}
		\else
		\begin{align}
			\textstyle \overline{G}_1 = f_0 G_1 + \sum_{i = 1}^{d_1} f_i X_{1,i}\;&,\; \textstyle\overline H = g_0 G_1 + \sum_{i = 1}^{d_1} g_i X_{1,i}\;, \label{eq:red-input-1}\\
			\textstyle \overline{G}_2 = a_0 G_2 + \sum_{i = 1}^{d_2} a_i X_{2,i} \;&, \;\textstyle\overline{X}_{2,1} = b_0 G_2 + \sum_{i = 1}^{d_2} b_i X_{2,i} \;. \label{eq:red-input-2}
		\end{align}
		\fi
		In particular, $\vec f, \vec g$ define polynomials $f, g$ of degree at most $d_1$ and similarly $\vec a, \vec b$ define polynomials $a, b$\footnote{These extend from \Cref{sec:lower-bound-overview} on how the reduction can choose the verification key.} of degree at most $d_2$. We will refer to them as polynomials throughout the proof.
		
		\item {\bf Signing queries:} As a part of the output of the $i$-th signing query, the group element $A_i$ can be described with $\vec\alpha^{(i)} \in \Zp^{d_1 + 1}$ or equivalently a polynomial $\alpha^{(i)}$ of degree at most $d_1$ such that 
		\ifnum\eprintversion=1
		\[
		A_i = \alpha^{(i)}_0 G_1 + \sum_{j = 1}^{d_1} \alpha^{(i)}_j X_{1,i} = \alpha^{(i)}(x) G_1\;.
		\]
		\else
		$A_i = \alpha^{(i)}_0 G_1 + \sum_{j = 1}^{d_1} \alpha^{(i)}_j X_{1,i} = \alpha^{(i)}(x) G_1$.
		\fi
		
		\item {\bf The $(d_1, d_2)$-$\SDH$ solution:} The SDH solution $(e', Z^*)$ is described with $\vec{\zeta} \in \Zp^{d + 2}$ where
		\ifnum\eprintversion=1
		\[
			Z^* = \zeta_0 G_1 + \sum_{i = 1}^{d_1} \zeta_i X_{1,i} + \zeta_{d_1 + 2} A^*\;,
		\] \else
		$Z^* = \zeta_0 G_1 + \sum_{i = 1}^{d_1} \zeta_i X_{1,i} + \zeta_{d_1 + 2} A^*$,
		\fi with $A^*$ being the forgery that the adversary returns to the reduction.
		
	\end{itemize}
	
	\input{figures/meta-reduction-simple}

	\heading{Unbounded adversary $\advA$.} The adversary $\advA$ in \Cref{fig:meta-reduction-simple} finds the discrete logarithm $x = \dlog_{\overline{G}_2} \overline{X}_{2,1}$, makes $\qnum$ signing queries on a uniformly random message $m$, and forges on $m + 1$ using $x$ and a tag $e_{i^*}$ with $i^* \getsr [\qnum]$. 
	Here, $\advA$ wins in the $\SUF'$ game unless it aborts when signatures are invalid or the tags $e_i$'s collide (occurs with probability at most $\qnum^2/p$). Hence, $\Adv^{\suf'}_{\BBS_1}(\advA, \secpar) \geq 1 - \qnum^2/p$.
	
	\heading{Meta-reduction $\metaR$:} The meta-reduction $\metaR$, given in \Cref{fig:meta-reduction-simple}, forwards its inputs to $\calR$ and has access to $\calR$'s algebraic representations. It simulates $\advA$ by making $\qnum$ queries on a random message $m$ and trying to forge using a tag from a random query. We note that there are two bad events defined where $\metaR$ aborts: \begin{itemize}[noitemsep,topsep=0pt]
		\item $\BadRep$: There exists some $i \in [\qnum]$ where $(b - e_i a) \ndivides (f + m \cdot g)a$. If this occurs, $\metaR$ can find the discrete logarithm $x$ of $X_{2,1}$. 
		In particular, when $(A_i, e_i)$ verifies, $\pairingOp{A_i}{\overline{X}_{2,1} - e_i \overline{G}_2}= \pairingOp{\overline{G}_1 + m \overline{H}}{\overline{G}_2}$. Therefore, 
		\ifnum\eprintversion=1
		\[
		\alpha^{(i)}(x) \cdot  (b(x) - e_i a(x)) = a(x)(f(x) + m g(x))\;.
		\]\else
		$\alpha^{(i)}(x)  (b(x) - e_i a(x)) = a(x)(f(x) + m g(x))$.
		\fi Since $(b - e_i a) \ndivides (f + m g) a$, we have that $x$ is one of the zeros of the non-zero polynomial $h(\varX ) = \alpha^{(i)}(\varX) (b(\varX) - e_i a(\varX)) - (f(\varX) + m g(\varX))a(\varX)$, which $\metaR$ can find efficiently. 
		
		\item $\BadIdx$: The sampled index $i^*$ is such that $(b - e_{i^*} a)$ does not divide $fa$ or $g a$, or $\deg (b - e_{i^*} a) < \max\{\deg a,\deg b\}$. If this does not occur, $\metaR$ can compute $A^* = \alpha^*(x) G_1$ where $\alpha^*(\varX) = \frac{f (\varX) + m^* g  (\varX)}{b (\varX) - e_{i^*} a (\varX)} a(\varX) \allowbreak= \frac{f (\varX) + m^* g  (\varX)}{b (\varX)/a (\varX) - e_{i^*} }$. With the degree constraint, we have $\deg \alpha^* = \deg (f + m^* g) + \deg a - \deg (b - e_{i^*} a) \leq \deg (f + m^* g)  \leq d_1$, so $A^*$ can be computed from the $(d_1, d_2)$-DL instance.
	\end{itemize} 
	 If $\BadIdx$ does not occur, the view of $\calR$ is identical to when it is running $\advA$ since $\metaR$ outputs a valid forgery. With $A^* = \alpha^*(x) G_1$, the representation $\vec \zeta$ of $Z^*$ can be viewed as a polynomial $\zeta(\varX)$ with $\deg \zeta \leq d_1$. Hence, $\metaR$ succeeds in the $(d_1, d_2)$-$\DL$ game, if $\calR$ succeeds in the $(d_1, d_2)$-$\SDH$ game and $\BadIdx$ does not occur. Also, if $\BadRep$ occurs, $\metaR$ also succeeds as discussed above. 
	 \ifnum\eprintversion=1 Also, it is easy to see that the running time of $\metaR$ depends on $t_\calR$, the time for finding the zeros of $\zeta(\varX) (\varX - e') - 1$, and the time for checking if $b - e_{i^*} a$ divides $f$ and $g$. These are dominated by $t_\calR$.
	 \else
	 Also, it is easy to see that the running time of $\metaR$ is dominated by $t_\calR$.
	 \fi
	 Finally, we bound the advantage of $\metaR$ and conclude the proof via the following lemma.
	\[
		\Adv^{(d_1,d_2)\text{-}\dl}_{\GroupGen}(\metaR, \secpar) \geq \Adv^{(d_1,d_2)\text{-}\sdh}_{\GroupGen}(\calR^{\advA}, \secpar) - \Pr[\BadIdx \eAnd \neg \BadRep] \;. \pcind \qed
	\]
	
	\begin{lemma}\label{lem:guessing-index}
		$\Pr[\BadIdx \eAnd \neg \BadRep] \leq \frac{2}{\qnum}$.
	\end{lemma}
	\begin{proof}[of \Cref{lem:guessing-index}]
		We fix any input and random tape of $\calR$, which fixes the polynomials $f, g, a, b$. Then, we show that for any such fixed polynomials, the probability that $\BadIdx \eAnd \neg \BadRep$ occurs is at most $2/\qnum$. By definition, $\BadIdx \eAnd \neg \BadRep$ implies (1) $\deg (b - e_{i^*} a) < \max\{\deg a, \deg b\}$, or (2) $(b - e_{i^*} a \divides (f + m g) a) \eAnd 
		(b - e_{i^*} a \ndivides f a \eOr b - e_{i^*} a \ndivides ga )$. Our approach is to upper bound the number of $e$'s that can lead to such event, which bounds the probability over $i^* \getsr [\qnum]$, and take the expectation over $m \gets \Zp$. For part (1) of the event, it is clear that only one $e$ can decrease the degree, i.e., by canceling out the leading coefficient. Thus, for fixed $f,g,a,b$ and distinct $e_1, \ldots, e_\qnum$, $\Pr_{i^* \getsr [\qnum]}[ \deg (b - e_{i^*} a) < \max\{\deg a, \deg b\} ] \leq 1/\qnum$.
		
		For part (2) of the event, define for fixed $f \neq 0, g, a \neq 0, b$ and $m \in \Zp$  \ifnum\eprintversion=1\[
		\varphi_{f,g,a,b}(m) := \{e \in \Zp :  (b - e a) \divides (f + m g) a \eAnd ((b - e a) \ndivides fa \eOr ( b - e a) \ndivides g a)\}\;.\]
		\else , the set $\varphi_{f,g,a,b}(m)$ of $e \in \Zp$ such that $(b - e a) \divides (f + m g) a \eAnd ((b - e a) \ndivides fa \eOr ( b - e a) \ndivides g a)$. \fi
		Then, for $m \neq m' \in \Zp$, $\varphi_{f,g,a,b}(m) \cap \varphi_{f,g,a,b}(m') = \emptyset$. This is because if there exists $e \in \varphi_{f,g,a,b}(m) \cap \varphi_{f,g,a,b}(m')$, we have that $(b - e a) \divides (f + m g)a$ and $(b - e a) \divides (f + m' g)a$, so $b - e a$ divides both $fa$ and $ga$ (a contradiction). Hence, $\sum_{\actualMsg \in \Zp} |\varphi_{f,g,a,b}(\actualMsg)| \leq p$. Accordingly, for any fixed $f \neq 0, g, a \neq 0, b \in \Zp[\varX]$, $\actualMsg \in \Zp$, and distinct $e_1, \ldots, e_\qnum \in \Zp$ (otherwise, $\advA$ aborts), we have \ifnum\eprintversion=1 \[
		\Pr_{i^* \getsr [\qnum]}\left[ \begin{array}{l}
			(b - e_{i^*} a) \divides (f + \actualMsg g)a \; \eAnd \\
			((b - e_{i^*} a )\ndivides fa \eOr (b - e_{i^*} a) \ndivides ga) 
		\end{array} \right]  \leq \frac{|\varphi_{f,g,a,b}(\actualMsg)|}{\qnum} \;. 
		\] \else 
		$\Pr_{i^* \getsr [\qnum]}[ e_{i^*} \in \varphi_{f,g,a,b}(\actualMsg)] \leq  |\varphi_{f,g,a,b}(\actualMsg)| / p$.
		\fi Taking expectation over all messages $m \in \Zp$ results in the $1/\qnum$ bound, and we conclude by applying the union bound.\qed
	\end{proof}

\end{proof}

\begin{corollary}[Generalized query patterns]\label{corol:generalization}
	\ifnum\eprintversion=1
	Let $\GroupGen$ be a group parameter generator, producing groups of order $p = p(\secpar)$, $d_1 = d_1(\secpar), d_2 = d_2(\secpar)$ be integers, and $\BBS_1 = \BBS[\GroupGen, 1]$. 
	\else
	Let $\GroupGen$, $p$, $d_1, d_2$, and $\BBS_1$ be as in \Cref{thm:straightline-algebraic}. 
	\fi 
	There exists an unbounded adversary $\advA$ making in total $\qnum  = \qnum(\secpar)$ queries on at least $k = k(\secpar)$ distinct messages with $\qnum/k$ queries per message, such that $\Adv^{\suf'}_{\BBS_1}(\advA, \secpar) \geq 1 - \frac{\qnum^2}{p}$. For any straight-line algebraic reduction $\calR$ playing the $(d_1, d_2)$-$\SDH$ game with running time $t_\calR = t_\calR(\secpar)$ and running the adversary only once, there exists a meta-reduction $\metaR$ against the $(d_1, d_2)$-$\DL$ assumption running in time roughly $t_\calR$ such that \ifnum\eprintversion=1\[
	\Adv^{(d_1, d_2)\text{-}\dl}_{\GroupGen}(\metaR, \secpar) \geq \Adv^{(d_1, d_2)\text{-}\sdh}_{\GroupGen}(\calR^{\advA}, \secpar) - \frac{k + 1}{\qnum} \;. 
	\]\else
	$	\Adv^{(d_1, d_2)\text{-}\dl}_{\GroupGen}(\metaR, \secpar) \geq \Adv^{(d_1, d_2)\text{-}\sdh}_{\GroupGen}(\calR^{\advA}, \secpar) - \frac{k + 1}{\qnum}$. 
	\fi
\end{corollary}
\begin{proof}
	We only sketch the changes to the proof of  \Cref{thm:straightline-algebraic} as follows: \begin{itemize}[topsep=0pt]
		\item {\bf Unbounded adversary $\advA$ and meta-reduction $\metaR$:} Both now sample uniformly random distinct messages $ (\mu_1, \ldots, \mu_{k}) \getsr \Zp^k$ and make $\qnum/k$ queries each, i.e., the queries are $\vec{m} \gets (\underbrace{\mu_1, \ldots, \mu_1}_{\qnum/k \text{ times}}, \ldots, \underbrace{\mu_k, \ldots, \mu_k}_{\qnum/k \text{ times}})$. They both forge on the minimum element $m^*$ in $\Zp \setminus \{\mu_1,\ldots,\mu_{k}\}$. The overall strategies for both are unchanged.
		
		\item {\bf Events $\BadIdx$ and $\BadRep$.} The two events are now defined as: ($\BadRep$) There exists $i \in [\qnum]$ such that $(b - e_{i} a) \ndivides (f + m_{i} g)a$, and ($\BadIdx$) For $i^* \allowbreak \getsr [\qnum]$, $b - e_{i^*} a$ does not divide $fa$ or $ga$ or $\deg (b - e_{i^*} a) = \max\{\deg a, \deg b\}$. The bound in \Cref{lem:guessing-index} changes to $(k + 1)/\qnum$ instead of $2/\qnum$. Intuitively, the term $k/\qnum$ appears because with $k$ distinct messages, there are now $k$ times more chances (by linearity of expectation) for the reduction to choose $e_i$ where $b + e_i a$ does not divide $fa$ or $ga$ but divides $(f + m_{i} g)a$. \qed
	\end{itemize}
\end{proof}

\subsection{Impossibility of Tight Rewinding Algebraic Reductions}\label{sec:lower-bound-algebraic-rewinding}
In this section, we extend our impossibility results to rewinding algebraic reductions that run the adversary at most $r$ times with the same random tape, with the ability to interleave these executions. The following theorem establishes that any such reduction incurs at least a $\Theta(\qnum / r^2)$ factor loss in their advantage. (For fine-grained query patterns, this scales linearly with $1/k$ as in \Cref{corol:generalization}.) Otherwise, the meta-reduction breaks the $(d_1, d_2)$-$\DL$ assumption. This result also encompasses those from the previous section, which we presented separately to illustrate our main ideas with simpler proofs before progressing to the more complex proof in \ifnum\eprintversion=0 the full version\else\Cref{sec:proof-rewinding-lower-bound}\fi.

\begin{theorem}\label{thm:rewinding-algebraic}
	Let $\GroupGen$, $p$, $d_1, d_2$, and $\BBS_1$ be as in \Cref{thm:straightline-algebraic}. 
	There exists an unbounded adversary $\advA$ making in total $\qnum  = \qnum(\secpar)$ queries on at least $k = k(\secpar)$ distinct messages with $\qnum/k$ queries per message, such that $\Adv^{\suf'}_{\BBS_1}(\advA, \secpar) \geq 1 - \frac{\qnum^2}{p}$. Moreover, for any algebraic reduction $\calR$ playing the $(d_1, d_2)$-$\SDH$ game, running in time $t_\calR = t_\calR(\secpar)$, and running the adversary at most $r = r(\secpar)$ times (possibly interleaving the executions), there exist meta-reductions $\metaR$ and $\metaR'$ (both in the AGM) against the $(d_1, d_2)$-$\DL$ assumption running in time roughly $t_\calR$ such that \ifnum\eprintversion=1 \[
	\Adv^{(d_1, d_2)\text{-}\dl}_{\GroupGen}(\metaR, \secpar) \geq \Adv^{(d_1, d_2)\text{-}\sdh}_{\GroupGen}(\calR^{\advA}, \secpar) - \frac{r^2 \cdot (k + 1)}{\qnum} - \Adv^{(d_1, d_2)\text{-}\dl}_{\GroupGen}(\metaR', \secpar) \;.
	\]\else 
	$\Adv^{(d_1, d_2)\text{-}\dl}_{\GroupGen}(\metaR, \secpar) \geq \Adv^{(d_1, d_2)\text{-}\sdh}_{\GroupGen}(\calR^{\advA}, \secpar) - \frac{r^2 \cdot (k + 1)}{\qnum} - \Adv^{(d_1, d_2)\text{-}\dl}_{\GroupGen}(\metaR', \secpar) $.
	\fi
\end{theorem}

\begin{proof}[Sketch]
	We give a proof sketch here for the simpler case of $k = 1$ and prove the general case in \ifnum\eprintversion=0 the full version\else\Cref{sec:proof-rewinding-lower-bound}\fi, following a similar argument from the straight-line reduction case. 
	Similar to before, the (algebraic) reduction $\calR$ and unbounded adversary $\advA$ interact in the following steps in one run of $\advA$: \begin{enumerate}[topsep=0pt, label=(\arabic*)] 
		\item $\calR$ sends $\overline{G}_1, \overline H, \overline{G}_2, X_{2,1}$ committing to polynomials $f, g, a, b$, respectively.
		\item $\advA$ replies with a random message $m \getsr \Zp$, requesting $\qnum$ signatures.
		\item $\calR$ sends valid signatures $(A_i, e_i)_{i \in [\qnum]}$ such that $(a - e_i b) \divides (f + m g)a$. 
		\item $\advA$ forges a signature on the tag $e_{i^*}$ where $i^* \getsr [\qnum]$. 
	\end{enumerate}
	However in this case, $\calR$ can rewind $\advA$ and give group different elements $\overline{G}_1, \overline H, \allowbreak \overline{G}_2,  X_{2,1}$ at Step (1) or different signatures $(A_i, e_i)_{i \in [\qnum]}$ at Step (3).\footnote{Sending the same elements over rewound runs does not increase the reduction's success probability, as $\advA$ can pick its outputs by evaluating a random function (defined by its random-tape) on its view so far. The meta-reduction can efficiently simulate this via lazy-sampling.} Our meta-reduction would also replicate the adversary $\advA$ except that it can only forge when $e_{i^*}$ is such that $(a - e_{i^*} b) $ divides both $fa$ and $ga$ in each run.
	
	As we will formally prove, the optimal strategy for the reduction is to rewind to Step (1) $O(r)$ times and Step (3) $O(r)$ times. 
	In essence, the rewindings to Step (1) allow the reduction to find the tuple $(f(\varX), g(\varX), a(\varX), b(\varX), m \in \Zp)$ such that there are many tags $e \in \Zp$ for which $(a - e b)$ divides $(f + m g)a$ but does not divide $fa$ or $ga$. We show that after $r$ rewindings to Step (1), the largest number of such tags $e$ over the $r$ rewound runs is at most $r$ in expectation. Then, $r$ additional rewindings to Step (3) of the run that produces the best tuple $(f,g, a,b, m)$ boost the failure probability of the meta-reduction by a factor of $r$, resulting in the $r^2$ factor in the bound. 
\end{proof}

\ifnum\eprintversion=1
\input{tex_files/proof-lower-bound-rewinding}
\fi

%% file: figures/meta-reduction-simple.tex
\begin{figure}[!t]
	\centering
	\begin{pchstack}[boxed,space=2pt]
	\begin{pcvstack}
	\procedure{\scriptsize Algorithm $\metaR(\inp)$:}{
		(G_1, (X_{1,i})_{i \in [d_1]}, G_2, (X_{2,i})_{i \in [d_2]}) \gets \inp\\
		(e', Z^*, \vec{\zeta}) \getsr \calR^{\advB}(\inp)\\
		\pccomment{$\textstyle Z^* = \zeta_0 G_1 + \sum_{i = 1}^{d_1} \zeta_i X_{1,i}  + \zeta_{d_1 + 2} A^*$}\\
		\pcreturn x \text{ s.t. } xG_1 = X_{1,1} \eAnd \zeta(x)(x - e') = 1
	}
	\procedure{\scriptsize Algorithm $\advA^{\oraS}(\overline{G}_1, \overline H, \overline{G}_2, \overline{X}_{2,1})$:}{
		\pcif \overline{G}_1 = 0_{\GOne} \eOr \overline{G}_2 = 0_{\GTwo} \pcthen \pcreturn \bot\\
		m \getsr \Zp \; x \gets \dlog_{\overline{G}_2}(\overline{X}_{2,1})\\
		\pcfor i \in [\qnum]: (\sigma_i = (A_i, e_i)) \getsr \oraS(m)\\
		 \pcif (\exists i \in [\qnum ] : \textstyle \Ver{\BBS}(\overline{X}_{2,1}, m, \sigma_i) = 0 ) \eOr \\
		\pcind[2] (\exists i \neq j \in [\qnum] : e_i = e_j) \pcthen \pcreturn \bot \\
		i^* \getsr [\qnum]\; ; \; m^* \gets m + 1\\
		\textstyle \pcreturn (m^*, (A^* = \frac{\overline{G}_1 + m^* \overline H}{x - e_{i^*}}, e_{i^*}))
	}
	\end{pcvstack}
	\procedure{\scriptsize Algorithm $\advB^{\oraS}(\overline{G}_1, \overline H, \overline{G}_2, \overline{X}_{2,1}, \vec{f}, \vec{g}, \vec{a}, \vec{b})$:}{
		\pccomment{$\vec f, \vec g, \vec a, \vec b$ satisfy (\ref{eq:red-input-1},\ref{eq:red-input-2}) , $\vec f, \vec a \neq \vec 0$}\\
		\pcif \overline{G}_1 = 0_{\GOne} \eOr \overline{G}_2 = 0_{\GTwo} \pcthen \pcreturn \bot\\
		m \getsr \Zp\; 	; \;	i^* \getsr [\qnum]\; ; \; m^* \gets m + 1\\
		\pcfor i \in [\qnum]:  (\sigma_i = (A_i, e_i), \vec \alpha^{(i)}) \getsr \oraS(m)\\
		\ifnum\eprintversion=1 \pccomment{$\textstyle A_i = \alpha^{(i)}_0 G_1 + \sum_{j = 1}^{d_1} \alpha^{(i)}_j X_{1,i} = \alpha^{(i)}(x) G_1$}\\ \fi
		 \pcif \textstyle (\exists i \in [\qnum] : \Ver{\BBS}((\overline{G}_1, \overline H, \overline{G}_2), \overline{X}_{2,1}, m, \sigma_i) = 0)  \\ 
		 \pcind \eOr(\exists i \neq j \in [\qnum] : e_i = e_j) \pcthen \pcreturn \bot\\
		 \pcif \exists i \in [\qnum] : (b - e_i a) \ndivides (f + m g) a \pcthen \pccomment{$\BadRep$}\\
		\pcind[2] \metaR \; \pcabort \text{and } \pcreturn x \text{ where }  xG_1 = X_{1,1} \; \eAnd  \\
		\pcind[3]\alpha^{(i)}(x) (b(x) - e_i a(x)) = (f(x) + m g(x)) a(x)\\
		\pcif (b - e_{i^*} a) \divides f a \eAnd (b - e_{i^*} a) \divides g a \; \eAnd \\
		\pcind \deg (b - e_{i^*} a) = \max\{\deg a, \deg b\}\pcthen\\
		\pcind \textstyle \pcreturn (m^*, (A^* = \frac{f (x) + m^* g (x)}{b(x) - e_{i^*} a(x)} \cdot a(x) G_1, e_{i^*}))\\
		\pcelse \metaR \; \pcabort \pcind[10] \pccomment{$\BadIdx$}
	}
	\end{pchstack}
	\caption{Meta-reduction $\metaR$, unbounded adversary $\advA$, and simulated adversary $\advB$. As discussed in the body, all representation vectors can be parsed as a polynomial (i.e., for a vector $\vec \zeta$, we write $\zeta$ as the corresponding polynomial). We distinguish between \ifnum\eprintversion=1 the meta-reduction\fi $\metaR$ and \ifnum\eprintversion=1 the simulated adversary\fi $\advB$ aborting with $\pcabort$ and $\pcreturn \bot$, respectively.} \label{fig:meta-reduction-simple}
\end{figure}

%% file: tex_files/proof-lower-bound-rewinding.tex
\subsection{Proof of \Cref{thm:rewinding-algebraic}}\label{sec:proof-rewinding-lower-bound}
For simplicity, assume that all algorithms implicitly takes as input the group description $\pGroupDescript$ and that $k \divides \qnum$.
First and similarly to the proof of \Cref{thm:straightline-algebraic}, we discuss the group elements output by the reduction $\calR$ and their representations. Then, we describe the unbounded adversary $\advA$ and the meta-reduction $\metaR$, also given in \Cref{fig:unbounded-adv-rewind,fig:meta-reduction-rewind}. Note that by now, we assume that the readers are familiar with the notations and underlying ideas of the proof of \Cref{thm:straightline-algebraic}, as we will refer back to these ideas. We also without loss of generality assume that for each run of $\advA$, the reduction receives a forgery $\sigma^{(*, j)} = (A^{(*, j)}, e^{(*, j)})$---in the case that $\advA$ or the reduction aborts that run, we let $A^{(*, j)} = 0_{\GOne}$. Also, throughout this proof, we will denote  the values defined in the $j$-th rewinding of $\advA$ with the superscript $(j)$. 

\heading{Reduction $\calR$.} 
The reduction takes as input the SDH instance $(G_1 , (X_{1,i})_{i \in [d_1]}, \allowbreak G_2 , (X_{2,i})_{i \in [d_2]})$ and runs $\advA$ at most $r$ times {\em with the ability to interleave the executions}. Then, for the $j$-th run ($j \in [r]$) of the adversary, the settings of the reduction $\calR$ is as follows: 
\begin{itemize}
	\item {\bf Public parameters and key:} These are the group elements $\overline{G}_1^{(j)}, \overline{H}^{(j)}, \allowbreak \overline{G}_2^{(j)}, \overline{X}_{2,1}^{(j)}$ output along with their algebraic descriptions $\vec{f}^{(j)}, \vec{g}^{(j)} \in \Zp^{d_1 + r + 1}, \allowbreak \vec{a}^{(j)}, \vec{b}^{(j)} \in \Zp^{d_2 + 1}$. The representations describe these elements with respect to the group elements the reduction has previously seen including the SDH instance and the forgeries from previously completed runs denoted $A^{(*, n)}$ for the forgery in the $n$-th run. 
	We note that if the $n$-th run ($n \in [r]$) is still incomplete (so $A^{(*, n)}$ is undefined), we let $f_{d_1 + n}^{(j)} = 0$ for simplicity. 
	Then, the representations satisfy the following equations
	\begin{align}
		\overline{G}_1^{(j)} &= f_0^{(j)} G_1 + \sum_{n = 1}^{d_1} f_n^{(j)} X_{1,n} + \sum_{n = 1}^r f_{d_1 + n}^{(j)} A^{(*, n)}\;,\label{eq:red-rew-input-1}\\
		\overline{H}^{(j)} &= g_0^{(j)} G_1 + \sum_{n = 1}^{d_1} g_n^{(j)} X_{1,n} + \sum_{n = 1}^r g_{d_1 + n}^{(j)} A^{(*, n)}\;, \\
		\overline{G}_2^{(j)} &= a_0^{(j)} G_2 + \sum_{n = 1}^{d_2} a_n^{(j)} X_{2,n}\;, \pcind[2]
		\overline{X}_{2,1}^{(j)} = b_0^{(j)} G_2 + \sum_{n = 1}^{d_2} b_n^{(j)} X_{2,n} \;. \label{eq:red-rew-input-2}
	\end{align}
	Also, if for all completed run $n$, $A^{(*, n)}$ is computed as a linear combination of $G_1, X_{1,1},\ldots,X_{1,d_1}$, we can still describe $\overline{G}_1^{(j)}, \overline{H}^{(j)}$ with polynomials $f^{(j)}, g^{(j)}$ of degree at most $d_1$. Moreover, $\overline{G}_2^{(j)}, \overline{X}_{2,1}^{(j)}$ can always be described with polynomials $a^{(j)}, b^{(j)}$ of degree at most $d_2$. Throughout the proof, we will refer to $f^{(j)}, g^{(j)}, a^{(j)}, b^{(j)}$ as polynomials instead.
	
	\item {\bf Signing queries:} As the output of the $i$-th signing query, the group element $A_i^{(j)}$ can be described with $\vec\alpha^{(i, j)} \in \Zp^{d_1 + r + 1}$ such that \[
	A_i^{(j)} = \alpha_0^{(i, j)} G_1 + \sum_{n = 1}^{d_1} \alpha_n^{(i, j)} X_{1,n} + \sum_{n = 1}^r \alpha_{d_1 + n}^{(i,j)} A^{(*, n)}\;,
	\] Similar to the above, we let $\alpha_{d_1 + n}^{(i,j)} = 0$ if $A^{(*, n)}$ is not yet defined. Also, if $A^{(*, n)}$'s are computed as a linear combination of $G_1, X_{1,1},\ldots,X_{1,d_1}$, we can write $A_i^{(j)} = \alpha^{(i,j)}(x) G_1$ for some polynomial $ \alpha^{(i,j)}$ of degree at most $d_1$.

\end{itemize}
\noindent {\bf The $(d_1, d_2)$-$\SDH$ solution:} Finally, $\calR$ returns the SDH solution $(e', Z^*)$ which can be described with $\vec{\zeta} \in \Zp^{d_1 + r + 1}$ where \[
Z^* = \zeta_0 G_1 + \sum_{n = 1}^{d_1} \zeta_n X_{1,n} + \sum_{j = 1}^r \zeta_{d_1 + j} A^{(*, j)}\;.
\] 

\input{figures/meta-reduction-rewind}

\heading{Unbounded adversary $\advA$.} The adversary $\advA$, given in \Cref{fig:unbounded-adv-rewind}, is similar to the one in the non-rewinding proof. However, its random coins are now two (exponentially large) random functions 
$F_1: \GOne^2 \times \GTwo^2 \to \Zp^k , F_2 : \GOne^2 \times \GTwo^2 \times (\Zp^2 \times \GOne)^\qnum \to [\qnum]$. These in particular are used to control randomnesses for different rewound runs of itself (without knowing that it has been rewound). The key steps to $\advA$ are as follows:
\begin{itemize}
	\item {\bf Processing the public parameters and the verification key:} Compute $x = \dlog_{\overline{G}_2} \overline{X}_{2,1}$. The adversary $\advA$ aborts if $\overline G_1 = 0_{\GOne}$ or $ \overline G_2 = 0_{\GTwo}$. (This does not occur in the real $\SUF'$ game.)
	
	\item {\bf Selecting signing query:} Compute 
	$(\mu_l)_{l \in [k]} \gets F_1(\overline G_1, \overline H, \overline{G}_2, \overline{X}_{2,1})$ and setup $\vec{m} \gets (\underbrace{\mu_1, \ldots, \mu_1}_{\qnum/k \text{ times}}, \ldots,\allowbreak \underbrace{\mu_k, \ldots, \mu_k}_{\qnum/k \text{ times}})$. Then, request $\qnum$ signing queries $(A_i, e_i) \gets \oraS(m_i)$. If $e_i = e_{i'}$ for $i \neq i' \in [\qnum]$ or the signatures do not verify or $\mu_i = \mu_{i'}$ for some $i \neq i' \in [k]$, $\advA$ aborts. 
	
	\item {\bf Forgery:} Compute 
	$i^* \gets F_2(\overline{G}_1, \overline H, \overline{G}_2, \overline{X}_{2,1}, \{(m_i, e_i, A_i)\}_{i \in [\qnum]})$\footnote{The tuples of messages and signatures are sorted before evaluating $F_2$. This is to avoid the reduction gaining advantage by simply reordering the oracle replies over several runs.}, set 
	$m^* \gets \min(\Zp\setminus\{m_l\}_{l \in [k]})$, i.e., the minimum element not used as queries, $e^* = e_{i^*}$ and $A^* = \frac{\overline G_1 + m^* \overline H}{x - e^*}$.
\end{itemize}
It is easy to see that $\advA$ wins in the $\SUF'$ game unless it aborts (when $e_i$'s or $m_i$'s contain duplicates---this occurs with probability  at most $(\binom{\qnum}{2} + \binom{k}{2})/p  \leq \qnum^2 / p$). Hence, $\Adv^{\suf'}_{\BBS_1}(\advA, \secpar) \geq 1 - \qnum^2/p$. 

\heading{Meta-reduction $\metaR$:} 
Now, we describe the meta-reduction (also given in \Cref{fig:meta-reduction-rewind}) playing the $(d_1, d_2)$-$\DL$ game and running the reduction $\calR$ as follows: \begin{itemize}
	\item {\bf $(d_1, d_2)$-SDH instance for $\calR$:} This is simply forwarding its $(d_1, d_2)$-$\DL$ instance $G_1, X_{1,1}, \ldots, X_{1,d_1}, G_2, \allowbreak X_{2,1}, \ldots, X_{2,d_2}$. 
	
	\item {\bf Simulating random functions $F_1$ and $F_2$:} Since the meta-reduction is not unbounded, it needs to simulate the random functions $F_1$ and $F_2$ to be consistent over all simulated runs. This is done by lazy-sampling, i.e., when computing $F_1(\cdot), F_2(\cdot)$ on a new input, it samples a uniformly random element from the co-domain of the functions and records the input-output pair in a table.  To compute the function on an input for the second time, the meta-reduction looks up the table for the recorded value. In the pseudocode, we model these as random oracles that $\metaR$ simulates. 
	
	\item {\bf For each run of $\advA$:} Denoted with the simulated adversary $\advB$ which the reduction is given access to. The runs are indexed by $j \in [r]$. \begin{itemize}
		\item The reduction $\calR$ returns the group elements $\overline{G}_1^{(j)}, \overline{H}^{(j)}, \overline{G}_2^{(j)}, \overline{X}_{2,1}^{(j)}$.
		Here, $\metaR$ aborts {\em this particular run} if $\overline{G}_1^{(j)}$ and $\overline{G}_2^{(j)}$ are not generators.
		
		Note that these elements come with their algebraic representations which defines the polynomials $f^{(j)}, g^{(j)}, a^{(j)}, b^{(j)}$ (assuming that all the forgeries $A^{(*,j)}$ returned by $\metaR$ are constructed as linear combinations of $X_{1,i}$). We may assume without loss of generality that these are consistent over all runs for the same group elements; otherwise, the meta-reduction can stick to the first representation. 
		
		\item {\bf Selecting the signing query:} Compute 
		$(\mu^{(j)}_l)_{l \in [k]} \getsr F_1(\overline{G}_1^{(j)}, \overline{H}^{(j)}, \allowbreak \overline{G}_2^{(j)}, \overline{X}_{2,1}^{(j)})$ (with $F_1$ simulated as described above). Then, structure the message queries $\vec{m}^{(j)}$ exactly as $\advA$ does. 
		
		\item {\bf Processing the signatures:} This is done after all the queries have been made. The meta-reduction replicates the aborts in {\em a particular run} of $\advA$ if one of the following occurs: \begin{enumerate}[label=(\alph*)]
			\item There exists $i \in [\qnum]$ such that the signature $(A_i^{(j)}, e_i^{(j)})$ is not valid for $m^{(j)}_i$.
			
			\item There exists $i \neq i' \in [\qnum]$ such that $e_i^{(j)} = e_{i'}^{(j)}$.
			
			\item There exists $i \neq i' \in [k]$ such that $\mu_i^{(j)} = \mu_{i'}^{(j)}$.
		\end{enumerate}
		
		
		Moreover, if for some $i \in [\qnum]$, $(b^{(j)} - e_i^{(j)} \cdot a^{(j)}) \ndivides (f^{(j)} + m^{(j)}_i \cdot g^{(j)}) a^{(j)}$, the meta-reduction aborts {\em its whole algorithm}.\footnote{We make note of the distinction of the meta-reduction {\em aborting its whole algorithm} and it {\em aborting a particular rewound run}. The former means aborting its whole interaction with $\calR$. However, the latter means simulating the same aborts that $\advA$ would do, and $\calR$ is still free to run more instances of $\advA$.} 
		We denote the event that {\em this latter abort occurs in any of the runs} as $\BadRep$. 
		
		\item {\bf Forging a signature:} 
		Select $m^{(*,j)}$ to be the minimal element of the set $\Zp\setminus\{\mu_l^{(j)}\}_{l \in [k]}$ (i.e., an unused message), and compute $i^{(j)}$ from $F_2$. 
		If $i^{(j)}$ is such that $(b^{(j)} + e_{i^{(j)}}^{(j)} a^{(j)})$ does not divide $f^{(j)}a^{(j)}$ or $g^{(j)}a^{(j)}$, or the degree of $b^{(j)} - e_{i^{(j)}}^{(j)} a^{(j)}$ is less than the maximum of $\deg b^{(j)}$ and $\deg a^{(j)}$, the meta-reduction aborts {\em its whole algorithm}.
		
		If neither $\BadIdx$ nor $\BadRep$ occurs, we have that the polynomial $(b^{(j)} + e_{i^{(j)}}^{(j)} a^{(j)})$ divides both $f^{(j)}a^{(j)}$ and $g^{(j)}a^{(j)}$ and the degree of $b^{(j)} - e_{i^{(j)}}^{(j)} a^{(j)}$ is exactly the maximum of $\deg b^{(j)}$ or $\deg a^{(j)}$. Then, let \[
		\alpha^{(*,j)}(\varX) = \frac{f^{(j)}(\varX) + m^{(*,j)} g^{(j)}(\varX)}{b^{(j)}(\varX) - e_{i^{(j)}}^{(j)} a^{(j)}(\varX)}a^{(j)}(\varX) = \frac{f^{(j)}(\varX) + m^{(*,j)} g^{(j)}(\varX)}{b^{(j)}(\varX) / a^{(j)}(\varX)- e_{i^{(j)}}^{(j)}}. \] Similar to the prior proof, we can see that  $\alpha^{(*,j)}$ is a polynomial of degree at most $d_1$ from the fact that the degree of $b^{(j)} - e_{i^{(j)}}^{(j)} a^{(j)}$ does not decrease; i.e., \[
		\deg \alpha^{(*,j)} = \underbrace{\deg \left(f^{(j)} + m^{(*,j)} g^{(j)}\right)}_{\leq d_1} \underbrace{-  \deg \left(b^{(j)} - e_{i^{(j)}}^{(j)} a^{(j)}\right) +  \deg a^{(j)}}_{\leq 0} \leq d_1
		\] Hence, the meta-reduction can compute the forgery \[
			\sigma^{(*,j)} = (A^{(*,j)} =  \alpha^{(*,j)}(x) G_1, \allowbreak e^{(*,j)} = e_{i^{(j)}}^{(j)})\;,
		\] from its $(d_1, d_2)$-$\DL$ instance. Note that it is easy to verify the validity of $\sigma^{(*,j)}$.
	\end{itemize}		
	
	\item {\bf Computing $(d_1, d_2)$-$\DL$ solution:} Unless the meta-reduction aborts, the SDH solution $(e', Z^*)$ from the reduction $\calR$ can be written as $Z^* = \zeta(x) G_1$ for a degree at most $d_1$ polynomial $\zeta$. This is simply because $A^{(*,j)}$ can be written as $\alpha^{(*,j)}(x) G_1$ with $\deg \alpha^{(*,j)} \leq d_1$ (assuming that neither $\BadIdx$ nor $\BadRep$ occurs). Now, since $Z^* = (x - e')^{-1}G_1$, we have that $\zeta(x) (x - e') = 1$. Thus, the discrete logarithm $x$ can be found by computing the zeros of the polynomial $\zeta(\varX) (\varX - e') - 1$ and returning the one satisfying $xG_1 = X_{1,1}$
\end{itemize}

Note that if neither $\BadIdx$ nor $\BadRep$ occur, $\calR$'s view is identical to its interaction with $\advA$. Also, if $\calR$ wins in the $(d_1, d_2)$-$\SDH$ game, $\metaR$ also wins in the $(d_1, d_2)$-$\DL$ game. Hence, \[
\Adv^{d\text{-}\dl}_{\GroupGen}(\metaR, \secpar) \geq \Adv^{d\text{-}\sdh}_{\GroupGen}(\calR^{\advA}, \secpar) - \Pr[\BadIdx \eOr \BadRep] \;.
\]
To bound the probability of $\BadIdx \eOr \BadRep$, we first define for each $j \in [r]$, the index $\runscompleted[j] \in [r]$ as the run-index of the $j$-th completed run. We refer to \Cref{fig:meta-reduction-rewind} for how it is defined within the meta reduction. As an example, if the reduction starts $5$ runs of $\advA$, and complete these runs in the order $(2,3,1,5,4)$, we have that \ifnum\eprintversion=0 
\begin{align*}
&\runscompleted[1] = 2\;, \runscompleted[2] = 3\;, \runscompleted[3] = 1\;, \\
&\runscompleted[4] = 5\;, \runscompleted[5] = 4\;.    
\end{align*}
\else
\begin{align*}
	&\runscompleted[1] = 2\;, \runscompleted[2] = 3\;, \runscompleted[3] = 1\;, \runscompleted[4] = 5\;, \runscompleted[5] = 4\;.    
\end{align*}
\fi
We note that this unusual definition is used only to deal with the ability of the reduction to concurrently run many executions of $\advA$ (and thus further generalizing our result). If $\calR$ runs each execution of $\advA$ sequentially, $\runscompleted$ is simply an identity map. Moreover, we remark that the forgery containing $A^{(*, j)}$ will only be sent to the reduction if the $j$-th run is completed.

We remark that we cannot simply bound $\Pr[\BadRep]$ separately via another meta-reduction $\metaR'$. This is because it could be the case that $\BadRep$ occurs at some point after the event $\BadIdx$ already occurred; meaning such meta-reduction $\metaR'$ would not be able to compute a forgery on that run.
To this end and as an intermediate step, we will define a sequence of events $\BadIdx_j$ and $\BadRep_j$ for $j \in [r]$ which denotes events where $\BadIdx$ or $\BadRep$ are triggered in the $j$-th {\em completed} run (i.e., the $\runscompleted[j]$-th run). Then, we will write the event $\BadIdx \eOr \BadRep$ so that we can easily consider the two cases without the undesirable situation above. For the definition of  $\BadIdx_j$ and $\BadRep_j$, we will use $J = \runscompleted[j]$ for readability and define them as follows: \begin{itemize}[noitemsep]
	\item $\BadIdx_j:$ The random index $i^{(J)}$ leads to the meta-reduction aborting, i.e., the polynomial $(b^{(J)} - e_{i^{(J)}}^{(J)} a^{(J)})$ does not divide  $f^{(J)}$ or $g^{(J)}$, or \begin{align*}
		&\deg (b^{(J)} - e_{i^{(J)}}^{(J)}  a^{(J)})  < \max\{\deg b^{(J)}, \deg a^{(J)}\} \;.
	\end{align*}
	
	\item $\BadRep_j:$ There exists $i \in [\qnum]$ such that $b^{(J)} - e_i^{(J)} a^{(J)} \ndivides (f^{(J)} + m^{(J)}_i  g^{(J)}) a^{(J)}$.
\end{itemize}
Next, we can rewrite the event $\BadIdx \eOr \BadRep$ as follows \begin{align*}
	\ifnum\eprintversion=1
	\BadIdx \eOr \BadRep &= \bigcup_{j = 1}^r (\BadIdx_j \eOr \BadRep_j)\\
	\else
	&\BadIdx \eOr \BadRep = \bigcup_{j = 1}^r (\BadIdx_j \eOr \BadRep_j)\\
	\fi
	&= \bigcup_{j = 1}^r \left( \bigcap_{n = 1}^{j - 1} \left(\neg \BadIdx_n \eAnd \neg \BadRep_n \right) \eAnd (\BadIdx_j \eOr \BadRep_j) \right)\\
	&= \bigcup_{j = 1}^r \left( \bigcap_{n = 1}^{j - 1} \left( \neg \BadIdx_n \eAnd \neg \BadRep_n \right) \eAnd (\BadRep_j \eOr (\neg \BadRep_j \eAnd \BadIdx_j)) \right).
\end{align*}
The second and third equalities follow from the fact that for boolean values $x,y, x \eOr y = x \eOr (\neg x \eAnd y)$ (and inductively applying it). We additionally define two more events $\BadIdx'_j, \BadRep'_j$ such that \begin{align*}
	\BadRep'_j &:= \bigcap_{n = 1}^{j - 1} (\neg \BadIdx_n \eAnd \neg \BadRep_n) \eAnd  \BadRep_j \\
	\BadIdx'_j &:= \bigcap_{n = 1}^{j - 1} (\neg \BadIdx_n \eAnd \neg \BadRep_n) \eAnd  \neg \BadRep_j \eAnd \BadIdx_j\;.
\end{align*}
Intuitively, $\BadRep'_j$ is the event that the meta-reduction has not aborted in the first $j-1$ {\em completed} runs, and $\BadRep_j$ occurs. On the other hand, $\BadIdx'_j$ is the event that the meta-reduction has not aborted in the first $j-1$ completed runs and all the $e_i^{(\runscompleted[j])}$'s does not trigger the $\BadRep_j$ abort condition, but the sampled index $i^{(\runscompleted[j])}$ triggers $\BadIdx_j$. Now, we can write 
\ifnum\eprintversion=0
\begin{align*}
\BadIdx \eOr \BadRep &= \bigcup_{j = 1}^r (\BadRep'_j \eOr \BadIdx'_j) \\
&= \left( \bigcup_{j = 1}^r \BadRep'_j\right) \eOr \left( \bigcup_{j = 1}^r \BadIdx'_j\right) \;.
\end{align*}
\else
\begin{align*}
	\BadIdx \eOr \BadRep &= \bigcup_{j = 1}^r (\BadRep'_j \eOr \BadIdx'_j) = \left( \bigcup_{j = 1}^r \BadRep'_j\right) \eOr \left( \bigcup_{j = 1}^r \BadIdx'_j\right) \;.
\end{align*}
\fi
Thus, with the union bound, we will bound the following probability. \[
\Pr[\BadIdx \eOr \BadRep] \leq \Pr\left[ \bigcup_{j = 1}^r \BadRep'_j\right] + \Pr\left[  \bigcup_{j = 1}^r \BadIdx'_j \right]\;.
\]
Importantly, we stress that instead of bounding $\BadIdx$ and $\BadRep$ directly, we will bound the probability that these two events occur instead. This is because $\BadIdx$ does not necessarily guarantee that the given algebraic representation always leads to a well-defined polynomials, and $\BadRep$ does not necessarily guarantee that prior to the event triggering the meta-reduction can compute the forgeries for the prior runs. In contrast, with how  $\BadRep'_j$ and $\BadIdx'_j$ are defined, these conditions are guaranteed. The two following lemmas then conclude the proof. \qed


\begin{lemma}\label{lem:rewinding-game}
	$ \Pr\left[  \bigcup_{j = 1}^r \BadIdx'_j \right] \leq \displaystyle \frac{r^2 \cdot (k + 1)}{\qnum}$.
\end{lemma}

\begin{lemma}\label{lem:bad-rep-rewind}
	There exists a meta-reduction $\metaR'$ playing the $(d_1, d_2)$-$\DL$ game running in time roughly $t_\calR$ such that $\Pr\left[ \bigcup_{j = 1}^r \BadRep'_j\right] \leq \Adv^{(d_1, d_2)\text{-}\dl}_{\GroupGen}(\metaR', \secpar)$.
\end{lemma}
\begin{proof}[of \Cref{lem:bad-rep-rewind}]
	We first consider the event $\bigcup_{j = 1}^r \BadRep'_j$. 
	For simplicity and readability, we use the shorthand $J = \runscompleted[j]$ to denote the run specified in the event $\BadRep'_j$ throughout this proof.
	By the definition of $\BadRep'_j$, when the event is triggered, the meta-reduction $\metaR$ is able to forge on all prior completed runs as $\BadIdx_n$ does not occur for $n < j$. 
	Hence, in all the runs the forgery $A^{(*,J)}$ can be computed by the meta-reduction as a linear combination of $G_1, X_{1,1}, \ldots, X_{1,d_1}$. Thus, the representation of the group elements $\overline{G}_1^{(J)}, \overline{H}^{(J)}, \overline{G}_2^{(J)}, \overline{X}_{2,1}^{(J)}$ and $A^{(J)}_i$ defines polynomials $f^{(J)}, g^{(J)}, a^{(J)}, b^{(J)}$,  and $\alpha^{(i,J)}$, respectively. These polynomials are of degree at most $d_1$ or $d_2$ depending on which group it is in. 
	
	Therefore, we can construct another meta-reduction $\metaR'$ that simulates the adversary $\advA$ in the same manner as $\metaR$. However, when one of $\BadRep'_j$ occurs, $\metaR'$ will try to extract a DL solution. Let the $i$-th signing query of the $j$-th completed run (which corresponds to the $J$-th run) be the query which trigger the event. In particular, $(b^{(J)} - e_i^{(J)} a^{(J)}) \ndivides (f^{(J)} + m^{(J)}_i \cdot  g^{(J)})  a^{(J)}$, and the signature $(A_i^{(J)}, e_i^{(J)})$ verifies, so we have that \[\pairingOp{A_i^{(J)}}{\overline{X}_{2,1}^{(J)} - e_i^{(J)} \overline{G_2}^{(J)}} = \pairingOp{\overline{G}_1^{(J)} + m_i^{(J)} \overline{H}^{(J)}}{\overline{G_2}^{(J)}}\] Accordingly, we have \begin{align*}
	&\alpha^{(i,J)}(x) \cdot (b^{(J)}(x) - e_i^{(J)} a^{(J)}(x)) = (f^{(J)}(x) + m^{(J)} g^{(J)}(x)) \cdot a^{(J)}(x)\;.
	\end{align*} 
	Because of the indivisbility condition, the underlying polynomials on both sides are not identical, so $x$ is one of the zeros of the non-zero polynomial 
	\ifnum\eprintversion=0
	\begin{align*}
	h(\varX ) &= \alpha^{(i, J)}(\varX) (b^{(J)}(\varX) - e_i^{(J)} a^{(J)}(\varX)) \\
	&\pcind[4] - (f^{(J)}(\varX) + m^{(J)}_i \cdot g^{(J)}(\varX)) a^{(J)}(\varX)\;.
	\end{align*}
	\else
	\begin{align*}
		h(\varX ) &= \alpha^{(i, J)}(\varX) \cdot (b^{(J)}(\varX) - e_i^{(J)} a^{(J)}(\varX))  - (f^{(J)}(\varX) + m^{(J)}_i \cdot g^{(J)}(\varX)) \cdot a^{(J)}(\varX)\;.
	\end{align*}
	\fi Thus, $\metaR'$ simply factors $h$ and checks which zero is actually a $(d_1, d_2)$-DL solution. The advantage bound follows via the success of this reduction. \qed
\end{proof}

\subsection{Proof of \Cref{lem:rewinding-game}}\label{sec:rewinding-game}
We first observe that the probability of the event can be bounded by the winning probability of any adversary $\advC$ in the game $\Rewind$ (defined in \Cref{fig:rewinding-game}). The game models the behavior of the unbounded adversary $\advA$ and the adversary $\advB$ simulated by the meta-reduction into two phases: (1) Selecting signing queries ($\oraCh_1$) and (2) Selecting index $i^* \in [\qnum]$ to forge on ($\oraCh_2$). In particular, for any adversary $\advC$, its winning probability is \[
\Adv^\Rewind_{p,d_1, d_2,k,\qnum,r}(\advC, \secpar) := \Pr[\Rewind^\advC_{p, d_1, d_2, k, \qnum, r}(\secpar) = 1]\;.
\] 
Note that there exists an adversary $\advC$ such that \[
\Pr\left[\bigcup_{j = 1}^r \BadIdx'_j\right] \leq \Adv^\Rewind_{p,d_1, d_2,k,\qnum, r}(\advC, \secpar)\;.
\] This follows from how the unbounded adversary $\advA$ and the meta-reduction $\metaR$ reply to $\calR$, and that the event $\bigcup_{j = 1}^r \BadIdx'_j$ corresponds exactly to when $\win$ is set to $\true$. Note that the abort conditions in the challenge oracle $\oraCh_2$ exactly corresponds to how each event $\BadIdx'_j$ is defined. Throughout this section, we will not use the notation $\runscompleted[j]$ for readability and refer to the run indices as $j \in [r]$ instead.

\begin{figure}[!t]
	\centering
	\begin{pcvstack}[boxed]
		\begin{pchstack}
			\procedure{\scriptsize Game $\Rewind^\advC_{p, d_1, d_2, k, \qnum, r}(\secpar)$:}{
				Q_1, Q_2 \gets 0\; ; \; \win \gets \false\\
				\advC^{\oraCh_1, \oraCh_2}(1^\secpar)\\
				\pcreturn \win \eAnd (Q_1 \leq r) \eAnd (Q_2 \leq r)
			}
			\procedure{\scriptsize Oracle $\oraCh_1(f, g, a, b)$:}{
				\pcif \deg f > d_1 \eOr \deg g > d_1 \eOr \\
				\pcind[2] \deg a > d_2 \eOr \deg b > d_2 \eOr \\
				\pcind[2] f = 0 \eOr a = 0 \pcthen\pcabort\\
				Q_1 \gets Q_1 + 1\\
				\vec \mu^{(Q_1)} =  (\mu^{(Q_1)}_l)_{l \in [k]} \getsr \Zp^k \\
				(f^{(Q_1)}, g^{(Q_1)}, a^{(Q_1)}, b^{(Q_1)}) \gets (f, g, a, b)\\
				\vec{m}^{(Q_1)} \gets (\underbrace{\mu^{(Q_1)}_1, \ldots, \mu^{(Q_1)}_1}_{\qnum/k \text{ times}}, \ldots, \underbrace{\mu^{(Q_1)}_k, \ldots, \mu^{(Q_1)}_k}_{\qnum/k \text{ times}}) \\
				\pcreturn (\mu^{(Q_1)}_l)_{l \in [k]}
			}
		\end{pchstack}
		\procedure{\scriptsize Oracle $\oraCh_2(j \in [Q_1], (e_i)_{i \in [\qnum]})$}{
			\pcif (\exists i \in [\qnum]:( b^{(j)} - e_i a^{(j)}) \ndivides (f^{(j)} + m^{(j)}_i g^{(j)})a^{(j)}) \eOr (\exists i \neq j \in [\qnum] : e_i = e_j) \pcthen \pcabort\\
			Q_2 \gets Q_2 + 1; i^* \getsr [\qnum]\\
			\pcif (b^{(j)} - e_{i^*} a^{(j)}) \ndivides f^{(j)}a^{(j)} \eOr  (b^{(j)} - e_{i^*} a^{(j)}) \ndivides g^{(j)}a^{(j)} \eOr \\
			\pcind \deg (b^{(j)} - e_{i^*}^{(j)} a^{(j)}) < \max\{\deg a^{(j)}, \deg b^{(j)}\} \pcthen \\
			\pcind[2] \win \gets \true\\
			\pcreturn i^*
		}
	\end{pcvstack}
	\caption{Rewinding game.} 
	\label{fig:rewinding-game}
\end{figure}

We assume that $\advC$ makes at most $r$ queries to $\oraCh_1$ and $\oraCh_2$. Moreover, we assume without loss of generality that for any $\advC$, there exists an adversary $\advC'$ making $r$ queries to $\oraCh_1$ and ``only 1 query to $\oraCh_2$" such that \[
\Adv^\Rewind_{p,d_1, d_2,k,\qnum, r}(\advC, \secpar) \leq r \cdot \Adv^\Rewind_{p,d_1, d_2,k,\qnum, r}(\advC', \secpar)\;.
\] This follows simply from a guessing argument where $\advC'$ guesses which query to $\oraCh_2$ made by $\advC$ lead to $\win \gets \true$ being set. Accordingly, we can also assume that $\advC'$ makes all the $\oraCh_1$ queries before $\oraCh_2$, since the $\oraCh_1$ queries made after the only $\oraCh_2$ query does not influence the winning event.

Hence, we will analyze the winning probability of $\advC'$ making only $1$ query to $\oraCh_2$ after $r$ queries to $\oraCh_1$ instead.
In particular, we assume without loss of generality that $\advC'$ is {\em deterministic} and denote the sequence of $\oraCh_1$ outputs as $\vec{\mu} = (\vec \mu^{(1)}, \ldots, \vec \mu^{(r)})$ (as random variables), we define the corresponding polynomial $f,g,a,b$ in the $j$-th oracle query as a function $(\cdot)^{(j)}(\vec{\mu})$ of all the sampled messages $\vec{\mu} = (\vec{\mu}^{(j)} = (\mu^{(j)}_l)_{l \in [k]} \in \Zp^k)_{j \in [r]}$. Note that $(\cdot)^{(j)}$ {\em only depends} on the first $j - 1$ elements $\vec{\mu}^{(1)}, \ldots, \vec{\mu}^{(j - 1)}$ of $\vec{\mu}$. 

We define the functions $X_{j, l}(\vec{\mu})$ 
determined by the value of $\vec{\mu}$ for $j \in [r], l \in [k]$ as \[
X_{j,l} (\vec{\mu}) := \left|\left\{e \in \Zp : \begin{array}{l}
	b^{(j)}(\vec{\mu}) + e \cdot a^{(j)}(\vec{\mu}) \divides (f^{(j)}(\vec{\mu}) + \mu^{(j)}_l \cdot g^{(j)}(\vec{\mu})) a^{(j)}(\vec{\mu}) \; \eAnd \\
	(b^{(j)}(\vec{\mu}) + e \cdot a^{(j)}(\vec{\mu}) \ndivides f^{(j)}(\vec{\mu})a^{(j)}(\vec{\mu}) \; \eOr\\
	b^{(j)}(\vec{\mu}) + e  \cdot a^{(j)}(\vec{\mu}) \ndivides g^{(j)}(\vec{\mu})a^{(j)}(\vec{\mu}) ) 
\end{array} \right\}\right|
\] i.e., the number of bad $e \in \Zp$ where the divisibility condition is not satisfied, meaning our meta-reduction will not be able to simulate the forgery. 

Then, we have that if $\advC'$ chooses run-index $j \in [r]$ for the query to $\oraCh_2$, its winning probability given that $\vec{\mu} = \vec{\actualMu}$ is at most $\frac{\min\{\sum_{l = 1}^k X_{j, l}(\vec{\actualMu}), \qnum\} + 1}{\qnum}$. This is simply because of how the winning condition and $X_{j,l}$ are defined. Here, we take into account the particular event where the degree $\deg (b^{(j)} - e_{i^*}^{(j)} a^{(j)})$ decreases from the maximum of $\deg a^{(j)}$ and $\deg b^{(j)}$ with the $1/\qnum$ additive term. Note that this follows by a similar argument as in \Cref{lem:guessing-index}, because only one $e \in \Zp$ can cause that by cancelling out the leading coefficient.  Given that $\vec{\mu} = \vec{\actualMu}$, the optimal strategy for $\advC'$ is to choose the index $j \in [r]$ with the maximum chance. Therefore, the winning probability of $\advC'$ is at most $\frac{\min\{\max_{j \in [r]} \sum_{l = 1}^k X_{j,l}(\vec{\actualMu}), \qnum\} + 1}{\qnum}$. 

We can then say that \begin{align*}
	\Adv^\Rewind_{p,d_1, d_2,\qnum,r}(\advC', \secpar) &= \sum_{\vec{\actualMu}} \Pr_{\vec{\mu} \getsr (\Zp^k)^r}[\vec{\mu} = \vec{\actualMu}] \cdot \Pr\left[ \advC' \text{ wins}~\middle|~\vec{\mu} = \vec{\actualMu} \right]\\
	&\leq \frac{1}{\qnum} \sum_{\vec{\actualMu}} \Pr_{\vec{\mu} \getsr (\Zp^k)^r}[\vec{\mu} = \vec{\actualMu}] \left(\min\left\{\max_{j \in [r]} \sum_{l = 1}^k X_{j,l}(\vec{\actualMu}), \qnum \right\} + 1 \right)\\
	&= \frac{1}{\qnum}\Expect_{\vec{\mu}} \left[\min\left\{\max_{j \in [r]} \sum_{l = 1}^k X_{j,l}(\vec{\mu}), \qnum \right\}\right]+ \frac{1}{\qnum}\;.
\end{align*}
The first equality follows from $\advC'$ only depending on the outputs of the oracle $\oraCh$. The second inequality follows from the observation above. The last equality follows from definition of expectation over $\vec \mu$.

Therefore, our goal is to compute the expectation of the maximum. To do so, we prove the following claim lower bounding the probability that the random variable $Y(\vec{\mu}) = \min\left\{\max_{j \in [r]} \sum_{l = 1}^k X_{j,l}(\vec{\mu}), \qnum \right\}$ (where the randomness comes from $\vec{\mu}$) is upper bounded by some integer $t \geq 1$.

We now prove the following identity on the expectation of a non-negative discrete random variable. \begin{claim}
	For a discrete random variable $X \in [0, N]$, $\Expect[X] = \sum_{t = 1}^{N} \Pr[X \geq t]$.
\end{claim}
\begin{proof}[of Claim] Consider 
	\ifnum\eprintversion=0
	\begin{align*}
		\Expect[X] &=  \sum_{t = 1}^{N} t \cdot \Pr[X = t]\\
		&= \sum_{t = 1}^{\qnum} \sum_{k = 1}^t \Pr[X = t] \\
		&= \sum_{k = 1}^\qnum \sum_{t = k}^{\qnum} \Pr[X = t] 
		= \sum_{k = 1}^\qnum \Pr[X \geq k]\;.
	\end{align*}
	\else
	\begin{align*}
		\Expect[X] &=  \sum_{t = 1}^{N} t \cdot \Pr[X = t]= \sum_{t = 1}^{\qnum} \sum_{k = 1}^t \Pr[X = t]  = \sum_{k = 1}^\qnum \sum_{t = k}^{\qnum} \Pr[X = t] 
		= \sum_{k = 1}^\qnum \Pr[X \geq k]\;.
	\end{align*}
	\fi
	The second equality follows from expanding $t = \sum_{k = 1}^t 1$. The third equality follows from swapping the order of the sum (since $k$ starts from $1$ to $t$, $t$ starts from $k$ to $\qnum$). The last equality follows from $X$ only taking positive integer values and is bounded by $N$. \qed
\end{proof}

The claim then gives us $
	\Expect_{\vec{\mu}} \left[Y(\vec{\mu})\right] = \sum_{t = 1}^\qnum \Pr_{\vec{\mu}}[Y(\vec{\mu}) \geq t]$.
Next, we observe that for each $t \in [\qnum]$, 
\ifnum\eprintversion=0
\begin{align*}
	\Pr[Y(\vec{\mu}) \geq t] &= \Pr_{\vec{\mu} \getsr (\Zp^k)^r}\left[ \bigcup_{j = 1}^\qnum (\sum_{l = 1}^k X_{j,l}(\vec{\mu}) \geq t) \right] \\
	&\leq \sum_{j = 1}^r \Pr_{\vec{\mu} \getsr (\Zp^k)^r}\left[ \sum_{l = 1}^k X_{j,l}(\vec{\mu}) \geq t \right]\;.
\end{align*}
\else
\begin{align*}
	\Pr[Y(\vec{\mu}) \geq t] &= \Pr_{\vec{\mu} \getsr (\Zp^k)^r}\left[ \bigcup_{j = 1}^\qnum (\sum_{l = 1}^k X_{j,l}(\vec{\mu}) \geq t) \right] \leq \sum_{j = 1}^r \Pr_{\vec{\mu} \getsr (\Zp^k)^r}\left[ \sum_{l = 1}^k X_{j,l}(\vec{\mu}) \geq t \right]\;.
\end{align*}
\fi
The first equality follows by definition, and the second inequality follows from the union bound. Hence, 
\ifnum\eprintversion=0
\begin{align*}
	\Expect_{\vec{\mu}} \left[Y(\vec{\mu})\right] 
	&\leq  \sum_{t = 1}^\qnum \sum_{j = 1}^r \Pr_{\vec{\mu} \getsr (\Zp^k)^r}\left[ \sum_{l = 1}^k X_{j,l}(\vec{\mu}) \geq t \right]\\
	&= \sum_{j = 1}^r \sum_{t = 1}^\qnum \Pr_{\vec{\mu} \getsr (\Zp^k)^r}\left[ \sum_{l = 1}^k X_{j,l}(\vec{\mu}) \geq t \right] \\
	&= \sum_{j = 1}^r \Expect_{\vec{\mu}}\left[\sum_{l = 1}^k X_{j,l}(\vec{\mu})\right]\\
	&= \sum_{j = 1}^r \sum_{l = 1}^k \Expect_{\vec{\mu}}\left[ X_{j,l}(\vec{\mu})\right]
\end{align*}
\else
\begin{align*}
	\Expect_{\vec{\mu}} \left[Y(\vec{\mu})\right] 
	&\leq  \sum_{t = 1}^\qnum \sum_{j = 1}^r \Pr_{\vec{\mu} \getsr (\Zp^k)^r}\left[ \sum_{l = 1}^k X_{j,l}(\vec{\mu}) \geq t \right]
	= \sum_{j = 1}^r \sum_{t = 1}^\qnum \Pr_{\vec{\mu} \getsr (\Zp^k)^r}\left[ \sum_{l = 1}^k X_{j,l}(\vec{\mu}) \geq t \right] \\
	&= \sum_{j = 1}^r \Expect_{\vec{\mu}}\left[\sum_{l = 1}^k X_{j,l}(\vec{\mu})\right]
	= \sum_{j = 1}^r \sum_{l = 1}^k \Expect_{\vec{\mu}}\left[ X_{j,l}(\vec{\mu})\right]
\end{align*}
\fi
The second equality follows by swapping the summation order.
The second to last equality follows again from the above claim. The last equality follows from linearity of expectation.

Finally, we claim that $\Expect_{\vec{\mu}}[X_{j,l}(\vec{\mu})] \leq 1$ for all $j \in [r], l \in [k]$, which concludes the proof as \[
\Adv^\Rewind_{p,d_1, d_2,\qnum,r}(\advC', \secpar) \leq \frac{\Expect_{\vec{\mu}} \left[Y(\vec{\mu})\right] + 1}{\qnum} \leq \frac{r \cdot (k + 1)}{\qnum}\;.  \pcind \qed 
\] 

\begin{claim}
	For $j \in [r],l \in [k]$, $\Expect_{\vec{\mu}}[X_{j,l}(\vec{\mu})] \leq 1$.
\end{claim}
\begin{proof}
	Note that $f^{(j)}, g^{(j)}, a^{(j)}, b^{(j)}$ all depend only on the first $j - 1$ values in $\vec{\mu} = (\vec{\mu}^{(n)} = (\mu^{(n)}_l)_{l \in [k]})_{n \in [r]}$, given any $\vec{\actualMu}_{[1:j-1]} = (\vec{\mu}^{(n)})_{n \in [ j - 1]} \in (\Zp^k)^{j - 1}$ these polynomials are fixed. Then, we consider the conditioned expectation \begin{align*}
		\Expect_{\vec{\mu}^{(j)}}[X_{j,l}(\vec{\actualMu}_{[1:j-1]} \| \vec{\mu}_j)] &= \frac{1}{p^k} \sum_{\vec \actualMu^{(j)} \in \Zp^k} X_{j,l}(\vec{\actualMu}_{[1:j-1]} \| \vec \actualMu^{(j)}) \;\\
		&=  \frac{1}{p^{k - 1}} \sum_{\vec \actualMu^{(j)}_{[k] \setminus \{l\}} \in \Zp^{k - 1}} \frac{1}{p} \sum_{\actualMu^{(j)}_l \in \Zp} X_{j,l}(\vec{\actualMu}_{[1:j-1]} \| \vec \actualMu^{(j)}_{[k] \setminus \{l\}} \| \actualMu^{(j)}_l) \\
		&\leq \frac{1}{p^{k - 1}} \sum_{\vec \actualMu^{(j)}_{[k] \setminus \{l\}} \in \Zp^{k - 1}} 1 = 1
	\end{align*}
	The second to last inequality follows from the same observation as in the proof of \Cref{lem:guessing-index}: in particular, for any fixed $\vec{\actualMu}_{[1:j-1]} \in  (\Zp^k)^{j - 1}$ and $\vec\actualMu^{(j)}_{[k] \setminus \{l\}}$, the number of $e \in \Zp$ that is counted towards $X_{j,l}(\vec{\actualMu}_{[1:j-1]} \| \vec \actualMu^{(j)}_{[k] \setminus \{l\}} \| \allowbreak \actualMu^{(j)}_l)$ over all values of $\actualMu^{(j)}_l \in \Zp$ is at most $p$. 
	
	With this, we can the the expectation over $\vec{\mu}_{[1:j - 1]}$ and the claim follows. \qed
\end{proof}

%% file: figures/meta-reduction-rewind.tex
\begin{figure}[!t]
	\centering
	\begin{pcvstack}[boxed]
		\procedure{\scriptsize Algorithm $\advA^{\oraS}(\overline{G}_1, \overline H, \overline{G}_2, \overline{X}_{2,1})$:}{
			\pcif \overline{G}_1 = 0_{\GOne} \eOr \overline{G}_2 = 0_{\GTwo} \pcthen \pcreturn \bot\\
			F_1 \getsr \{f: \GOne^2 \times \GTwo^2 \to \Zp^k\} \pccomment{Random functions}\\
			F_2 \getsr \{f: \GOne^2 \times \GTwo^2 \times (\Zp^2 \times \GOne)^\qnum \to [\qnum]\} \\
			\protect\dgamebox{$(\mu_l)_{l \in [k]}$} \gets F_1(\overline{G}_1, \overline H, \overline{G}_2, \overline{X}_{2,1}) \; ; x \gets \dlog_{\overline{G}_2}(\overline{X}_{2,1})\\
			\protect\dgamebox{$\vec{m} \gets (\underbrace{\mu_1, \ldots, \mu_1}_{\qnum/k \text{ times}}, \ldots, \underbrace{\mu_k, \ldots, \mu_k}_{\qnum/k \text{ times}})$}\\
			\pcfor i \in [\qnum]: (\sigma_i = (A_i, e_i)) \getsr \oraS(m_i)\\
			\pcif (\exists i \in [\qnum ] : \textstyle \Ver{\BBS}((\overline{G}_1, \overline H, \overline{G}_2), \overline{X}_{2,1}, m_i, \sigma_i) = 0 ) \eOr \\
			\pcind (\exists i \neq i' \in [\qnum] : e_i = e_{i'} \protect\dgamebox{$\eOr \; \exists i \neq i' \in [k] : \mu_i = \mu_{i'}$}) \pcthen\\
			\pcind[2] \pcreturn \bot \\
			i^* \getsr F_2(\overline{G}_1, \overline H, \overline{G}_2, \overline{X}_{2,1}, \{(m_i, e_i, A_i)\}_{i \in [\qnum]})\; \\
			\pccomment{The message-signature tuples are sorted.}\\
			\protect\dgamebox{$m^* \gets \min \; (\Zp\setminus\{\mu_l\}_{l \in [k]})$}\\
			\textstyle \pcreturn \left(m^*, \left(A^* = \frac{\overline{G}_1 + m^* \overline H}{x - e_{i^*}}, e_{i^*}\right)\right)
		}
	\end{pcvstack}
	\caption{Unbounded adversary $\advA$. The integers $k$ and $\qnum$ are defined as in the lemma statement. The code in dashed boxes indicates the difference from the case $k = 1$.} \label{fig:unbounded-adv-rewind}
\end{figure}

\begin{figure}[!]
	\centering
	\begin{pcvstack}[center,boxed,space=2pt]
	\begin{pchstack}[space=2pt]
		\procedure{\scriptsize Algorithm $\metaR(\inp)$:}{
			(G_1, (X_{1,i})_{i \in [d_1]}, G_2, (X_{2,i})_{i \in [d_2]}) \gets \inp\\
			\text{Mapping } F_1 : \bits^* \to \Zp^k \; ; \; F_2 : \bits^* \to [\qnum] \\ 
			 \pccomment{For simulating oracles $\oraF_1, \oraF_2$}\\
			\text{Mapping } \runscompleted : [r] \to [r]\\
			\pccomment{Keep track of completed run ordering}\\
			\started, \closed \gets 0 \\
			\pccomment{Count number of started \& completed runs}\\
			(e', Z^*, \vec{\zeta}) \getsr \calR^{\hat{\advA}^{\oraF_1, \oraF_2}}(\inp)\\
			\pccomment{$\textstyle Z^* = \zeta_0 G_1 + \sum_{i = 1}^{d_1} \zeta_i X_{1,i} + \sum_{j = 1}^r \zeta_{d_1 + j} A^{(*, j)}$}\\
			\pcreturn x \text{ s.t. } xG_1 = X_{1,1} \eAnd \zeta(x)(x - e') = 1
		}
		\begin{pcvstack}
			\procedure{\scriptsize Oracle $\oraF_1(\str)$:}{
				\pccomment{$\str = (\overline{G}_1, \overline H, \overline{G}_2, \overline{X}_{2,1})$}\\
				\pcif F_1[\str] = \bot \pcthen \\
				\pcind F_1[\str] \getsr \Zp^k\\
				\pcreturn F_1[\str]
			}
			\procedure{\scriptsize Oracle $\oraF_2(\str)$:}{
				\pccomment{$\str = (\overline{G}_1, \overline H, \overline{G}_2, \overline{X}_{2,1}, \{(m_i, \sigma_i)\}_{i \in [\qnum]})$}\\
				\pcif F_2[\str] = \bot \pcthen \\
				\pcind F_2[\str] \getsr [\qnum]\\
				\pcreturn F_2[\str]
			}
		\end{pcvstack}
	\end{pchstack}
	\hrulefill\\
	\begin{pchstack}
		\procedure{\scriptsize Algorithm $(\advB^{\oraF_1, \oraF_2})^{\oraS}(\overline{G}_1^{(j)}, \overline H^{(j)}, \overline{G}_2^{(j)}, \overline{X}_{2,1}^{(j)}, \vec{f}^{(j)}, \vec{g}^{(j)}, \vec{a}^{(j)}, \vec{b}^{(j)})$:}{
		\pccomment{$\vec f^{(j)}, \vec g^{(j)}, \vec a^{(j)}, \vec b^{(j)}$ satisfy (\ref{eq:red-rew-input-1}--\ref{eq:red-rew-input-2}), and $j \in [r]$ denotes the run's numbering}\\
					\pcif \overline{G}_1 = 0_{\GOne} \eOr \overline{G}_2 = 0_{\GTwo} \pcthen \pcreturn \bot \pcind[16] \pccomment{$f^{(j)}, a^{(j)} \neq 0$}\\
		\started \gets \started + 1 \pccomment{Note that $\started = j$}\\
		\protect\dgamebox{$(\mu^{(j)}_l)_{l \in [k]}$} \gets \oraF_1(\overline{G}_1^{(j)}, \overline H^{(j)}, \overline{G}_2^{(j)}, \overline{X}_{2,1}^{(j)})\\
		\protect\dgamebox{$\vec{m}^{(j)} \gets (\underbrace{\mu_1^{(j)}, \ldots, \mu_1^{(j)}}_{\qnum/k \text{ times}}, \ldots, \underbrace{\mu_k^{(j)}, \ldots, \mu_k^{(j)}}_{\qnum/k \text{ times}})$}\\
		\pcfor i \in [\qnum]: (\sigma_i^{(j)} = (A_i^{(j)}, e_i^{(j)}), \vec \alpha^{(i,j)}) \getsr \oraS({m}^{(j)}_i)\\
		\pccomment{$\textstyle A_i^{(j)} = \alpha^{(i,j)}_0 G_1 + \sum_{n = 1}^{d_1} \alpha^{(i,j)}_n X_{1,i} + \sum_{n = 1}^{r} \alpha^{(i,j)}_{d_1 + n} A^{(*, n)}$}\\
		\pcif \textstyle (\exists i \in [\qnum]: \Ver{\BBS}(\overline{X}_{2,1}^{(j)}, m^{(j)}_i, \sigma_i^{(j)}) = 0)  \eOr (\exists i \neq i' \in [\qnum] : e_i^{(j)} = e_{i'}^{(j)}) \\
		\pcind \protect\dgamebox{$\eOr \; (\exists i \neq i' \in [k]: \mu_i^{(j)} = \mu_{i'}^{(j)})$ }\pcthen \pcreturn \bot\\
		i^{(j)} \gets \oraF_2(\overline{G}_1^{(j)}, \overline H^{(j)}, \overline{G}_2^{(j)}, \overline{X}_{2,1}^{(j)}, \{(m_i^{(j)}, e_i^{(j)}, A_i^{(j)})\}_{i \in [\qnum]})\; \\
		\protect\dgamebox{$\displaystyle m^{(*,j )} \gets \min (\Zp\setminus\{\mu^{(j)} _l\}_{l \in [k]}) $}\; \\
		\closed \gets \closed + 1\; ; \; \runscompleted[\closed] \gets j \\
		\pcind[14]\pccomment{Register the ordering of the completed runs  with $\runscompleted$.}\\
		\pcif \exists i \in [\qnum]: (b^{(j)} - e_i^{(j)} a^{(j)}) \ndivides (f^{(j)} + m^{(j)}_i g^{(j)}) a^{(j)}  \pcthen \metaR \; \pcabort \pcind[5] \pccomment{$\BadRep_\closed$}\\
		\pcif (b^{(j)} - e_{i^{(j)}}^{(j)} a^{(j)}) \divides f^{(j)}a^{(j)} \eAnd (b^{(j)} - e^{(j)}_{i^{(j)}} a^{(j)}) \divides g^{(j)}a^{(j)} \eAnd \\
		\pcind[2] \deg (b^{(j)} - e_{i^{(j)}}^{(j)} a^{(j)}) = \max\{\deg a^{(j)}, \deg b^{(j)}\}\pcthen\\
		\pcind[3] \textstyle \alpha^{(*,j)}(\varX) \gets \frac{f^{(j)} (\varX) + m^{(*,j)} g^{(j)}(\varX)}{b^{(j)} (\varX) - e_{i^{(j)}}^{(j)}  a^{(j)} (\varX)} a^{(j)} (\varX)\\
		\pcind[3] \textstyle \pcreturn (m^{(*,j )}, (A^{(*,j )} = \alpha^{(*,j)}(x)G_1, e_{i^{(j)} }^{(j)} ))\\
		\pcelse \metaR \; \pcabort \pcind[28] \pccomment{$\BadIdx_\closed$}
		}
	\end{pchstack}
	\end{pcvstack}
	\vspace*{-10pt}
	\caption{Meta-reduction $\metaR$ and simulated adversary $\advB$ which shares global states with $\metaR$. The integers $k, r$ and $\qnum$ are defined as in the lemma statement. As discussed in the body, all representation vectors can be parsed as a polynomial. (For vector $\vec \zeta$, we write $\zeta$ as the corresponding polynomial.) We distinguish between \ifnum\eprintversion=1 the meta-reduction\fi $\metaR$ and \ifnum\eprintversion=1 the simulated adversary\fi $\advB$ aborting with $\pcabort$ and $\pcreturn \bot$, respectively. The code in dashed boxes indicates the difference from the case $k = 1$.} 
	\label{fig:meta-reduction-rewind}
\end{figure}

%% file: tex_files/proof-any-length.tex
\section{Proof of \Cref{lemma:tight-proof-extend}}\label{sec:tight-proof-extend}

For simplicity, assume that all algorithms implicitly takes as input the group description $\pGroupDescript$.
Consider a $\SUF'$ adversary $\advA$, taking as input the group elements $G_1, \vec{H}, G_2, X_{2,1}$, making $\qnum$ signing queries on distinct messages $\vec{m}_1, \ldots, \vec{m}_\qnum \in \Zp^\ell$ and returning a forgery $(\vec{m}^*, \sigma^* = (A^*, e^*))$ such that $A^* = \frac{1}{x - e^*} (G_.1 + \sum_{j = 1}^\ell \vec{m}^*[j] \vec{H}[j])$. We consider the following events: \begin{itemize}
	\item $\Coll$: For some $i \neq i' \in [\qnum]$, $ \sum_{j = 1}^\ell \vec{m}_i[j] \vec{H}[j] =  \sum_{j = 1}^\ell \vec{m}_{i'}[j] \vec{H}[j]$, or for the forgery $(\vec{m}^*, \sigma^*)$, there exists an index $i \in [\qnum]$ such that $ \sum_{j = 1}^\ell \vec{m}_i[j] \vec{H}[j] =  \sum_{j = 1}^\ell \vec{m}^*[j] \vec{H}[j]$. 
	
	\item $\Forge_\ell$: $\Coll$ does not occur and $\advA$ wins in the game.
\end{itemize}
Then, it is easy to see that \[
\Adv^{\suf'}_{\BBS}(\advA, \secpar) \leq \Pr[\Coll] + \Pr[\Forge_\ell]\;.
\]
The two following claims then conclude the proof. \qed

\begin{claim}
	There exists an adversary $\advB_\Coll$ running in time roughly that of $\advA$ such that \[
	\Pr[\Coll] \leq \Adv^{\dlog}_{\GroupGen}(\advA, \secpar) + \frac{1}{p}\;.
	\]
\end{claim}
\begin{proof}
	Consider $\advB_\Coll$ defined as follows: \begin{itemize}
		\item Takes as input the group elements $G_1, Y \in \GOne, G_2$. Set $\vec{H} \gets \vec{\alpha} G_1 + \vec{\beta} Y, \allowbreak X_{2,1} \gets x G_2$ where $\vec \alpha, \vec \beta \getsr \Zp^\ell$ and $x \getsr \Zp$.
		\item Run $\advA$ with input $(G_1, \vec{H}, G_2, X_{2,1})$
		\item For each signing query $\vec{m}_i$, computes the signature $\sigma_i$ by sampling $e_i \getsr \Zp$ and $A_i \gets \frac{1}{x - e_i} (G_1 + \sum_{j = 1}^\ell \vec{m}_i[j]\vec{H}[j])$. 
		\item After the forgery $(\vec{m^*}, \sigma^*)$ is given, check if $\Coll$ occurs. If so, let $\vec{m} \neq \vec{m}'$ be the two messages triggering the event. Then, the reduction returns $-\frac{\sum_{j = 1}^\ell (\vec{m}[j] - \vec{m}'[j]) \alpha[j]}{\sum_{j = 1}^\ell (\vec{m}[j] - \vec{m}'[j]) \beta[j]}$ if $\sum_{j = 1}^\ell (\vec{m}[j] - \vec{m}'[j]) \beta[j] \neq 0$. Otherwise, abort.
	\end{itemize}
	It is clear that the running time is roughly $t_\advA$.
	The correctness of the reduction follows from the event $\Coll$ implying that \[
	0_{\GOne} =  \sum_{j = 1}^\ell (\vec{m}[j] - \vec{m}'[j]) \vec{H}[j] =  \sum_{j = 1}^\ell (\vec{m}[j] - \vec{m}'[j]) \vec{\alpha}[j]  G + \sum_{j = 1}^\ell (\vec{m}[j] - \vec{m}'[j]) \vec{\beta}[j]  Y\;.
	\] Also, note that since $\vec{\beta}$ is information theoretically hidden from the view of $\advA$, the abort from $\sum_{j = 1}^\ell (\vec{m}[j] - \vec{m}'[j]) \beta[j] = 0$ can only occur with probability at most $1/p$. Hence, this concludes the proof. \qed
\end{proof}

\begin{claim}
	There exists an adversary $\advB'$ against SUF' game of $\BBS_1 = \BBS[\GroupGen, 1]$ running in time roughly that of $\advA$ and making $\qnum$ distinct signing queries such that \[
		\Pr[\Forge_\ell]\leq \Adv^{\suf'}_{\BBS_1}(\advB', \secpar)\;.
	\]
\end{claim}
\begin{proof}
	Consider $\advB'$ playing the SUF' game of $\BBS_1$ defined as follows: \begin{itemize}
		\item Takes as input the group elements $G_1, H, G_2, X_{2,1}$. If $H = 0_{\GOne}$, make a query on message $m = 0$ to receive $\sigma = (A, e)$ and return the forgery $(1, \sigma)$. Otherwise, compute $\vec{H} \gets \vec{\alpha} H$ where $\vec \alpha \in \Zp^\ell$.
		\item Run $\advA$ with input $(G_1, \vec{H}, G_2, X_{2,1})$.
		\item For each signing query $\vec{m}_i$, compute $m_i' \gets \sum_{j = 1}^\ell \vec{m}_i[j] \alpha[j]$, make a signing query to its oracle on $m'_i$ to receive $\sigma_i = (A_i, e_i)$, and return $\sigma_i$ to $\advA$.
		\item For the forgery $(\vec{m}^*, \sigma^*)$, compute $\mu^* =  \sum_{j = 1}^\ell \vec{m}^*[j] \alpha[j]$ and return the forgery $(\mu^*, \sigma^*)$. 
	\end{itemize}
	The running time of $\advB'$ is roughly that of $\advA$. 
	Note that the view of $\advA$ remains the same as $\vec{H}$ is uniform in $\GOne^\ell$ and the signatures from the oracles are valid with $e_i$ uniformly random. Moreover, if $H = 0_{\GOne}$, we can see that the reduction trivially wins in the game. 
	
	Now, we consider when $\advA$ wins the game and $\Coll$ does not occur. Since $\Coll$ does not occur and by how we set up $\vec{H}$, we have that \begin{enumerate}[label=(\alph*)]
		\item The signing queries $m_i'$ made by $\advB_\suf$ are distinct. 
		\item The message $\mu^*$ in the forgery is distinct from all the signing queries.
	\end{enumerate}
	Therefore, if $\advA$ wins in the game, $\sigma^* = (A^*, e^*)$ is such that $e^* \in \{e_i\}_{i \in [\qnum]}$ and 
	\ifnum\eprintversion=0
	\begin{align*}
	A^* &= \frac{1}{x - e^*} \left(G_1 + \sum_{j = 1}^\ell \vec{m}^*[j]\vec{H}[j]\right) \\
	&=  \frac{1}{x - e^*} \left(G_1 + \sum_{j = 1}^\ell \vec{m}^*[j] \vec{\alpha}[j] H\right) \\
	&= \frac{G_1 + \mu^* H}{x - e^*}\;.
	\end{align*} \else
	\begin{align*}
		A^* &= \frac{1}{x - e^*} \left(G_1 + \sum_{j = 1}^\ell \vec{m}^*[j]\vec{H}[j]\right) 
		=  \frac{1}{x - e^*} \left(G_1 + \sum_{j = 1}^\ell \vec{m}^*[j] \vec{\alpha}[j] H\right) 
		= \frac{G_1 + \mu^* H}{x - e^*}\;.
	\end{align*}
	\fi Hence, $\advB_\suf$ also wins in the SUF' game, concluding the proof. \qed
\end{proof}